\documentclass[a4paper,twocolumn,10pt,accepted=2024-10-04]{quantumarticle}
\pdfoutput=1
\usepackage{amsmath,amssymb,amsfonts,mathtools,dsfont,relsize}
\usepackage{graphicx}
\usepackage{microtype}

\allowdisplaybreaks

\usepackage[normalem]{ulem} % just for strikeout
\usepackage[dvipsnames]{xcolor}
\usepackage{comment}

\usepackage[colorlinks=true, urlcolor=violet, linkcolor=blue, citecolor=red, hyperindex=true, linktocpage=true]{hyperref}

\usepackage{amsthm}
\newtheorem{thm}{Theorem}
\newtheorem{cor}[thm]{Corollary}
\newtheorem{lem}[thm]{Lemma}
\newtheorem{prop}[thm]{Proposition}
\newtheorem{defn}{Definition}

\usepackage{soul} % what is this?
\usepackage{tikz}
\usetikzlibrary{calc}
\usetikzlibrary{positioning}
\usetikzlibrary{arrows,shapes}

\usepackage{enumitem}
\usepackage{suffix} % for Ollie's \matel*{}{}{} command

\renewcommand{\thesection}{\arabic{section}}
\renewcommand{\thesubsection}{\thesection.\arabic{subsection}}
\renewcommand{\thesubsubsection}{\thesubsection.\arabic{subsubsection}}
\makeatletter
\renewcommand{\p@subsection}{}
\def\l@subsubsection#1#2{}
\makeatother

\makeatletter
\g@addto@macro\bfseries{\boldmath}
\makeatother

\newcommand{\vpd}[0]{\vphantom{\dagger}}
\newcommand{\vps}[0]{\vphantom{*}}
\newcommand{\vpp}[0]{\vphantom{\prime}}

\DeclarePairedDelimiter{\floor}{\lfloor}{\rfloor}
\DeclarePairedDelimiter{\norm}{\lVert}{\rVert}
\DeclarePairedDelimiter{\abs}{\lvert}{\rvert}

\DeclarePairedDelimiter{\ket}{\lvert}{\rangle}
\DeclarePairedDelimiter{\bra}{\langle}{\rvert}
\DeclareMathOperator*{\Span}{span}
\DeclareMathOperator*{\Endomorphism}{End}
\DeclareMathOperator*{\automorphism}{Aut}

\newcommand{\ii}[0]{\mathrm{i}}
\newcommand{\bvec}[1]{\boldsymbol{#1}}
\newcommand{\kron}[1]{\delta^{\,}_{#1}}
\newcommand{\ident}[0]{\mathds{1}}

\newcommand{\Comps}{\mathbb{C}}
\newcommand{\Ints}{\mathbb{Z}}

\newcommand{\Unitary}[1]{\mathsf{U} \hspace{-0.1em} \left( #1 \right)}\WithSuffix\newcommand\Unitary*[1]{\mathsf{U} \hspace{-0.1em} (#1)}
\newcommand{\SU}[1]{\mathsf{SU} \left( #1 \right)}\WithSuffix\newcommand\SU*[1]{\mathsf{SU} \hspace{-0.1em} (#1)}

\newcommand{\Order}[1]{\mathsf{O} \left( #1 \right)}

\DeclareMathOperator*{\Dim}{dim}
\newcommand{\DimOf}[1]{\Dim \left( #1 \right)}\WithSuffix\newcommand\DimOf*[1]{\Dim ( #1 )}

\newcommand{\BKop}[2]{\ket{#1} \hspace{-0.4mm} \bra{#2}}

\newcommand{\matel}[3]{\left\langle #1 \middle| #2 \middle| #3 \right\rangle}
\WithSuffix\newcommand\matel*[3]{\langle #1 | #2 | #3 \rangle}

\newcommand{\tr}[1]{{\rm tr} \left[  #1  \right]}
\DeclareMathOperator*{\trace}{tr}
\newcommand{\ptr}[2]{\trace^{\,}_{#1} \left[ #2 \right]}

\newcommand{\SpanOf}[1]{\Span \left\{ #1 \right\}}\WithSuffix\newcommand\SpanOf*[1]{\Span \{ #1 \} }
\newcommand{\End}[1]{\Endomorphism \left( #1 \right)}\WithSuffix\newcommand\End*[1]{\Endomorphism ( #1 )}
\newcommand{\Aut}[1]{\automorphism \left( #1 \right)}\WithSuffix\newcommand\Aut*[1]{\automorphism ( #1 )}

\newcommand{\comm}[2]{\left[ #1, \, #2 \right]}
\newcommand{\acomm}[2]{\left\{ #1, \, #2 \right\} }

\newcommand{\Hilbert}[0]{\mathcal{H}}
\newcommand{\HilDim}[0]{\mathcal{D}}
\newcommand{\Nspins}[0]{N}
\newcommand{\Nmeas}[0]{M}
\newcommand{\Nregions}[0]{m}
\newcommand{\Noutcome}[0]{\mathcal{N}}
\newcommand{\Dist}[0]{L}
\newcommand{\LRvel}[0]{v}
\newcommand{\interval}[0]{\mathcal{I}}

\newcommand{\numQ}[0]{k}
\newcommand{\observ}[0]{\mathcal{O}}
\newcommand{\mobserv}[0]{A}

\newcommand{\mObserv}[2]{\mobserv^{#2}_{#1}}

\def\DensMat{\rho}

\newcommand{\meig}[1]{a^{\,}_{#1}}
\newcommand{\mEig}[1]{a^{\vpp}_{#1}}

\DeclareMathOperator*{\Stab}{Stab}
\DeclareMathOperator*{\UStab}{stab}
\newcommand{\StabEl}[0]{\mathsf{S}}
\newcommand{\StabOf}[1]{\Stab \left( #1 \right)}\WithSuffix\newcommand\StabOf*[1]{\Stab ( #1 )}
\newcommand{\UStabOf}[1]{\UStab \left( #1 \right)}\WithSuffix\newcommand\UStabOf*[1]{\UStab ( #1 )}
\newcommand{\stabilizer}[0]{\mathcal{S}}

\newcommand{\chan}{\mathcal{W}}
\newcommand{\chandag}{\mathcal{W}^{\dagger}}

\newcommand{\MeasChannel}[0]{\mathcal{M}}

\newcommand{\QECChannel}[0]{\mathcal{R} }

\newcommand{\cliff}[0]{U}

\newcommand{\Cliff}[1]{\cliff^{\vpd}_{#1}}
\newcommand{\CliffDag}[1]{\cliff^{\dagger}_{#1}}
\newcommand{\uchan}[0]{\mathcal{U}}

\newcommand{\conjchan}[0]{\chi}

\newcommand{\ConjChan}[1]{\conjchan^{\vpd}_{#1}}

\newcommand{\ConjChanPow}[2]{\conjchan^{#2}_{#1}}

\newcommand{\isometry}[0]{\mathbb{V}}
\newcommand{\umeas}[0]{\mathcal{M}} %for general mentions
\newcommand{\Umeas}[1]{\umeas^{\vpd}_{\left[ #1 \right]}}
\newcommand{\Umeasdag}[1]{\umeas^{\dagger}_{\left[ #1 \right]}}

\newcommand{\Proj}[2]{P^{(#1)}_{#2}}

\newcommand{\Pauli}[2]{\sigma^{#1}_{#2}}
\newcommand{\PauliGroup}[1]{\mathbf{P}^{\vps}_{#1}}\WithSuffix\newcommand\PauliGroup*[1]{\mathbf{P}^{\,}_{#1}}
\newcommand{\CliffGroup}[1]{\mathbf{C}^{\vps}_{#1}}\WithSuffix\newcommand\CliffGroup*[1]{\mathbf{C}^{\,}_{#1}}

\newcommand{\RecGroup}[1]{\mathbf{R}^{\vps}_{#1}}\WithSuffix\newcommand\RecGroup*[1]{\mathbf{R}^{\,}_{#1}}

\newcommand{\PX}[1]{X^{\vps}_{#1}}
\newcommand{\PXp}[2]{X^{#2}_{#1}}
\newcommand{\PY}[1]{Y^{\vps}_{#1}}

\newcommand{\PZ}[1]{Z^{\vps}_{#1}}
\newcommand{\PZp}[2]{Z^{#2}_{#1}}

\newcommand{\LX}[0]{X^{\vps}_{{\rm L}}}
\newcommand{\LXp}[1]{X^{#1}_{{\rm L}}}
\newcommand{\LY}[0]{Y^{\vps}_{{\rm L}}}
\newcommand{\LYp}[1]{Y^{#1}_{{\rm L}}}
\newcommand{\LZ}[0]{Z^{\vps}_{{\rm L}}}
\newcommand{\LZp}[1]{Z^{#1}_{{\rm L}}}
\newcommand{\SSid}[1]{\widetilde{\ident}^{\vps}_{#1}}
\newcommand{\SSX}[1]{\widetilde{X}^{\vps}_{#1}}
\newcommand{\SSXp}[2]{\widetilde{X}^{#2}_{#1}}

\newcommand{\SSZ}[1]{\widetilde{Z}^{\vps}_{#1}}
\newcommand{\SSZp}[2]{\widetilde{Z}^{#2}_{#1}}

\newcommand{\SSProj}[2]{\widetilde{P}^{(#1)}_{#2}}

\usepackage{qcircuit}

\usepackage[numbers,sort&compress]{natbib}

\begin{document}

\title{Quantum teleportation implies symmetry-protected \quad\quad$\vps$ topological order}

\author{Yifan Hong}
\email{yifan.hong@colorado.edu}
\affiliation{Department of Physics and Center for Theory of Quantum Matter, University of Colorado, Boulder, CO 80309, USA}

\author{David T. Stephen}
\affiliation{Department of Physics and Center for Theory of Quantum Matter, University of Colorado, Boulder, CO 80309, USA}
\affiliation{Department of Physics and Institute for Quantum Information and Matter, California Institute of Technology, Pasadena, California 91125, USA}

\author{Aaron J. Friedman}
\affiliation{Department of Physics and Center for Theory of Quantum Matter, University of Colorado, Boulder, CO 80309, USA}

\begin{abstract}
    We constrain a broad class of teleportation protocols using insights from locality. In the ``standard'' teleportation protocols we consider, all  outcome-dependent unitaries are Pauli operators conditioned on linear functions of the measurement outcomes. We find that all such protocols involve preparing a ``resource state'' exhibiting symmetry-protected topological (SPT) order with Abelian protecting symmetry $\mathcal{G}_{\numQ}= (\Ints^{\,}_2 \times \Ints^{\,}_2)^\numQ$. The $\numQ$ logical states are teleported between the edges of the chain by measuring the corresponding $2\numQ$ string order parameters in the bulk and applying outcome-dependent Paulis. Hence, this single class of nontrivial SPT states is both necessary and sufficient for the standard teleportation of $\numQ$ qubits. We illustrate this result with several examples, including the cluster state, variants thereof, and a nonstabilizer hypergraph state.
\end{abstract}

\maketitle

\section{Introduction}

\subsection{Background}

Entangled states of many-body quantum systems are valuable resources for a variety of tasks in quantum information processing. A prototypical example is quantum teleportation: By measuring an entangled state and applying unitaries conditioned on the outcomes (i.e., feedback), quantum information can be transferred over larger distances than can be achieved using unitary time evolution alone \cite{Bennett_1993}. Teleportation protocols can also be concatenated to form a \emph{quantum repeater} that transfers quantum information over arbitrary distances \cite{Briegel_1998}.

It is then natural to ask what types of entangled many-body \emph{resource states} can be used to teleport quantum information by measuring from one end of a 1D chain to the other and applying outcome-dependent feedback, as illustrated in Fig.~\ref{fig:intro}. In fact, such entangled resource states form the backbone of many notions of quantum computation, and especially \emph{measurement-based} quantum computation (MBQC), in which logical operations are applied to a quantum state as it is teleported \cite{CliffordHierarch, Raussendorf2001}. These same states are also closely related to the concept of localizable entanglement \cite{LocalizableEntanglement}, which has important connections to the ground-state properties and phase transitions of quantum spin chains \cite{Verstraete2004, Verstraete2004a, Jin2004, Pachos2004, Venuti2005, Skrovseth2009}. Moreover, the effect of imperfect teleportation on many-body ground states has recently been studied and implemented on quantum devices~\cite{lee2022, Zhu_2023, eckstein2024robust, Sala2024}.
Accordingly, significant effort has been invested into characterizing entangled quantum states by their utility as resource states for tasks such as MBQC and quantum teleportation \cite{Raussendorf2005a, Vandennest2006, Gross2007, Gross2009, Bremner2009, Doherty2009, Gross2010, DominicMBQC_SPT, Prakash2015, hypergraphMiyake2D, Stephen2017, Marvian2017, Wei2018a, Bao_teleport, google_teleport}. 

Here we consider a different question. Rather than consider particular quantum states and diagnose their suitability for various quantum tasks, we instead seek to classify and constrain teleportation protocols directly, in analogy to phases of matter. The idea is to use this alternative, top-down approach to identify classes of entangled resource states compatible with teleportation. To do so, we restrict our consideration to a particular class of ``standard'' teleportation protocols, which include the most straightforwardly implementable protocols known to the literature. From the conditions imposed by standard teleportation, we then investigate the implications for compatible resource states.

\begin{figure}[t]
\centering
\includegraphics[width=0.42\textwidth]{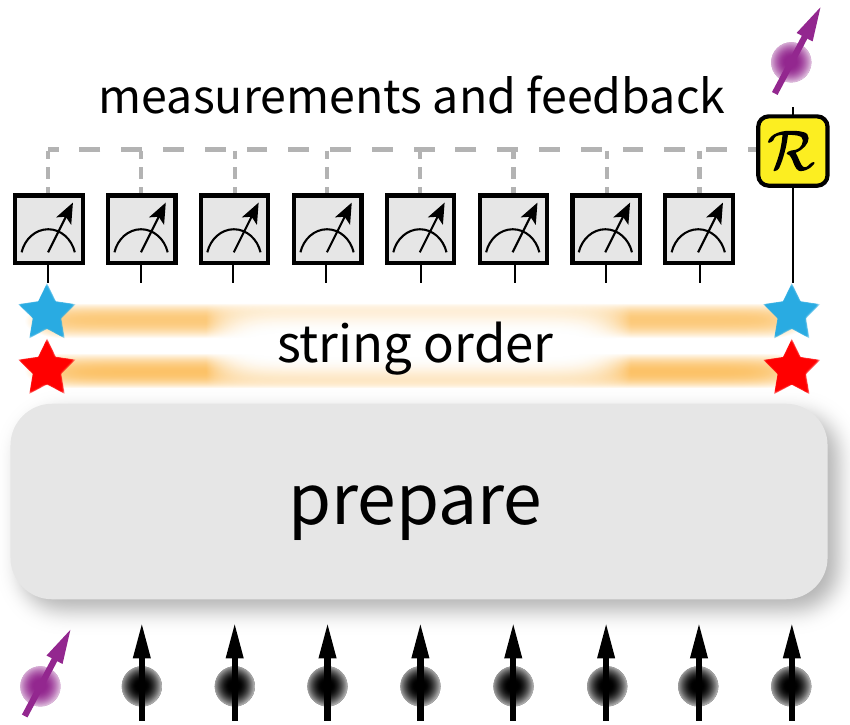}
\caption{Illustration of the central result of this paper: All standard teleportation protocols amount to preparing an SPT state with symmetry $\mathcal{G} = \Ints^{\,}_2 \times \Ints^{\,}_2$ with a logical state $\ket{\psi}$ encoded (in purple), measuring the corresponding pair of string order parameters, and applying outcome-dependent Paulis. Teleporting $\numQ \geq 1$ qubits requires measuring $2\numQ$ strings.}
\label{fig:intro}
\end{figure}

Of particular interest is the class of many-body states in 1D with symmetry-protected topological (SPT) order \cite{Chen2011, MPS1dEntSpect, MPS_classification, Chen2013}, which are known to be useful resource states for quantum teleportation and MBQC \cite{Doherty2009, Gross2010, Miyake2010, hypergraphMiyake2D, DominicMBQC_SPT,  Prakash2015, Stephen2017, Marvian2017, Wei2018a}. These states are short-range entangled, yet cannot be connected to product states via finite-depth unitaries that respect the ``protecting'' symmetry group $\mathcal{G}$. For an Abelian symmetry $\mathcal{G}$, the corresponding SPT states exhibit perfect long-range order characterized by \emph{nonlocal} stringlike objects. These \emph{string order parameters} act as commuting symmetry operators in the interior of any finite interval, but are decorated by endpoint operators at the boundaries \cite{denNijs1989, Kennedy1992, MPS1dDetect}. Moreover, measuring the symmetry charge in any such bulk region reduces the string order parameter to a long-range two-point correlation function of the two endpoint operators, leaving behind a pattern of long-range entanglement that can be used for teleportation \cite{Marvian2017}. Remarkably, a state's suitability for teleportation is a property of entire SPT \emph{phases}, and not merely special points within the phase \cite{Doherty2009, DominicMBQC_SPT, Miyake2010, Raussendorf2022}, establishing a compelling link between quantum information processing and quantum phases of matter.

Another source of insight into quantum information processing is \emph{locality}, which imposes bounds on the speed with which quantum information can be transferred. The Lieb-Robinson Theorem \cite{Lieb1972} establishes an emergent speed limit in the context of unitary time evolution under generic local Hamiltonians. Additionally, more general quantum channels captured by local Lindbladians obey the same bound \cite{Poulin}. However, when local projective measurements are combined with \emph{nonlocal} feedback (via instantaneous classical signalling), the standard Lieb-Robinson bound can be overcome. Nevertheless, a similar bound on generic useful quantum tasks has recently been derived even in the presence of nonlocal feedback \cite{SpeedLimit}. These results impose optimal quantitative constraints on the distance over which logical qubits can be teleported.

In this paper, we use similar techniques to derive more \emph{qualitative} constraints on teleportation protocols and their associated resource states. Importantly, this requires only a few physically motivated assumptions---namely that the initial state is short-range entangled and that the outcome-dependent corrections have some particular form. In particular we consider ``standard'' teleportation protocols, whose correction operations are Pauli gates conditioned on parities of measurement outcomes. We then find that \emph{all} standard teleportation protocols can be written in a particular ``canonical form'' that corresponds to preparing a resource state with $\Ints_2^{2\numQ}$ SPT order, measuring its string order parameters, and applying Pauli gates conditioned on the outcomes, as depicted in Fig.~\ref{fig:intro}. Because such SPT states were known to be viable resource states for teleportation, our main result---captured by Theorem~\ref{thm:spt order}---implies an equivalence between standard teleportation and resource states with $\Ints_2^{2\numQ}$ SPT order.

\subsection{Summary of assumptions and results}
% Organization of the paper

The main result of this work is that a standard class of  teleportation protocols---which includes most, if not all, teleportation protocols known to the literature---is equivalent to preparing a state with a particular SPT order, measuring its string order parameters, and applying outcome-dependent Pauli gates to the destination qubits. We summarize our assumptions and main results below.

\textbf{Summary of assumptions.} We consider ``standard'' teleportation protocols $\chan$ with the following properties:
\begin{itemize}
    \item The $\numQ$ logical states $\ket{\psi_n}$ are localized to single sites before and after the protocol $\chan$.
    \item Each qubit is teleported a distance greater than that which can be achieved using unitaries alone.
    \item The unitary part of $\chan$ is local in the sense of obeying a Lieb-Robinson bound~\cite{Lieb1972}.  
    \item The nonunitary part of $\chan$ is limited to projective measurements and outcome-dependent ``recovery'' operations. We note that other operations cannot enhance teleportation~\cite{SpeedLimit}. 
    \item All measured observables $A_j$ are unitarily connected to single-qubit observables%, and thus have precisely two unique eigenvalues (outcomes)
    . 
    \item No unitaries act on previously measured sites. 
    \item Finally, the most important assumptions are that: (\emph{i}) the recovery operations are linear functions of the measurement outcomes and (\emph{ii}) act on the physical qubits as Pauli operators or strings thereof.
\end{itemize}

These criteria are also summarized in Def.~\ref{def:standard teleport} of standard teleportation. The first five assumptions simply define the precise notion of teleportation we consider. The penultimate, technical criterion is imposed to simplify proofs, but it can almost certainly be relaxed. Most importantly, the final criterion is our main conceptual assumption, and we find that it strongly constrains the class of teleportation protocols and corresponding resource states. We expect that relaxing this final assumption in various ways will lead to different classes of SPT resource states beyond what we explore herein.

For convenience, we write any such protocol $\chan$ in the ``canonical form'' $\chan = \QECChannel \,  \MeasChannel \, \uchan$~\eqref{eq:chan canonical form}, where $\uchan$ acts unitarily (see Sec.~\ref{subsec:canonical}). We view $\uchan$ as preparing a ``resource'' state~\eqref{eq:resource state}, which is then subjected to the nonunitary measurement and recovery channels (respectively, $\MeasChannel$ and $\QECChannel$) to achieve standard teleportation. 

\textbf{Summary of results.} Our main results for such standard teleportation protocols are as follows:
\begin{itemize}
    % \item At least two measurements are required to achieve teleportation. See also Prop.~\ref{prop:M=1 useless}.
    %
    % \item The distance $L=d(i,f)$~\eqref{eq:L def} over which a single qubit can be teleported using $\Nmeas$ measurements is bounded by Eq.~\ref{eq:k=1 speed limit}. See also Lemma~\ref{lem:M>2 speedup}.
    %
    \item The distance $L$~\eqref{eq:L def} over which $\numQ$ qubits can be teleported using $\Nmeas$ measurements is bounded by
    \begin{equation*}
        \tag{\ref{eq:k>1 bound}}
        \Dist \leq \LRvel T + 2 \, \LRvel \left\lfloor \frac{\Nmeas}{2\numQ} \right\rfloor \left( T - 1 \right) \, ,~
    \end{equation*}
    with the additional requirement $T \geq 1 + \numQ/v$~\eqref{eq:k>1 min depth}, as established in Theorem~\ref{thm:standard bound}.
    \item All measured observables in the canonical-form protocol must commute with one another and do not act on the final logical sites, as proven in Cor.~\ref{cor:meas comm}.
    \item If $\chan$ teleports $\numQ$ qubits as described above acting on an initial state $\ket{\Psi_0}$, then the \emph{resource state} $\ket{\Psi_t} = \uchan \ket{\Psi_0}$~\eqref{eq:resource state} has SPT order protected by a $(\Ints_2 \times \Ints_2)^{\numQ}$ symmetry. For $\numQ=1$, e.g., all such states are unitarily equivalent to the cluster state~\cite{Briegel_2001}. This result appears in Theorem~\ref{thm:spt order}.
    \item The resource state $\ket{\Psi_t}$~\eqref{eq:resource state} has perfect string order, and thus lies at the fixed point of the $(\Ints_2 \times \Ints_2)^{\numQ}$ SPT phase~\cite{MPSStringOrder}. Each of the $2 \numQ$ string order parameters is the product of all measured observables used to determine the $X$-type (or $Z$-type) recovery operation for each of the $\numQ$ logical qubits. 
    \item The resource state need not be a stabilizer state, as with the hypergraph state we consider in Sec.~\ref{subsec:hypergraph}.
\end{itemize}

Importantly, while it was previously known that generic 1D SPT states could be used as resources for quantum teleportation, our results conversely show that \emph{all} standard teleportation protocols imply a \emph{single} class of SPT resource state. This establishes an \emph{equivalence} between (standard) teleportation and SPT states, a direct connection between a particular quantum phase of matter and a useful quantum tasks, and lends insight into the physical mechanism underlying quantum teleportation.

\subsection{Organization of the paper}

We begin with a generic overview in Sec.~\ref{sec:overview}, in which we assume only that the system of interest contains $\Nspins$ qubits and that all unitary operations act locally. The statements and proofs therein may also be found in the literature, and hold beyond the scope of this work. For example, we specify the intuitive notion of ``physical'' teleportation protocols in Def.~\ref{def:phys teleport}, and briefly explain the bounds~\cite{Lieb1972, SpeedLimit} that constrain them in Sec.~\ref{subsec:bounds}.

The main assumptions appear in Sec.~\ref{sec:restricted teleportation}. We formally define the ``standard'' teleportation of $\numQ \geq 1$ logical qubits in Def.~\ref{def:standard teleport}. We require that these protocols achieve a teleportation distance $\Dist$ \eqref{eq:L def} that cannot be achieved using unitary evolution alone, to distinguish them from quantum state transfer (see Def.~\ref{def:state transfer}). We also restrict all measurements to single-qubit observables, as is standard in  quantum information processing \cite{QC_book}. We  assume that no physical unitaries are applied to previously measured qubits; such unitaries have no direct utility to teleportation, but needlessly complicate the statement of our main results. Our crucial assumption---which we expect is directly related to our main results---is that all quantum operations conditioned on the outcomes of prior measurements correspond to physical Pauli gates, and all classical side processing involves only linear functions of the measurement outcomes. We expect the extension of this analysis beyond these restrictions to lead to new classes of teleportation protocols with analogous results; we defer such an investigation to future work.

In Sec.~\ref{sec:teleport+locality}, we constrain standard teleportation protocols in several ways. We first show that at least two measurements---along with outcome-dependent feedback---are required to beat purely unitary protocols. We then show that two measurements are required per ``measurement region,'' with each measurement region no larger than $2 \LRvel (T-1)$, where $\LRvel$ is the Lieb-Robinson velocity \cite{Lieb1972} and $T$ is the duration of time evolution or depth. Generic teleportation protocols involve concatenating measurement regions using pairs of measurements. Importantly, we show that \emph{all} measured observables must commute, and provide bounds on teleportation with $\Nmeas$ measurements. These results hold for any $\numQ \geq 1$.

We note that any quantum protocol $\chan$ can be written in ``canonical form'' (see Def.~\ref{def:chan canonical form}), $\chan = \QECChannel \, \MeasChannel \, \uchan$ by commuting and redefining various terms. The protocol $\chan$ begins with unitary time evolution, which prepares a \emph{resource state} $\ket{\Psi_t} = \uchan \, \ket{\Psi_0}$ (see Def.~\ref{def:resource state}). The protocol $\chan$ then applies all measurements (via $\MeasChannel$) to $\ket{\Psi_t}$, followed by all outcome-dependent Pauli operations (via $\QECChannel$). 

In Sec.~\ref{sec:SPT proof}, we prove that the resource state $\ket{\Psi_t}$ is an SPT with protecting symmetry $\Ints_2^{2\numQ}$. For $\numQ=1$, this state is in the same phase as the cluster state \cite{Briegel_2001}. We also show that the measured observables define the bulk operators of $2 \numQ$ \emph{string order parameters} that define the SPT order. Essentially, measuring the bulk of each pair of string order parameters leads to localizable entanglement that can be used to teleport a single logical qubit~\cite{LocalizableEntanglement,Marvian2017}. The general form of such protocols is depicted schematically in Fig.~\ref{fig:intro}.

In Sec.~\ref{sec:examples} we highlight these results using various examples, including nonstabilizer states and states for which the SPT order was not previously identified. Finally, in Sec.~\ref{sec:outlook}, we summarize our main results, comment on protocols beyond standard teleportation, and discuss future directions and extensions of this work.

\section{Teleportation overview}
\label{sec:overview}

Here we briefly review the task of quantum teleportation. In Sec.~\ref{subsec:teleport prelim} we identify the physical systems of interest, corresponding to qubits on a graph $G$, along with various useful operators. In Sec.~\ref{subsec:state transfer}, we define quantum state transfer---the measurement-free analogue of teleportation---and prove several results.  In Sec.~\ref{subsec:Stinespring}, we review the Stinespring representation \cite{AaronDiegoFuture,  SpeedLimit, AaronMIPT, Stinespring, ChoisThm, Takesaki} of measurements and outcome-dependent operations as unitary channels on a dilated Hilbert space. In this sense, teleportation can be viewed as state transfer on the dilated Hilbert space, as we explain in Sec.~\ref{subsec:phys teleport}, where we also consider an example of teleportation using the 1D cluster state \cite{Briegel_2001}, highlighting the Stinespring representation. Lastly, in Sec.~\ref{subsec:bounds}, we present Theorem~\ref{thm:LR theorem} due to Lieb and Robinson \cite{Lieb1972}, which establishes a bound \eqref{eq:LR bound} on quantum state transfer. We then present Theorem~\ref{thm:FYHL}  \cite{SpeedLimit}, which extends Theorem~\ref{thm:LR theorem} to generic quantum channels (involving measurements and nonlocal feedback) with a modified bound \eqref{eq:FYHL bound}, and constrains teleportation. 

\subsection{Preliminaries}
\label{subsec:teleport prelim}

We consider physical systems of $\Nspins$ qubits identified with the vertices $v \in V$ of an arbitrary graph $G$. The $\Nspins$-qubit Hilbert space has the tensor-product structure
\begin{align}
    \label{eq:physical Hilbert space}
    \Hilbert^{\vpp}_{\rm ph}  =  \bigotimes\limits_{v \in V} \, \Comps^2  = \Comps^{2^\Nspins} \, ,~~
\end{align}
and dimension $\HilDim^{\,}_{\rm ph} = \Dim (\Hilbert^{\,}_{\rm ph}) = 2^\Nspins$. We denote \emph{Pauli-string} operators acting on $\Hilbert^{\,}_{\rm ph}$ as, e.g.,
\begin{align}
    \label{eq:Pauli string def}
    \Gamma^{\vpp}_{\bvec{n}} = \bigotimes\limits_{v \in V} \, \Pauli{n^{\,}_v}{v} \equiv \Pauli{\bvec{n}}\, , ~~
\end{align}
where $\bvec{n} = \{\, n^{\,}_v \, | \, \forall \,  v \in V \, \}$ is a length-$\Nspins$ vector that encodes the Pauli content $n^{\,}_v \in \{0,1,2,3\}$ on each vertex $v\in V$, and $\Pauli{0}{v} \equiv \ident$, $\Pauli{1}{j}=\Pauli{x}{j}=X_j$, and so on. The $4^{\Nspins}$ unique Pauli strings form a complete, orthonormal basis for operators acting on $\Hilbert^{\,}_{\rm ph}$ \eqref{eq:physical Hilbert space}; when multiplied by $\ii^k$ (for $k=0,1,2,3$), these operators realize  the \emph{Pauli group} $\PauliGroup*{\Nspins}$. Relatedly, the \emph{Clifford group} $\CliffGroup*{\Nspins}$ is the set of unitary operators that map the Pauli group $\PauliGroup*{\Nspins}$ to itself. Note that we frequently label both the Clifford and Pauli groups by the vertex set or Hilbert space to which they correspond (e.g., $\CliffGroup*{\Nspins} = \CliffGroup*{V}=\CliffGroup*{\rm ph}$).

Additionally, for a given $\Nspins$-qubit state $\ket{\Psi} \in \Hilbert^{\,}_{\rm ph}$, we define the \emph{unitary stabilizer group} 
\begin{equation}
    \label{eq:UStab}
    \stabilizer = \UStabOf{\ket{\Psi}} = \left\{ \, \StabEl \in \Unitary{2^\Nspins} \,  \middle| \, \StabEl \ket{\Psi} = \ket{\Psi} \,\right\}\, ,~~
\end{equation}
as the set of all unitaries $\StabEl$ for which $\ket{\Psi}$ is a $+1$ eigenstate (note that $\stabilizer$ always contains the identity $\ident$). Crucially, we \emph{do not} require that $\stabilizer$ \eqref{eq:UStab} be a subset of the Pauli group $\PauliGroup*{\Nspins}$---i.e., we do not restrict to ``stabilizer states'' $\ket{\Psi}$ \cite{gottesman1997stabilizer}, for which all $\StabEl\in \UStabOf*{\ket{\Psi}}$ are Pauli strings \eqref{eq:Pauli string def}. In fact, in Sec.~\ref{subsec:hypergraph}, we explicitly consider teleportation protocols involving a nonstabilizer \emph{hypergraph} resource state \cite{CliffordHierarch, hypergraph, hypergraphMiyake, hypergraphMiyake2D}. Moreover, we do not require that $\stabilizer$ \eqref{eq:UStab} be Abelian (which is guaranteed when  $\stabilizer \subset \PauliGroup*{\Nspins}$).

\subsection{Quantum state transfer}
\label{subsec:state transfer}
Before considering quantum teleportation, we briefly discuss the related task of \emph{quantum state transfer}. By convention, ``quantum state transfer'' refers to unitary protocols $\chan$ applied to a system of physical qubits, while ``quantum teleportation'' also involves projective measurements and outcome-dependent operations (feedback). 

Both tasks transfer $\numQ \geq 1$ ``logical states'' within a given many-body state. A logical state is simply a particular state of a quantum bit---i.e., $\ket{\psi^{\,}_n} = \alpha^{\,}_n \, \ket{0}^{\,}_n + \beta^{\,}_n \ket{1}^{\,}_n$ for the $n$th ``logical qubit'' with logical basis states $\ket{0}^{\,}_n$ and $\ket{1}^{\,}_n$. The simplest protocols $\chan$ transfer logical states between individual physical qubits; however, logical states may also be ``embedded'' in many physical qubits. 

Logical qubits can also be characterized in terms of their associated logical \emph{operators}. If the many-body state $\ket{\Psi}$ encodes $k\geq 1$ logical states $\ket{\psi^{\,}_n}$ with coefficients $\alpha^{\,}_n,\beta^{\,}_n$ (for $1\leq n \leq k$), then the corresponding logical basis operators $\LXp{(n)},\LYp{(n)},\LZp{(n)}$ are defined such that,
\begin{subequations}
\label{eq:logical alpha beta expval}
    \begin{align}
    \matel*{\Psi}{\LXp{(n)}}{\Psi}  &= \matel*{\psi^{\vpp}_n}{\PX{}}{\psi^{\vpp}_n} =  2 \, \mathrm{Re}(\alpha^*_n \beta^{\vpp}_n) \label{eq:logical X expval gen} \\
    \matel*{\Psi}{\LYp{(n)}}{\Psi}   &= \matel*{\psi^{\vpp}_n}{\PY{}}{\psi^{\vpp}_n} =  2\, \mathrm{Im}(\beta^*_n \alpha^{\vpp}_n) \label{eq:logical Y expval gen}  \\
    \matel*{\Psi}{\LZp{(n)}}{\Psi}  &= \matel*{\psi^{\vpp}_n}{\PZ{}}{\psi^{\vpp}_n} = \abs{\alpha^{\vpp}_n}^2 - \abs{\beta^{\vpp}_n}^2 \, ,~~\label{eq:logical Z expval gen} 
    \end{align}
\end{subequations}
so that the logical Paulis $\LXp{(n)},\LYp{(n)},\LZp{(n)}$ span all logical operations on the $n$th logical qubit. 

In general, the logical operators must reproduce the Pauli algebra and satisfy Eq.~\ref{eq:logical alpha beta expval}. However, in the context of state transfer and teleportation, we assume that the initial and final many-body states are separable with respect to the $\numQ$ logical qubits and all other physical qubits, and realized on $\numQ$ physical sites. Often, the $n$th logical qubit is realized on the physical qubit $j$, so that $\LXp{(n)},\LYp{(n)},\LZp{(n)}$ reduce to the standard Paulis $\PX{j}$, $\PY{j}$, and $\PZ{j}$.

\begin{defn}\label{def:state transfer}
A unitary protocol $\chan$ acting on $\Hilbert^{\,}_{\rm ph}$ \eqref{eq:physical Hilbert space} achieves \emph{quantum state transfer} of $\numQ$ logical qubits from the initial logical vertex set $I = \{i_1, \dots, i_{\numQ} \} \subset V$ to the final logical vertex set $F  = \{f_1, \dots, f_{\numQ} \} \subset V$ if %and only if
\begin{equation}
\label{eq:W state transfer}
    \ket{ \Psi^{\vpp}_T \left( \bvec{\alpha},\bvec{\beta} \right) } =  \chan\, \ket{ \Psi^{\vpp}_0 \left( \bvec{\alpha},\bvec{\beta} \right) } \, ,~~
\end{equation}
with the initial and final states given by
\begin{align}
    \ket{ \Psi^{\vpp}_0 \left( \bvec{\alpha},\bvec{\beta} \right) }  &=  \left[ \bigotimes\limits_{n =1}^{\numQ} \, \ket{\psi^{\vps}_n(\alpha^{\vps}_n, \beta^{\vps}_n)}^{\vpd}_{i_n} \right] \otimes \ket{\Phi^{\vpp}_0}^{\vpd}_{I^c} \label{eq:teleport initial mb state} \\
    \ket{ \Psi^{\vpp}_T \left( \bvec{\alpha},\bvec{\beta} \right) } &= \left[ \bigotimes\limits_{n =1}^{\numQ} \, \ket{\psi^{\vps}_n(\alpha^{\vps}_n, \beta^{\vps}_n)}^{\vpd}_{f_n} \right] \otimes \ket{\Phi^{\vpp}_T}^{\vpd}_{F^c}\, , ~\label{eq:teleport final mb state}
\end{align}
where the logical state $\ket{\psi^{\,}_n} = \alpha^{\,}_n \ket{0} + \beta^{\,}_n \ket{1}$ is transferred from site $i^{\,}_n \in I$  \eqref{eq:teleport initial mb state} % the initial state $\ket{\Psi^{\,}_0}$ 
to site $f^{\,}_n \in F$  \eqref{eq:teleport final mb state}, %the final state $\ket{\Psi^{\,}_T}$, 
$\bvec{\alpha} = (\alpha^{\,}_1, \dots, \alpha^{\,}_{\numQ})$ and  $\bvec{\beta} = (\beta^{\,}_1, \dots, \beta^{\,}_{\numQ})$ 
satisfy $\abs{\alpha^{\,}_n}^2 +\abs{\beta^{\,}_n}^2 =1$ for each $n$, and $\ket{\Phi^{\,}_0}$ and $\ket{\Phi^{\,}_T}$ are arbitrary many-body states on the $\Nspins-\numQ$ qubits in $I^c$ and $F^c$, respectively. 
\end{defn}

We now present Prop.~\ref{prop:teleportation conditions}, which establishes necessary and sufficient conditions for quantum state transfer in terms of logical and stabilizer operators, an equivalence between the transfer of logical states and logical operators, and that the logical operators in the initial and final many-body states are spanned by the Pauli operators $\PX{i_n},\PZ{i_n}$ and $\PX{f_n},\PZ{f_n}$, for $i^{\,}_n \in I$ and $f^{\,}_n \in F$,  respectively.

\begin{prop}[State-transfer conditions]
\label{prop:teleportation conditions}
Given an initial state of the form \eqref{eq:teleport initial mb state}, the following conditions are necessary and sufficient for a unitary protocol $\chan$ to achieve quantum state transfer of $\numQ$ logical states $\ket{\psi^{\,}_n}$,
\begin{subequations}
\label{eq:teleportation conditions}
\begin{align}
    \chandag \PX{f_n} \chan &= \PX{i_n} \, \StabEl^{}_{n,x}  \label{eq:X teleportation condition} \\
    \chandag \PZ{f_n} \chan &= \PZ{i_n} \, \StabEl^{}_{n,z}  \label{eq:Z teleportation condition} \, ,
\end{align}
\end{subequations}
for $1 \leq n \leq \numQ$, where $\StabEl^{}_{n,\nu} \in \UStabOf{\ket{\Psi^{\,}_0}}$ \eqref{eq:UStab} are unitary and stabilize $\ket{\Psi^{\,}_0}$ \eqref{eq:teleport initial mb state}; see also Def.~\ref{def:state transfer}.
\end{prop}
The proof of Prop.~\ref{prop:teleportation conditions} appears in App.~\ref{app:Proof state=operator teleport}. Loosely speaking, Prop.~\ref{prop:teleportation conditions} establishes that the task of transferring $\numQ$ logical states from the logical sites $I=\{i_1,\dots,i_{\numQ}\} \subset V$ to the logical sites $F=\{f_1,\dots,f_{\numQ}\} \subset V$ in the Schr\"odinger picture is equivalent to transferring $\numQ$ pairs of $X$- and $Z$-type logical operators from the logical sites in $F$ to those in $I$ in the Heisenberg picture. The result for the $n$th $Y$-type logical operator follows from $\PY{} = \ii \PX{} \PZ{}$:
\begin{align}\label{eq:Y teleportation condition}
    \chandag \PY{f_n} \chan &= \PY{i_n} \, \StabEl^{}_{n,y} \, ,
\end{align}
where $\StabEl^{}_{n,y} \equiv \PZ{i_n} \StabEl^{}_{n,x} \PZ{i_n} \StabEl^{}_{n,z} \in \UStabOf{\ket{\Psi^{\vps}_0}}$. 

Importantly the result \eqref{eq:teleportation conditions} of Prop.~\ref{prop:teleportation conditions} captures how \emph{any} logical operation $\observ^{\,}_n$ on the $n$th logical qubit in the \emph{final} state $\ket{\Psi^{\,}_T}$ \eqref{eq:teleport final mb state} %following state transfer 
is related to the same logical operation $\observ^{\,}_n$ on the $n$th logical qubit in the \emph{initial} state $\ket{\Psi^{\,}_0}$ \eqref{eq:teleport initial mb state}. Since the logical Paulis $\Pauli{\nu}{n,{\rm L}}$ for $\nu=1,2,3$ form a complete basis for $\observ^{\,}_n$ in the final state $\ket{\Psi^{\,}_T}$ \eqref{eq:teleport final mb state}, Prop.~\ref{prop:teleportation conditions} implies that 
\begin{align}
    O^{\vpp}_{f_n} (T) &= \chan^{\dagger} \, O^{\vpp}_{f_n} \chan \notag \\
    % &= \frac{1}{2} \trace\limits_{f_n} \left[ O^{\vpp}_{f_n} \right] \ident + \frac{1}{2} \sum\limits_{\nu=1}^3 \, \trace\limits_{f_n} \left[ O^{\vpp}_{f_n} \, \Pauli{\nu}{f_n}\right] \, \chan^{\dagger} \, \Pauli{\nu}{f_n} \, \chan \notag \\
    &= \frac{1}{2} \sum\limits_{\nu=0}^3 \, \trace\limits_{f_n} \left[ O^{\vpp}_{f_n} \, \Pauli{\nu}{f_n}\right] \, \chan^{\dagger} \, \Pauli{\nu}{f_n} \, \chan \notag \\
    % &= \frac{1}{2} \trace\limits_{i_n} \left[ O^{\vpp}_{i_n} \right] \ident + \frac{1}{2} \sum\limits_{\nu=1}^3 \, \trace\limits_{i_n} \left[ O^{\vpp}_{i_n} \, \Pauli{\nu}{i_n}\right] \, \Pauli{\nu}{i_n} \, \StabEl^{(\nu)}_n \notag \\
    &= \frac{1}{2} \sum\limits_{\nu=1}^3 \, \trace\limits_{i_n} \left[ O^{\vpp}_{i_n} \, \Pauli{\nu}{i_n}\right] \, \Pauli{\nu}{i_n} \, \StabEl^{}_{n,\nu} \notag \\
    &\equiv O^{\vpp}_{i_n} \StabEl^{}_{n,O} \, ,
    \label{eq:logical observable transfer}
\end{align}
is equivalent to $O^{\,}_{i_n}$ acting on $\ket{\Psi^{\,}_0}$ \eqref{eq:teleport initial mb state}.

Prop.~\ref{prop:teleportation conditions} establishes that state transfer requires growing the support of logical operators. Quantitatively, the degree of ``overlap'' between two operators can be captured by commutator norms. To this end, we define the standard \emph{spectral norm} of an operator as
\begin{equation}
    \label{eq:spectral norm}
    \norm{\mobserv} \equiv \max\limits_{\lambda \in {\rm eigs}(\mobserv)} \, \abs{\lambda} \, ,~~
\end{equation}
which is  magnitude of the largest eigenvalue of $\mobserv$. Using the fact that $\norm{\acomm{\Pauli{\nu}{n,{\rm L}}}{\StabEl} } = 2$, we find that \cite{SpeedLimit}
\begin{align}
\label{eq:X Z Heisenberg commutator}
    \norm{\comm{\chandag \PX{f_n} \chan}{\PZ{i_n}}} = 2 \, ,
\end{align}
meaning that successful state transfer implies that the Heisenberg-evolved logical operator $\PX{f_n}(T)$ for the $n$th logical qubit realizes the operator $\PX{i_n}$ acting on the initial state, and obeys the Pauli algebra \eqref{eq:X Z Heisenberg commutator}.

Finally, we define the ``task distance'' $\Dist$ for quantum state transfer (and later, teleportation), as the minimum distance any logical qubit is transferred.

\begin{defn}
    \label{def:task dist}
    The \emph{task distance} (or \emph{teleportation distance}) for a protocol $\chan$ that maps the initial state $\ket{\Psi^{\,}_0}$ \eqref{eq:teleport initial mb state} with $\numQ$ logical qubits on the sites $I=\{i_1,i_2,\dots,i_{\numQ}\}\subset V$ to the state $\ket{\Psi^{\,}_T}$ \eqref{eq:teleport final mb state} with logical qubits on the sites $F = \{f_1,f_2,\dots,f_{\numQ}\}\subset V$ is given by
    \begin{equation}
        \label{eq:L def}
        \Dist \equiv \min\limits_n \, d(i_n,f_n) \, , ~~
    \end{equation}
    where the %distance 
    metric $d(x,y)$ is defined by the graph $G$.
\end{defn}

\subsection{Stinespring measurement}
\label{subsec:Stinespring}

To facilitate our discussion of quantum teleportation---which involves measurements and outcome-dependent feedback---we now review the Stinespring formalism \cite{SpeedLimit, AaronMIPT, AaronDiegoFuture, Stinespring}, which derives from the Stinespring Dilation Theorem \cite{Stinespring, ChoisThm, Takesaki}. The formalism leads to a \emph{unitary} representation of projective measurements and outcome-dependent operations in a \emph{dilated} Hilbert space,
\begin{equation}
    \label{eq:dilated Hilbert space}
    \Hilbert^{\vpd}_{\rm dil} \equiv \Hilbert^{\vpd}_{\rm ph} \otimes \Hilbert^{\vpd}_{\rm ss} \, ,~~
\end{equation}
where $\Hilbert^{\,}_{\rm ph}$ is the physical Hilbert space \eqref{eq:physical Hilbert space} and $\Hilbert^{\,}_{\rm ss}$ is the ``Stinespring'' Hilbert space. For a protocol $\chan$ involving $\Nmeas$ single-qubit measurements, $\Hilbert^{\,}_{\rm ss} = \Comps^{2 \Nmeas}$ contains $\Nmeas$ qubits, which store the $\Nmeas$ outcomes (and reflect the states of the corresponding measurement apparati).

The crucial result of the Stinespring Theorem \cite{Stinespring, Takesaki, ChoisThm} in this context is that \emph{all} valid quantum channels (i.e., positive maps) can be represented via \emph{isometries}. An isometry is a map $\isometry \, :\,  \Hilbert \to \Hilbert'$, where $\HilDim \leq \HilDim'$. When $\HilDim=\HilDim'$, $\isometry$ is a \emph{unitary}; when $\HilDim<\HilDim'$, $\isometry$ \emph{dilates} $\Hilbert$ to $\Hilbert'= \Hilbert^{\,}_{\rm dil}$  \eqref{eq:dilated Hilbert space}. Crucially, one can always \emph{embed} the isometry $\isometry \, :\,  \Hilbert^{\,}_{\rm ph} \to \Hilbert^{\,}_{\rm dil}$ in a unitary $\cliff \, : \, \Hilbert^{\,}_{\rm dil} \to \Hilbert^{\,}_{\rm dil}$ by working in $\Hilbert^{\,}_{\rm dil}$ from the outset; this merely requires identifying a \emph{default} initial state of all measurement apparati, which we take to be $\ket{0}^{\,}_{\rm ss}$ without loss of generality \cite{Takesaki, ChoisThm, AaronMIPT, SpeedLimit, Stinespring, AaronDiegoFuture}.

We first consider the Stinespring representation of projective measurements. An observable $\mobserv$ with $\Noutcome$ unique eigenvalues $\meig{n}$ has a \emph{spectral decomposition} given by
\begin{align}
    \label{eq:spectral decomp}
    \mobserv &= \sum\limits_{n=0}^{\Noutcome-1} \, \mEig{n} \, \Proj{n}{\mobserv} \, ,~~
\end{align}
where the projectors are orthogonal, idemptotent, and complete. The trace of each projector is equal to the degeneracy of the corresponding eigenvalue. The dilated unitary channel that captures measurement of $\mobserv$ \eqref{eq:spectral decomp} is
\begin{align}
\label{eq:unitary meas general}
    \Umeas{\mobserv} &= \sum\limits_{m,n=0}^{\Noutcome-1} \, \Proj{m}{\mobserv} \otimes \BKop{n+m}{n}^{\vpp}_{\rm ss} \, ,~~
\end{align}
which is unique up to the definition of the Stinespring basis and choice of default state \cite{Stinespring, Takesaki, AaronMIPT, SpeedLimit, AaronDiegoFuture, ChoisThm}. Projecting the post-measurement state onto the $n$th Stinespring state leads to the physical state $\ket{\psi'} \propto \Proj{n}{\mobserv} \, \ket{\psi}$ \eqref{eq:spectral decomp}, as required. For example, when $\mobserv = \PZ{j}$, we have
\begin{align}
\label{eq:fluorescent measurement}
    \Umeas{\PZ{j}} &= \frac{1}{2} \sum\limits_{m=0,1} \left( \ident + (-1)^m \, \PZ{j} \right) \otimes \SSXp{j}{m} ~~  \nonumber \\
    &= \Proj{0}{j} \otimes \SSid{j} + \Proj{1}{j} \otimes \SSXp{j} \, ,
\end{align}
which is simply a CNOT gate from the physical site $j$ to the Stinespring qubit $j$. In general, we use tildes to denote operators on the Stinespring registers, which we label according to the corresponding observable. 

The Stinespring formalism also represents outcome-dependent operations as controlled unitary operations on the physical qubits, conditioned on the Stinespring qubits. Since we assume classical communication (of the outcomes) to be \emph{instantaneous}, such operations may be conditioned on \emph{any} prior outcomes (in the Schr\"odinger picture), and take the generic form
\begin{align}
    \label{eq:unitary feedback general}
    \QECChannel^{\vpd}_{\Omega} &= \sum\limits_{\bvec{n}} \, \ConjChan{\bvec{n}} \otimes \SSProj{\bvec{n}}{\Omega} \, , ~~
\end{align}
where $\SSProj{\bvec{n}}{\Omega}$ is an $\SSZ{}$-basis projector onto the outcome ``trajectory'' $\bvec{n}=(n^{\,}_1, \dots, n^{\,}_{s})$ for the subset $\Omega \subset V^{\,}_{\rm ss}$.  The ``recovery operators'' $\conjchan^{\,}_{\bvec{n}}$ must act \emph{unitarily} on $\Hilbert^{\,}_{\rm ph}$ and/or any unused Stinespring registers. For example, one might apply $\PX{i}$ to the physical qubit $i$ if the $j$th outcome was $-1$, and do nothing otherwise, so that $\QECChannel^{\,}_{\Omega}$~\eqref{eq:unitary feedback general} becomes
\begin{align}
\label{eq:feedback CNOT}
    \QECChannel^{\vpd}_{j \to i} &=  \ident^{\vpd}_{i} \otimes \SSProj{0}{j} + \PX{i} \otimes \SSProj{1}{j} \, , ~~
\end{align}
which is simply a CNOT gate from the Stinespring site $j$ to the physical qubit $i$, in contrast to $\MeasChannel^{\,}_{[Z]} $\eqref{eq:fluorescent measurement}.

\subsection{Physical teleportation}
\label{subsec:phys teleport}

We are now equipped to define quantum teleportation in a general sense. Importantly, we distinguish this notion of \emph{physical teleportation}---as described in Def.~\ref{def:phys teleport}---from the slightly restricted class of quantum teleportation protocols we consider later on. We then consider a  canonical example teleportation protocol involving the 1D cluster state~\cite{Briegel_2001}; the corresponding circuit is depicted in Fig.~\ref{fig:cluster state circuit}.

\begin{defn}
\label{def:phys teleport}
\emph{Physical teleportation} refers to any unitary protocol $\chan$ acting on a \emph{dilated} Hilbert space $\Hilbert^{\,}_{\rm dil}$ \eqref{eq:dilated Hilbert space} that (i) fulfills the criteria for quantum state transfer (see  Def.~\ref{def:state transfer}) and (ii) involves the combination of $\Nmeas>1$ measurements, instantaneous classical communication, and outcome-dependent operations.
\end{defn}

We stress that physical teleportation requires that the $\numQ$ logical qubits be localized to $\numQ$ physical qubits both before and after the protocol $\chan$. This is in contrast to more general definitions of teleportation---which we refer to as ``virtual teleportation''---in which the logical qubits may be spread out amongst many physical qubits, but can nonetheless be decoded using certain measurements thereof. This is the case for general teleportation protocols based on matrix product states~\cite{Gross2007,DominicMBQC_SPT}; we relegate the treatment of virtual teleportation to future work~\cite{VirtualTeleportFuture}. 

Importantly, the class of teleportation protocols described by Def.~\ref{def:phys teleport} is too broad to constrain. However, as far as we are aware, all \emph{known} physical teleportation protocols---in which the final state is of the form $\ket{\Psi^{\,}_T}$~\eqref{eq:teleport final mb state}---belong to a slightly restricted class of protocols. Before defining that class in Sec.~\ref{sec:restricted teleportation}, we briefly consider a canonical example protocol belonging to this class (depicted in Fig.~\ref{fig:cluster state circuit} for $\Nspins=5$ qubits) based on the 1D cluster state~\cite{Briegel_2001} that teleports a single logical qubit.

\begin{figure}[t]
\centering
\includegraphics[width=0.45\textwidth]{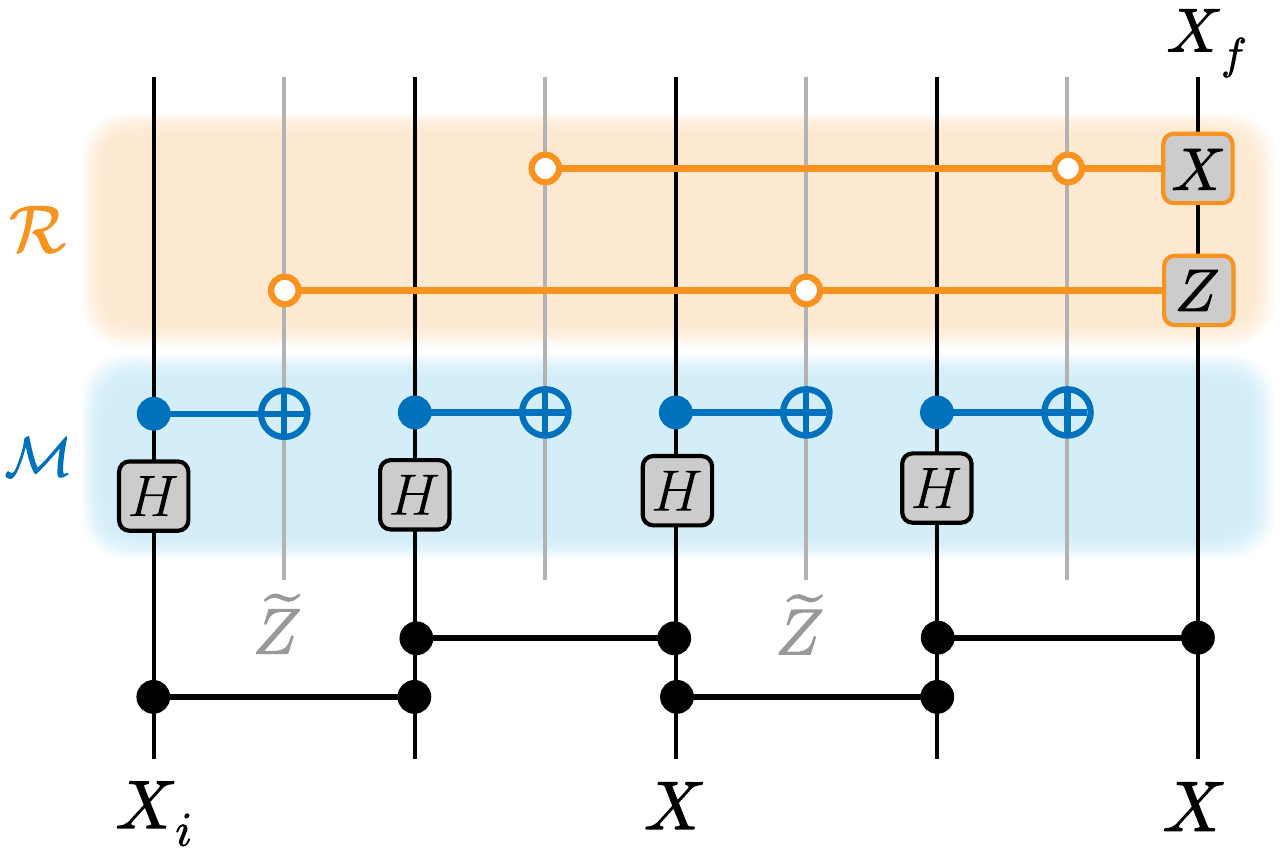}
\caption{Dilated quantum circuit for 1D cluster-state teleportation of a single logical qubit a distance $\Dist=4$, along with Heisenberg evolution of the logical $\PXp{\rm L}{}$  operator. The $\Nmeas=4$ Stinespring qubits are depicted in gray, the measurement channel in blue, and the outcome-dependent operations in orange (the operation on site $f=5$ is conditioned on the parity of the Stinespring qubits, indicated by white circles). The initial state is given by Eq.~\ref{eq:cluster initial state}, with $\ket{\psi}$ on the leftmost site, so that $\PXp{\rm L}{}(T) = \PXp{1}{} \PXp{3}{} \PXp{5}{}$ is a valid logical operator.}
\label{fig:cluster state circuit}
\end{figure}

The initial state is the product state
\begin{align}\label{eq:cluster initial state}
    \ket{\Psi^{\vps}_0} = \ket{\psi}^{\vpp}_1 \otimes \ket{+++\dots}^{\vpp}_{\rm rest} \otimes \ket{\bvec{0}}^{\vpp}_{\rm ss} \, ,~~
\end{align}
where $\ket{+}$ is the $+1$ eigenstate of $\PX{}$,  $\ket{\bvec{0}}=\ket{0\cdots 0}$ is the default Stinespring state, and we take $\Nspins$ to be odd.  We first apply to $\ket{\Psi^{\,}_0}$~\eqref{eq:cluster initial state} the (physical) unitary circuit
\begin{align}
\label{eq:cluster CZ circ}
    \uchan = \prod_{j=1}^{\Nspins-1} \mathrm{CZ}^{\vps}_{j,j+1} \, ,
\end{align}
where the controlled-$\PZ{}$ gate $\mathrm{CZ}=\mathrm{diag}(1,1,1,-1)$ in the two-qubit computational ($\PZ{}$) basis. Applying $\uchan$ to $ \ket{\Psi^{\vps}_0}$ prepares the so-called cluster state \cite{Briegel_2001}.

We then measure $\PX{}$ on all sites $j=1,\dots,\Nspins-1$ (where $i=1$ and $f=\Nspins$ are the initial and final logical sites, respectively). This is represented by the measurement channel $\MeasChannel$, corresponding to the blue-shaded region of Fig.~\ref{fig:cluster state circuit}, which is the product of the channels
\begin{align}
    \MeasChannel^{\vpd}_j = \Proj{+}{j} \otimes \SSid{j} + \Proj{-}{j} \otimes \SSX{j} \, ,~~
\end{align}
where the projectors are in the $\PX{j}$ basis, and the Stinespring registers are marked in light gray in Fig.~\ref{fig:cluster state circuit}. 

Finally, we apply the outcome-dependent recovery operations. The parity of outcomes on ``odd'' sites $j=2n+1$ (with $j<\Nspins$) determines a final $\PZ{\Nspins}$ rotation, while the parity of outcomes on ``even'' sites $j=2n+2$ (with $j<\Nspins)$ determines a final $\PX{f}=\PX{\Nspins}$ rotation, according to  
\begin{subequations}
\label{eq:cluster QEC}
\begin{align}
    \QECChannel^{\vpd}_z &= \ident^{\vpd}_{f} \otimes \SSProj{0}{\rm even} + \PX{f} \otimes \SSProj{1}{\rm even} \\
    \QECChannel^{\vpd}_x &= \ident^{\vpd}_{f} \otimes \SSProj{0}{\rm odd} + \PZ{f} \otimes \SSProj{1}{\rm odd} \, ,
\end{align}
\end{subequations}
where, e.g., $\SSProj{n}{\rm even}$ projects onto the parity-$n$ configuration $\SSZ{\rm even} = (-1)^n$ of the even Stinespring qubits. The channels above appear in the  orange-shaded part of Fig.~\ref{fig:cluster state circuit}; they are named according to the logical operator they modify, and may be applied in either order. After applying these three steps, the initial logical state $\ket{\psi}$ appears on site $f$ for any sequence of measurement outcomes.

Having defined the cluster-state teleportation protocol $\chan$, we can evolve the final-state logical operator $\LX=\PX{f}$  backwards in time via the Heisenberg picture. We have
\begin{subequations}
\label{eq:cluster R update}
\begin{align}
    \QECChannel^\dagger \PX{f} \QECChannel &= \QECChannel^\dagger_x \PX{f} \QECChannel^{\vpd}_x = \PX{f} \prod_{j=1}^{\Dist/2} \SSZ{2j-1} \label{eq:cluster R update X}\\
    \QECChannel^\dagger \PZ{f} \QECChannel &= \QECChannel^\dagger_z \PZ{f} \QECChannel^{\vpd}_z = \PZ{f} \,  \prod_{j=1}^{\Dist/2} \SSZ{2j} \, ,~~\label{eq:cluster R update Z}
\end{align}
\end{subequations}
where $\Dist = \Nspins-1$. Next, evolving through $\MeasChannel$ gives
\begin{subequations} \label{eq:cluster premeasurement logical}
\begin{align}
    \MeasChannel^\dagger \QECChannel^\dagger \PX{f} \QECChannel \MeasChannel &= \PX{f}\, \prod_{j=1}^{\Dist/2} \PX{2j-1} \, \SSZ{2j-1}  \\
    \MeasChannel^\dagger \QECChannel^\dagger \PZ{f} \QECChannel \MeasChannel &= \PZ{f}\, \prod_{j=1}^{\Dist/2} \PX{2j} \, \SSZ{2j}  \, ,~~
\end{align}
\end{subequations}
and evolving through the circuit $\uchan$ \eqref{eq:cluster CZ circ} leads to
\begin{subequations}
\begin{align}
    \chandag \PX{f} \chan &= \PX{i}\, \underbrace{\prod_{j=1}^{\Dist/2} \PX{2j+1}\, \SSZ{2j-1}}_{\StabEl^{}_x} \\
    \chandag \PZ{f} \chan &= \PZ{i}\, \underbrace{\prod_{j=1}^{\Dist/2} \PX{2j} \, \SSZ{2j}}_{\StabEl^{}_z} \, ,~~
\end{align}
\end{subequations}
from which we identify $\StabEl^{\,}_x$ and $\StabEl^{\,}_z$ according to Prop.~\ref{prop:teleportation conditions}, and confirm that they indeed stabilize $\ket{\Psi^{\,}_0}$ \eqref{eq:cluster initial state}.

We now comment on several properties of the foregoing protocol. Observe that there are effectively two types of commuting measurements, corresponding to even versus odd sites. There are also two types of recovery channel---one for each logical operator---which depend only on the measurement outcomes on the even sites and odd sites, respectively. Finally, if we take the stabilizers $\StabEl^{}_x$ and $\StabEl^{}_z$ and evolve them forwards through $\uchan$, we find
\begin{subequations} \label{eq:cluster string ops}
\begin{align}
    \uchan \, \StabEl^{\vpp}_x \, \uchan^\dagger &= \PZ{2} \left(\prod_{j=2}^{\Dist/2} \PX{2j-1} \right)\PX{N} \\
    \uchan \, \StabEl^{\vpp}_z \, \uchan^\dagger &= \PZ{1}\PX{2} \left( \prod_{j=2}^{\Dist/2} \PX{2j} \right) \PZ{N}  ,~~
\end{align}
\end{subequations}
up to $\SSZ{}$ operators, which we recognize as the string order parameters of the cluster state \cite{Briegel_2001}. Observe that, while these operators necessarily commute in the bulk of the chain, the \emph{endpoint} operators anticommute---i.e., $\{\PZ{2},\PZ{1}\PX{2}\}=\{\PX{\Nspins},\PZ{\Nspins}\}=0$. This anticommutation demonstrates the nontrivial SPT order of the cluster state, as we discuss in further detail in Sec.~\ref{subsec:SPT intro}. 

In the remainder, we prove that many of the above features---including the emergence of string order parameters for the state to which $\MeasChannel$ is applied---are ubiquitous features of a large class of teleportation protocols.

\subsection{Bounds on teleportation}
\label{subsec:bounds}

Before defining a restricted class of teleportation protocols in Sec.~\ref{sec:restricted teleportation}, we briefly discuss general bounds on generic useful quantum tasks that derive from \emph{locality}. We first consider bounds for unitary protocols $\chan = \uchan$ (e.g., state transfer) that act only on $\Hilbert^{\,}_{\rm ph}$ \eqref{eq:physical Hilbert space}. In general, we regard $\uchan$ as a finite-depth quantum circuit (FDQC), though \emph{all} of our results are compatible with $\uchan$ being generated by evolution under a local Hamiltonian, up to exponentially small operator tails. In either case, the task distance $\Dist$ \eqref{eq:L def} obeys the Lieb-Robinson Theorem \cite{Lieb1972}, which we state below as Theorem~\ref{thm:LR theorem}.
 \begin{thm}[Lieb Robinson \cite{Lieb1972}]
    \label{thm:LR theorem}
    Consider a unitary $\chan$ realizing time evolution on $\Hilbert^{\,}_{\rm ph}$ \eqref{eq:physical Hilbert space} for total time $T$, and two operators $A$ and $B$ with support in regions $X$ and $Y$, where $\norm{A}=\norm{B}=1$ \eqref{eq:spectral norm} and $d(X,Y)=\Dist$. Then,
    \begin{equation}
        \label{eq:LR bound}
        \norm{\comm{\, \chan^{\dagger} A \, \chan \, }{B}} \lesssim e^{\LRvel \, T- \Dist} \, , ~~
    \end{equation}
    where the Lieb-Robinson velocity $\LRvel$ depends on $\chan$ and the graph $G$, and the inequality is exact in the limit $\Dist,T \gg 1$. 
\end{thm}
Theorem~\ref{thm:LR theorem} has been proven in the literature for numerous classes of local Hamiltonians (and circuits) \cite{Lieb1972, hastings_rev, tight_LR}. It has also been extended to long-range interactions \cite{LRfossfeig, chen2019finite, Kuwahara:2019rlw, Tran:2020xpc} and even local Lindblad dynamics \cite{Poulin}, up to minor modifications to the Lieb-Robinson velocity $\LRvel$. In the case of dynamics generated by local quantum circuits, the linear ``light cone'' is sharp, in that the commutator norm \eqref{eq:LR bound} is identically zero until $\LRvel T \geq \Dist$, after which it is $\Order{1}$. Detailed reviews of Lieb-Robinson bounds and their applications appear in, e.g., Refs.~\citenum{hastings_rev} and \citenum{chen2023speed}. 

In the context of teleportation, the use of projective measurements and outcome-dependent operations far from the corresponding measurement can realize a task distance $\Dist$ \eqref{eq:L def} that exceeds the Lieb-Robinson bound \eqref{eq:LR bound}. However, strictly local nonunitary channels are captured by Lindblad dynamics, and obey the \emph{same} Lieb-Robinson bound \eqref{eq:LR bound} \cite{Poulin}; hence, such channels are not useful to teleportation (compared to time evolution alone), and we do not consider them herein. Although teleportation can overcome the conventional Lieb-Robinson bound \eqref{eq:LR bound}, it is nonetheless constrained by an alternative bound due to Ref.~\citenum{SpeedLimit}, which we now state in the specific context of teleporting $\numQ$ qubits as Theorem~\ref{thm:FYHL}.

\begin{thm}[Teleportation bound \cite{SpeedLimit}]
\label{thm:FYHL}
    Consider a protocol $\chan$ involving (i) unitary time evolution generated by a local Hamiltonian $H(t)$ for duration $T$, (ii) $\Nmeas$ local projective measurements, and (iii) arbitrary local operations based all prior measurement outcomes. If $\chan$  teleports $\numQ$ logical states (per Def.~\ref{def:standard teleport}) or generates $\numQ$ bits of entanglement, then the task distance $\Dist$ \eqref{eq:L def} obeys the bound
    \begin{align}
    \label{eq:FYHL bound}
        \Dist \leq 2 \left( 1 + \left\lfloor \frac{\Nmeas}{\numQ} \right\rfloor \right) \LRvel \, T \, , ~~
    \end{align}
    where $\LRvel$ is determined solely by the Hamiltonian $H$ that generates time evolution in $\chan$. If $\chan$ is a circuit, then $T$ is the circuit depth and $\LRvel$ is the maximum distance between any two qubits acted upon by a single gate.  
\end{thm}

While the teleportation bound \eqref{eq:FYHL bound} depends on the number of measurements $\Nmeas$, the general bounds of Ref.~\citenum{SpeedLimit} depend on the number of ``measurement regions'' $\Nregions$ (see Sec.~\ref{subsec:k=1 daisy chain}). Teleporting a single qubit ($\numQ=1$), e.g., obeys
\begin{equation}
    \label{eq:FYHL k=1}
    \Dist \leq 2 \left( \Nregions + 1 \right) \LRvel T \, ,~~
\end{equation}
and we prove in Sec.~\ref{sec:teleport+locality} for all $\numQ \geq 1$ that standard teleportation protocols require $\Nmeas=2\numQ$ outcomes  per ``region,'' so the factor of two in the bound \eqref{eq:FYHL bound} is loose. In fact, it was conjectured in Ref.~\citenum{SpeedLimit} that a minimum of two distinct measurements are required per logical qubit, per measurement region. By comparison, generating long-range entanglement---e.g., preparing an $\Nspins$-qubit Greenberger-Horne-Zeilinger (GHZ) state \cite{GHZ89,GHZ_prep}---may succeed with a single measurement per region.

The bounds of Ref.~\citenum{SpeedLimit} apply to generic quantum tasks that transfer or generate quantum information, entanglement, and/or correlations using arbitrary local quantum channels and instantaneous classical communication. Even when gates are conditions on the outcomes of arbitrarily distant measurements,  a finite speed limit emerges. While all such protocols obey the generalized bound \eqref{eq:FYHL k=1}, only protocols involving local time evolution, projective measurements, and outcome-dependent feedback are compatible with violations of the standard Lieb-Robinson bound \eqref{eq:LR bound} \cite{Poulin, SpeedLimit}---and thus, teleportation.

\section{Protocols of interest}
\label{sec:restricted teleportation}

Here we detail the physical teleportation protocols of interest, motivated both by practical constraints and analytical necessity, culminating in Def.~\ref{def:standard teleport} in Sec.~\ref{subsec:standard teleport}. Importantly, the ``standard teleportation protocols'' captured by Def.~\ref{def:standard teleport} include all known physical teleportation protocols.

\subsection{Canonical form}
\label{subsec:canonical}

Here we define the \emph{canonical form} $\chan$ of a measurement-based protocol $\chan'$. We note that physical unitaries $\cliff$ that act after a dilated channel (i.e., $\QECChannel_j'$ or $\MeasChannel_j'$) can always be ``pulled through'' to act before that channel, at the cost of modifying the latter according to $\cliff \QECChannel_j' = \QECChannel^{\vpp}_j \cliff$, where $\QECChannel^{\vpp}_j = \cliff \QECChannel_j' \cliff^{\dagger}$, and likewise for $\MeasChannel_j'$. See Fig. \ref{fig:canonical form}. 

\begin{defn}
    \label{def:chan canonical form}
    The \emph{canonical form} $\chan$ of the \emph{na\"ive} protocol $\chan'$ involving physical unitaries, projective measurements, and outcome-dependent operations is given by
    \begin{equation}
        \label{eq:chan canonical form}
        \chan = \QECChannel \,  \MeasChannel \, \uchan \, , ~~
    \end{equation}
    where $\uchan$ includes \emph{all} physical unitaries, $\MeasChannel$ includes \emph{all} measurements, and $\QECChannel$ includes \emph{all} recovery operations. 
\end{defn}

We now briefly comment on Def.~\ref{def:chan canonical form}. Crucially, $\chan$ and $\chan'$ are mathematically and physically the \emph{same} protocol (i.e., $\chan=\chan'$). However, the \emph{effective} measured observables  $\mobserv^{\,}_j$ and recovery operators $\conjchan^{\,}_j$ in $\chan$ may differ from those in the na\"ive protocol $\chan'$. The canonical form \eqref{eq:chan canonical form} is essential to defining the class of protocols of interest and numerous other results. Additionally, familiar teleportation protocols (such as the cluster-state example in Sec.~\ref{subsec:phys teleport}) already realize the canonical form \eqref{eq:chan canonical form}. 

However, by identifying modified channels as needed, any protocol $\chan'$ whose projective measurements are effectively instantaneous and whose recovery operations effectively realize a quantum circuit can be rewritten in canonical form $\chan$ \eqref{eq:chan canonical form}. The unitary operations captured by $\uchan$ may be generated by evolution under a local, time-dependent Hamiltonian. Using, e.g., the interaction picture, the corresponding unitaries can be ``pulled through'' the measurement channels and recovery operations without complication, because the latter are effectively instantaneous \cite{SpeedLimit}. We note that, in the presence of outcome-dependent measurements, one also includes in $\chan$ \eqref{eq:chan canonical form} a channel $\MeasChannel'$ between $\MeasChannel$ and $\QECChannel$ comprising outcome-dependent measurements. However, such channels are generally not useful to physical teleportation (see Def.~\ref{def:phys teleport}), and we do not consider them herein. Lastly, only protocols $\chan'=\chan$ \eqref{eq:chan canonical form} that are \emph{already} in  canonical form saturate the teleportation bound \eqref{eq:FYHL bound} of Theorem~\ref{thm:FYHL}.

\begin{figure}[t]
\centering
\includegraphics[width=0.45\textwidth]{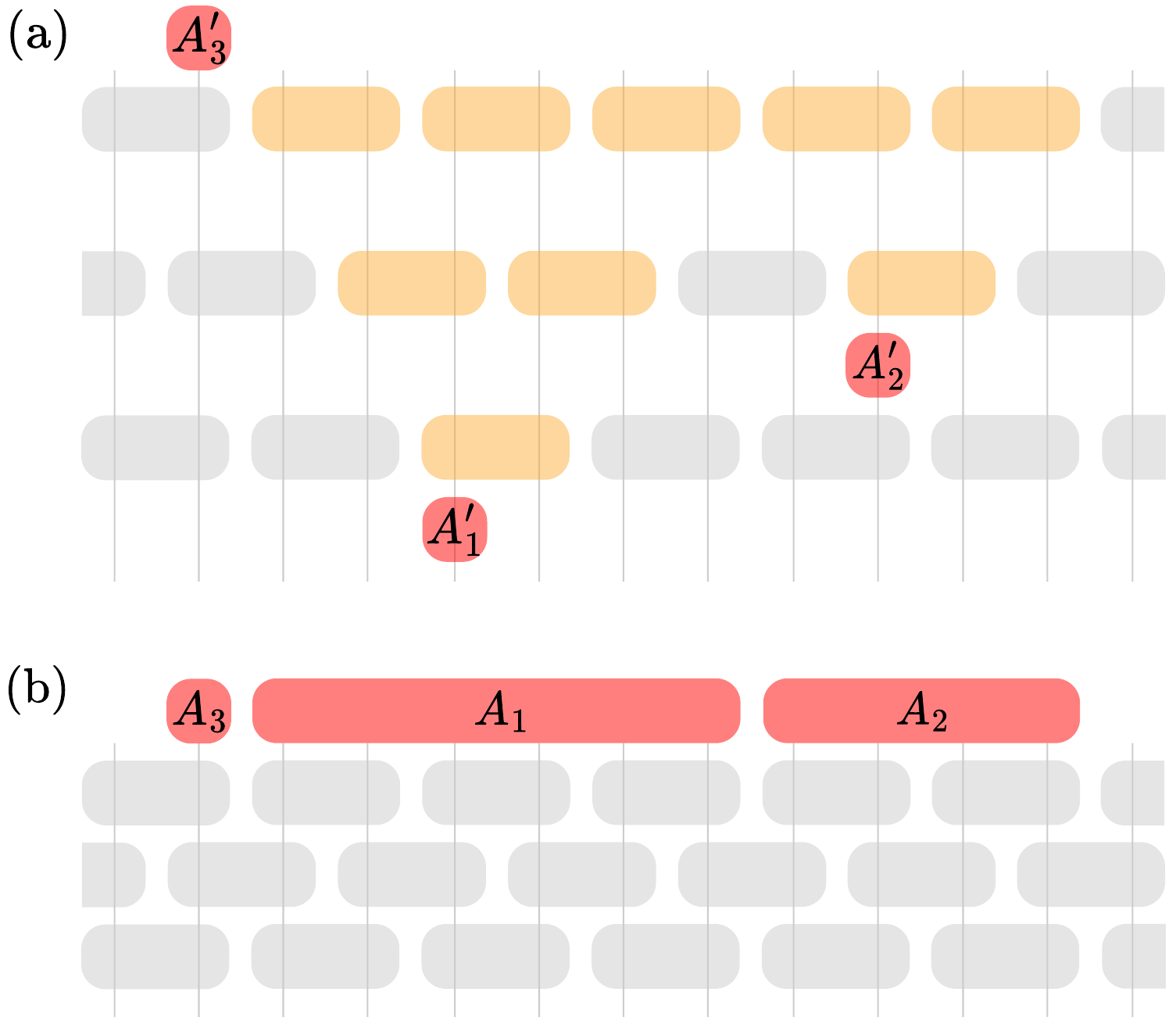}
\caption{\label{fig:canonical form}({\bf a}) A na\"ive teleportation protocol $\chan$  and ({\bf b}) the equivalent canonical-form protocol $\chan'$ \eqref{eq:chan canonical form}. The three layers of unitary gates are the same in both protocols; the measured observables $\mobserv_j'$ (in red) are modified as the unitaries in their light cone (in orange) are ``pulled through'' to realize $\mobserv^{\,}_j$ \eqref{eq:CF observable}.}
\end{figure}

\subsection{Measurement gates}
\label{subsec:allowed meas}

The first restriction we impose relates to the projective measurements in the dilated channel $\MeasChannel$ \eqref{eq:unitary meas general}. In the literature on quantum information processing, by convention, one only considers single-qubit measurements corresponding to $\PZ{j}$. From a theoretical standpoint, multiqubit projective measurements are \emph{universal} \cite{Nielsen_2003, leung2002twoqubit}, and can achieve teleportation acting on a trivial (product) states; as a result, no meaningful constraints can be derived in their presence. From an experimental perspective, in most quantum hardware with qubits amenable to useful quantum tasks (such as teleportation), all measurements are implemented using a combination of physical unitaries and single-qubit $\PZ{}$ measurements \eqref{eq:fluorescent measurement}. 

Here, we restrict to the measurement of single-qubit observables $\mobserv_j'$ in the na\"ive protocol $\chan'$. In the canonical-form protocol $\chan$ \eqref{eq:chan canonical form}, we allow observables $\mobserv^{\,}_j$ that are unitarily equivalent to single-qubit operators, i.e.,
\begin{equation}
    \label{eq:CF observable}
    \mobserv^{\vpp}_j = \cliff \, \mobserv_j' \, \cliff^{\dagger} \, ,~~
\end{equation}
for generality. Since single-qubit rotations are ``free'' (they do not count toward the depth $T$ of $\chan$), this restriction is equivalent to the standard one in the literature. 

We now state Prop.~\ref{prop:single-qubit observable}, which establishes that measuring any single-qubit observable is equivalent to measuring an \emph{involutory} observable (with eigenvalues $\pm 1$). 
\begin{prop}[Equivalent involutory measurement]
    \label{prop:single-qubit observable}
    The measurement of any nontrivial single-qubit observable $\mobserv^{\,}_j$ with real eigenvalues $\meig{0} > \meig{1}$ is equivalent to measurement of its \emph{involutory part} $\bar{\mobserv}^{\,}_j$, defined to be
    \begin{subequations}
    \label{eq:involutory part}
    \begin{align}
        \bar{\mobserv}^{\vpd}_{j} &= \frac{1}{\meig{0} - \meig{1}} \sum\limits_{n=1}^3 \, \trace\limits_j \, \left[ \, \mObserv{j}{\vpp} \, \Pauli{n}{j} \, \right]\, \Pauli{n}{j} \label{eq:involutory part decomp} \\
        &= \left( \meig{0}-\meig{1} \right)^{-1} \, \left( 2 \, \mObserv{j}{\vpp} - \left( \meig{0}+\meig{1}\right) \, \ident \right) \label{eq:involutory part nice} \, ,~~
    \end{align}
    \end{subequations}
    which is unitary, Hermitian, and squares to the identity.
\end{prop}

The proof of Prop.~\ref{prop:single-qubit observable} appears in App.~\ref{app:Proof of involutory meas}.  Importantly, Prop.~\ref{prop:single-qubit observable} implies that all measurement channels take the simple form $\MeasChannel^{\,}_{j}$ \eqref{eq:fluorescent measurement}, with $\PZ{j} \to \bar{\mobserv}^{\,}_j$ \eqref{eq:involutory part}. If the effective observable is given by $\mobserv^{\,}_j$ \eqref{eq:CF observable} in the canonical-form protocol $\chan$ \eqref{eq:chan canonical form}, then $\bar{\mobserv}^{\,}_j = \cliff \, \bar{\mobserv}_j' \cliff^{\dagger}$, where $\bar{\mobserv}_j'$ is the involutory part of the single-qubit observable $\mobserv_j'$ (unitary channels preserve involutions). As a result of Prop.~\ref{prop:single-qubit observable}, all measurements effectively have outcomes $\pm 1$, and spectral projectors $\Proj{n}{\mobserv} = \left( \ident + (-1)^n \, \bar{\mobserv} \right)/2$  \eqref{eq:spectral decomp}. This also simplifies the implementation of outcome-dependent operations, and facilitates the following result.

We next state Prop.~\ref{prop:feedback required}, which establishes that, if the outcome of a measurement channel $\MeasChannel^{\,}_j$ \eqref{eq:unitary meas general} is \emph{not} utilized later in the protocol $\chan$ then, on average, $\MeasChannel^{\,}_j$ either acts trivially on all logical operators or causes $\chan$ to fail. 

\begin{prop}[Necessity of feedback]
    \label{prop:feedback required}
    Consider the Heisenberg evolution of a  logical operator $\Gamma \in \PauliGroup*{\rm ph}$ \eqref{eq:Pauli string def} through the channel $\MeasChannel^{\,}_j$ corresponding to the measurement of a single-qubit observable $\mobserv^{\,}_j$ in the protocol $\chan$, averaged over outcomes. If no aspect of $\chan$ depends on the outcome of measuring $\mobserv^{\,}_j$, then $\MeasChannel^{\,}_j$ either acts trivially on $\Gamma$ (i.e., $\Gamma \to \Gamma$), or trivializes $\Gamma$ (i.e., $\Gamma \to 0$).
\end{prop}

The proof of Prop.~\ref{prop:feedback required} appears in App.~\ref{app:Proof feedback required}. Simply put, all projective measurements must be accompanied by at least one outcome-dependent operation to be useful to a quantum task $\chan$ characterized by a set of logical operators. Absent feedback, a measurement either does nothing to a given logical operator $\Gamma$, or sends $\Gamma \to 0$, implying failure of $\chan$. Hence, we may ignore any measurements in $\chan$ without corresponding feedback.

\subsection{Recovery operations}
\label{subsec:recovery}

Like all useful quantum tasks, a teleportation protocol $\chan$ outputs a state $\DensMat^{\,}_T = \BKop{\Psi^{\,}_T}{\Psi^{\,}_T}$ \eqref{eq:teleport final mb state}, from which, e.g., one extracts expectation values of observables. Hence, $\DensMat^{\,}_T$ must be prepared and sampled many times; accordingly, when $\chan$ involves measurements, the \emph{effective output state} $\DensMat^{\,}_T$ of $\chan$ is the one that is prepared upon \emph{averaging over outcomes}. When either $\DensMat^{\,}_T$ or its logical part is \emph{pure}, as with physical teleportation (see Def.~\ref{def:phys teleport}), the above immediately implies the following Proposition.

\begin{prop}[Hybrid preparation of pure states]
    \label{prop:pure output}
    Suppose that a protocol $\chan$ involving measurements and feedback outputs the state $\DensMat^{\,}_T$ when averaged over outcomes. If the \emph{logical state} $\DensMat^{\,}_{\rm L} = \trace_{F^c} \left[ \DensMat^{\,}_T \right]$ supported on $F \subset V$ is both pure and separable with respect to $F^c$, then $\chan$ must realize $\DensMat^{\,}_{\rm L}$ for \emph{every} sequence of measurement outcomes.
\end{prop}

The proof of Prop.~\ref{prop:pure output} appears in App.~\ref{app:Proof every trajectory}. Note that $F^c$ includes all Stinespring qubits, so that $\DensMat^{\,}_{\rm L}$ is implicitly averaged over all outcomes. Additionally, Prop.~\ref{prop:pure output} applies to protocols $\chan$ for which $\DensMat^{\,}_T$ itself is pure. As an aside, ``virtual'' teleportation protocols---in which the logical state is stored in the virtual legs of a tensor network \cite{Gross2007, DominicMBQC_SPT}---generally do not realize the logical states $\ket{\psi^{\,}_n}$ on \emph{any} finite subset $F \subset V$ of the physical qubits, even up to unitary operations on $F$. Instead, information is recovered from ostensibly unrelated operations \cite{Raussendorf2017}; we relegate virtual teleportation to future work.

We now consider the implications for the outcome-dependent operations in $\QECChannel$. In general, using measurements to achieve a quantum task only outputs the desired state up to certain \emph{byproduct operators}. The cost of using measurements to exceed the Lieb-Robinson bound \eqref{eq:LR bound} is that the state actually teleported is not $\ket{\psi}$, but $\conjchan \, \ket{\psi}$, where $\conjchan$ is the single-qubit byproduct operator, which is always invertible (i.e., unitary). For example, the cluster-state teleportation protocol considered in Sec.~\ref{subsec:phys teleport} outputs the logical state $\PZp{f}{n^{\,}_A} \, \PXp{f}{n^{\,}_B} \, \ket{\psi}$, rather than $\ket{\psi}$ itself (on qubit $f$). Averaging over outcomes twirls the state $\DensMat^{\,}_{\rm L} = \BKop{\psi}{\psi}$ over the Pauli matrices, so that $\DensMat^{\,}_{\rm L} = \ident / 2$ is simply a random classical bit. The same conclusion is can be reached by considering the Heisenberg evolution of logical operators, as captured by Prop.~\ref{prop:feedback required}.

At the same time, Prop.~\ref{prop:pure output} implies that the errors introduced by measurements must be correctable for \emph{each} sequence of outcomes $\bvec{n}$. This is only possible if the measurement outcomes alone are sufficient to determine the byproduct operator $\conjchan$, and that operator is \emph{invertible}, in which case one merely applies $\conjchan^{-1}$ to recover the logical state $\DensMat^{\,}_{\rm L}$ for every $\bvec{n}$. This is captured by the \emph{recovery channel} $\QECChannel$ in $\chan$ \eqref{eq:chan canonical form}, which applies physical unitaries conditioned on the states of the Stinespring qubits, which are ``read only'' after the corresponding measurement. 

In principle, there is a broad class of dilated ``recovery'' channels \eqref{eq:unitary feedback general} that are compatible with  physical teleportation. Instead, we consider a restricted class of \emph{Clifford} recovery channels, which we now define.

\begin{defn}
    \label{def:Clifford recovery}
    An outcome-dependent channel $\QECChannel^{\,}_{\Omega}$  \eqref{eq:unitary feedback general} is said to be \emph{Clifford} if it can be written as
    \begin{align}
    \label{eq:unitary feedback teleport}
    \QECChannel^{\vpd}_{\Omega} &= \ident \otimes \SSProj{0}{\Omega} + \ConjChan{\Omega} \otimes \SSProj{1}{\Omega} \, , ~~
    \end{align}
    in the canonical-form protocol $\chan$ \eqref{eq:chan canonical form}, up to multiplication on either side by elements of the physical Clifford group $\CliffGroup*{\rm ph}$, where $\conjchan^{\,}_{\Omega}$ is a physical Pauli string \eqref{eq:Pauli string def}, $\Omega \subset \widetilde{V}$ is a subset of the Stinespring qubits, and the operators
    \begin{align}
        \label{eq:SSZ parity}
        \SSProj{n}{\Omega} = \frac{1}{2} \Big( \SSid{} + (-1)^n \, \prod\limits_{j \in \Omega} \SSZ{j} \Big) \, , ~~
    \end{align}
    project onto states of $\Omega$ with $\SSZ{}$ parity $(-1)^n$.
\end{defn}

As far as we are aware, the recovery operations of \emph{all} physical teleportation protocols known to the literature are Clifford in the sense of Def.~\ref{def:Clifford recovery}. We also preclude outcome-dependent measurements, which are not only higher-order Clifford \cite{CliffordHierarch} on $\Hilbert^{\,}_{\rm dil}$ \eqref{eq:dilated Hilbert space}, but are generally not useful to physical teleportation on their own or in combination with the Clifford recovery channels of Def.~\ref{def:Clifford recovery}. Essentially, conditional measurements fail to preserve crucial properties of the logical operators (e.g., unitarity).

As an aside, we conjecture that physical teleportation protocols can be classified according to their \emph{recovery group}---the set generated by the \emph{recovery operators} $\conjchan^{\,}_{\Omega}$ for in the recovery channels $\QECChannel^{\,}_{\Omega}$ \eqref{eq:unitary feedback teleport} in $\chan$ \eqref{eq:chan canonical form} under multiplication. That group is equivalent to the set of all possible \emph{byproduct} operators. Restricting to the Clifford recovery channels of Def.~\ref{def:Clifford recovery}, the recovery group must be a subset of the Pauli group $\PauliGroup{\rm ph}$. Since the phase of the logical states $\ket{\psi^{\,}_n}$ is unimportant, we have $\PX{} \PZ{} \cong \PZ{} \PX{}$, so that the recovery group is effectively \emph{Abelian}. Since $\PauliGroup{\rm ph}$ also spans all physical operators, this is likely the simplest nontrivial recovery group that can realize. We relegate the consideration of more complicated recovery channels (and groups) to future work.

\subsection{Standard teleportation protocols}
\label{subsec:standard teleport}

We now define the ``standard'' teleportation protocols that we consider in the remainder. These protocols fulfill the criteria of physical teleportation and quantum state transfer (see Defs.~\ref{def:phys teleport} and Def.~\ref{def:state transfer}), as well as the properties mentioned in the preceding subsections. As far as we are aware, ``standard'' teleportation protocols include all physical teleportation protocols known to the literature.

\begin{defn}
    \label{def:standard teleport}
    A \emph{standard teleportation protocol} $\chan$ is a physical teleportation protocol (see Def.~\ref{def:phys teleport}) for which the initial state $\ket{\Psi^{\,}_0}$ \eqref{eq:teleport initial mb state} is a product state and, when written in canonical form $\chan = \QECChannel \, \MeasChannel \, \uchan$ \eqref{eq:chan canonical form}, (i)  all measured observables $\mobserv^{\,}_j$ are unitarily equivalent to single-qubit operators; (ii)  all recovery operations are Clifford \eqref{eq:unitary feedback teleport} in the sense of Def.~\ref{def:Clifford recovery}; (iii) no unitaries are applied to sites already measured (in the Schr\"odinger picture); and (iv) the task distance $\Dist = d(I,F) > \LRvel T$ \eqref{eq:L def} exceeds the Lieb-Robinson bound \eqref{eq:LR bound}, where $\LRvel,T$ are determined by the local unitary $\uchan$, which is compatible with Theorem~\ref{thm:LR theorem}.
\end{defn}

We now comment on the details standard teleportation (Def.~\ref{def:standard teleport}), compared to the more general case of physical teleportation (Def.~\ref{def:phys teleport}). We require that $\Dist > \LRvel T$ to ensure that $\chan$ is a teleportation---rather than state-transfer---protocol (see Def.~\ref{def:state transfer}). In combinatition with the requirement that $\ket{\Psi^{\,}_0}$ \eqref{eq:teleport initial mb state} be a product state, demanding $\Dist > \LRvel T$ furnishes the locality-based analyses of Sec.~\ref{sec:teleport+locality}, which also requires that $\ket{\Psi^{\,}_0}$ be separable  with respect to a biparation of the graph $G$ compatible with Lemma~\ref{lem:M=2 dist}. More generally, one might restrict to states $\ket{\Psi^{\,}_0}$ \eqref{eq:teleport initial mb state} that can be prepared using a FDQC $\uchan^{\,}_0$; however, $\uchan^{\,}_0$ can always be incorporated into $\uchan$ itself. The restriction to single-qubit measurements is motivated by experiment. 

The remaining conditions in Def.~\ref{def:standard teleport} potentially preclude physical teleportation protocols. However, to the best of our knowledge, all physical teleportation protocols known to the literature obey these restrictions as well. First, the classical calculations required for teleportation cannot be arbitrary difficult; in general, one should demand that a finite number of classical bits are required to determine the outcome-dependent recovery operations. Here, we restrict to protocols that use a minimal two classical bits per qubit, captured by the Clifford recovery operations of Def.~\ref{def:Clifford recovery}. We also preclude unitary operations on sites that have already been measured, as such operations (\emph{i}) cannot extract further information from the state and (\emph{ii}) cannot be useful to teleportation in the sense of enhancing the teleportation distance $\Dist$ \eqref{eq:L def}. This is a technical assumption that facilitates later proofs, which we generally expect can be relaxed. In Sec.~\ref{sec:outlook}, we comment on how assumption (\emph{ii}) of Def.~\ref{def:standard teleport} can be relaxed.

Given a teleportation protocol $\chan$ in canonical form \eqref{eq:chan canonical form}, we define an associated \emph{resource state} for $\chan$.

\begin{defn}
\label{def:resource state}
The \emph{resource state} $\ket{\Psi^{\,}_t}$ %associated with 
of a standard teleportation protocol $\chan$ with canonical form $\chan = \QECChannel \, \MeasChannel \, \uchan$ \eqref{eq:chan canonical form} applied to the initial product state $\ket{\Psi^{\,}_0 (\bvec{\alpha},\bvec{\beta})}$ \eqref{eq:teleport initial mb state} is
\begin{equation}
    \label{eq:resource state}
    \ket{\Psi^{\vpp}_t} = \uchan \, \ket{\Psi^{\vpp}_0} \, , ~~
\end{equation}
where $\uchan$ realizes unitary time evolution on $\Hilbert^{\,}_{\rm ph}$. 
\end{defn}

We note that the resource state $\ket{\Psi^{\,}_t}$ \eqref{eq:resource state} is only defined on the sites on which $\chan$ acts. For qubits on a graph $G$ in $d>1$ dimensions, a teleportation protocol $\chan$ acts on some path $\cal C$ connecting the initial and final logical site(s), in which case the resource state $\ket{\Psi^{\,}_t}$ \eqref{eq:resource state} refers only to the state of the qubits along the path $\mathcal{C} \subset G$.

As an aside, we do not distinguish between unitaries $\uchan$ generated by local quantum circuits versus time evolution under some local Hamiltonian. We only demand that $\uchan$ satisfies the Lieb-Robinson bound \eqref{eq:LR bound} of Theorem~\ref{thm:LR theorem}. The particular details of $\uchan$ are then encoded in the velocity $\LRvel$, where $\LRvel=1$ for circuits with nearest-neighbor gates.

\section{Constraints on teleportation}
\label{sec:teleport+locality}

We now derive constraints on the standard teleportation protocols $\chan$ stipulated in Def.~\ref{def:standard teleport}, making use of the %various 
results of Secs.~\ref{sec:overview} and \ref{sec:restricted teleportation}. We assume throughout that $\chan$ is written in canonical form \eqref{eq:chan canonical form} as detailed in Def.~\ref{def:chan canonical form}. In Sec.~\ref{subsec:meas+feedback} we consider the action of the dilated channels on logical operators, deriving an effective form of $\chan$ \eqref{eq:canonical Clifford form}. We then prove numerous results for the teleportation of a single logical qubit ($\numQ=1$), and explain how these proofs generalize to the teleportation of $\numQ>1$ qubits in Sec.~\ref{subsec:k>1 main}. 

\subsection{Action of dilated channels}
\label{subsec:meas+feedback}

Here we consider the action of  single-qubit measurements \eqref{eq:fluorescent measurement} and ``Clifford'' recovery operations \eqref{eq:unitary feedback teleport} on the logical operators \eqref{eq:logical alpha beta expval} in canonical-form protocols. We first derive an effective recovery channel for the logical operators $\LX$ and $\LZ$ in Lemma~\ref{lem:recovery factorization}. We then consider the combination of measurements and feedback on the logical operators in Lemma~\ref{lem:measure = attach}. Both proofs extend to any $\numQ \geq 1$.

We now present Lemma~\ref{lem:recovery factorization}, which establishes an effective recovery channel $\QECChannel^{\,}_{\rm eff}$ for any $\QECChannel$ in a canonical-form standard teleportation protocol (see Defs.~\ref{def:chan canonical form} and \ref{def:standard teleport}).

\begin{lem}[Effective recovery channel] 
    \label{lem:recovery factorization}
    Any standard teleportation protocol $\chan = \QECChannel \, \MeasChannel \, \uchan$ \eqref{eq:chan canonical form} that teleports a single logical qubit is equivalent to a protocol of the form
    \begin{equation}
    \label{eq:canonical Clifford form}
        \chan \cong \chan^{\vpd}_{\rm eff} = 
        \QECChannel^{\vpd}_x \, \QECChannel^{\vpd}_z \, \MeasChannel \, \uchan \, , ~~
    \end{equation}
    where $\QECChannel^{\,}_x$ ($\QECChannel^{\,}_z$) act  nontrivially only on $\LX$ ($\LZ$) as 
    \begin{subequations}
        \label{eq:effective recovery}
        \begin{align}
        \QECChannel^{\vpd}_{x} &= \ident \otimes \SSProj{0}{r^{\,}_{x}} + \LZ \otimes \SSProj{1}{r^{\,}_{x}}  \label{eq:effective recovery X} \\
        \QECChannel^{\vpd}_{z} &= \ident \otimes \SSProj{0}{r^{\,}_{z}} + \LX \otimes \SSProj{1}{r^{\,}_{z}} \label{eq:effective recovery Z}\, , ~~
    \end{align}
    \end{subequations}
    where $r^{\,}_{\nu} \subset \widetilde{V}$ is a subset of Stinespring qubits determined by $\QECChannel$ \eqref{eq:unitary feedback teleport}, the projectors correspond to outcome parities of the cluster $r^{\,}_{\nu}$ \eqref{eq:SSZ parity}, $\QECChannel^{\,}_x$ and $\QECChannel^{\,}_z$ effectively commute, and $\QECChannel^{\,}_{\nu} = \ident$ if $\QECChannel$ acts trivially on $\Pauli{\nu}{\rm L}$. 
\end{lem}

The proof of Lemma~\ref{lem:recovery factorization} appears in App.~\ref{app:recovery factorization proof}. The Lemma establishes (\emph{i}) when and how a Clifford recovery operation $\QECChannel^{\,}_s$ \eqref{eq:unitary feedback teleport} acts nontrivially on a logical operator $\Gamma$,  (\emph{ii}) that the effective recovery channel $\QECChannel^{\,}_{\rm eff} = \QECChannel^{\,}_x \, \QECChannel^{\,}_z$ \eqref{eq:canonical Clifford form} is equivalent to $\QECChannel$ in its action on all logicals $\Gamma$ (but not other operators), and (\emph{iii}) that the recovery channels for different logical operators can be written in any order, and do not mix. For $\numQ>1$ logical qubits, the pair of channels $\QECChannel^{\,}_{\nu}$ \eqref{eq:canonical Clifford form} are replaced by $2\numQ$ channels $\QECChannel^{\,}_{n,\nu}$ for $n \in [1,\numQ]$, in any order. We stress that both $\MeasChannel$ and $\uchan$ in $\chan$ \eqref{eq:canonical Clifford form} are unaltered compared to $\chan=\QECChannel \, \MeasChannel \, \uchan$ \eqref{eq:chan canonical form}.

We now state Lemma~\ref{lem:measure = attach}, which constrains the combined action of all dilated channels on logical basis operators. 

\begin{lem}[Constraints on dilated channels]
    \label{lem:measure = attach}
    Consider a standard teleportation protocol $\chan$ in canonical form (see Defs.~\ref{def:chan canonical form} and \ref{def:standard teleport}) involving the measurement of $\Nmeas$ (single-qubit) observables $\mobserv^{\,}_j$ \eqref{eq:CF observable}. Suppose that $\Gamma$ is a logical operator (e.g., $\LX$) acting on the final state $\ket{\Psi^{\,}_T}$ \eqref{eq:teleport final mb state}. Then $\chan$ fails to achieve teleportation \eqref{eq:teleportation conditions} \emph{unless} 
    \begin{equation}
    \label{eq:obs logical commute}
        \comm{\mobserv^{\vpp}_j}{\Gamma (t^{\vpp}_j)} = 0 \, , ~~
    \end{equation}
    for all observables $\mobserv^{\,}_j$ \eqref{eq:CF observable} and all logical operators $\Gamma$ evolved to the Heisenberg time $t^{\,}_j$ immediately prior to the measurement of $\mobserv^{\,}_j$. Moreover, if $\QECChannel$ acts nontrivially on $\Gamma$, then the measurement channel acts as
    \begin{align}
    \label{eq:meas attach}
        \Gamma (t^{\vpp}_{j-1}) = \umeas^{\dagger}_j \, \Gamma (t^{\vpp}_j) \, \umeas^{\vpd}_j \cong \bar{\mobserv}^{\vpp}_j \, \Gamma (t^{\vpp}_j) \, , ~~
    \end{align}
    and trivially otherwise, where $\bar{\mobserv}^{\,}_j$ is the involutory part  \eqref{eq:involutory part} of  $\mobserv^{\,}_j$ \eqref{eq:CF observable} and $\cong$ becomes a genuine equality upon projecting onto the default initial Stinespring state $\ket{0}$. 
\end{lem}

The proof of Lemma~\ref{lem:measure = attach} appears in App.~\ref{app:Proof meas = attach}. Simply put, Lemma~\ref{lem:measure = attach} extends Lemma~\ref{lem:recovery factorization} to the combination of a measurements and recovery operations acting on involutory logical operators $\Gamma \in \bvec{\Gamma} \subset \PauliGroup*{\rm ph}$. For a given $\Gamma$, this combination either acts trivially, violates the assumption that $\Gamma(t)$ remains a logical operator (i.e., involutory), or attaches the involutory part $\bar{\mobserv}^{\,}_j$ \eqref{eq:involutory part} of the measured operator $\mobserv^{\,}_j$ to $\Gamma (t^{\,}_j)$, where $t^{\,}_j$ is the (Heisenberg) time immediately prior to the $j$th measurement. Moreover,  $\mobserv^{\,}_j$ must commute with $\Gamma(t^{\,}_j)$. For $\MeasChannel^{\,}_j$ to act nontrivially on $\Gamma$, an odd number of recovery Paulis $\conjchan^{\,}_{j,k}$ conditioned on the outcome $j$ must \emph{anticommute} with the operators $\Gamma (\tau^{\,}_{j,k})$, where $\tau^{\,}_{j,k}$ is the (Heisenberg) time immediately prior to the channel $\QECChannel^{\,}_{j,k}$. This holds for each measured observable $\mobserv^{\,}_j$ and all logical operators $\Gamma$.

Note that Lemma~\ref{lem:measure = attach} holds for arbitrary (i.e., non-Clifford) $\uchan$ in the canonical-form protocol \eqref{eq:chan canonical form}. The proof of Lemma~\ref{lem:measure = attach} holds independently of any physical operations that act prior to the first measurement. Additionally, any physical unitary can be ``pulled through'' a measurement $\umeas^{\,}_j$ without affecting the result of Prop.~\ref{prop:single-qubit observable}, so that Lemma~\ref{lem:measure = attach} continues to hold. Likewise, one can safely ``pull through'' any recovery operations $\QECChannel$ that act prior to some measurement (in the Schr\"odinger picture), to realize a canonical-form protocol \eqref{eq:chan canonical form}.

Summarizing the above, the action of any canonical-form standard teleportation protocol $\chan$ on a logical operator $\Gamma \in \bvec{\Gamma}$ is constrained by Lemmas~\ref{lem:recovery factorization} and \ref{lem:measure = attach}. The former establishes that $\QECChannel$ is equivalent in its action on all $\Gamma \in \bvec{\Gamma}$ to the same protocol $\chan$ with $\QECChannel \to \QECChannel^{\,}_{\rm eff} = \prod_{n=1}^{\numQ} \, \QECChannel^{\,}_{n,x} \, \QECChannel^{\,}_{n,z}$ \eqref{eq:canonical Clifford form}, where $\conjchan^{\,}_{n,\nu}$ is the \emph{other} logical operator for qubit $n$ or the identity. The latter establishes that each measured observable $\mobserv^{\,}_j$ must commute with every logical operator $\Gamma(t^{\,}_j)$ at the Heisenberg time $t^{\,}_j$ immediately prior to $\MeasChannel^{\,}_j$, and that $\MeasChannel^{\,}_j$ attaches $\bar{\mobserv}^{\,}_j$ \eqref{eq:involutory part} to $\Gamma (t^{\,}_j)$ if the effective channel $\QECChannel^{\,}_{\Gamma}$ \eqref{eq:effective recovery} is nontrivial (i.e., $\conjchan^{\,}_{\Gamma} \neq \ident$).

\subsection{At least two measurements are required}
\label{subsec:M=2 sufficient}

We now prove that at least two measurements are needed to exceed the Lieb-Robinson bound \eqref{eq:LR bound}, as required by Def.~\ref{def:standard teleport} of standard teleportation. We first prove in Prop.~\ref{prop:M=1 useless} that any physical teleportation protocol $\chan$ (see Def.~\ref{def:phys teleport}) involving a \emph{single} measurement ($\Nmeas=1$) does not exceed the bound \eqref{eq:LR bound}: Even allowing for \emph{arbitrary} outcome-dependent recovery operations, the task distance $\Dist$ \eqref{eq:L def} realized by $\chan$ with $\Nmeas=1$ obeys the  bound $\Dist \leq \LRvel T$ \eqref{eq:LR bound} for $\Nmeas=0$.  Hence, a single measurement gives no enhancement compared to no measurements  at all. In Lemma~\ref{lem:compatible observables}, we
prove that two measurements (with corresponding Clifford recovery channels) are sufficient to realize $\Dist > \LRvel T$ in standard teleportation protocols.

\begin{prop}[A single measurement is insufficient]
\label{prop:M=1 useless}
    Consider a \emph{physical} teleportation protocol $\chan= \QECChannel \MeasChannel \uchan$ \eqref{eq:chan canonical form} in canonical form (see Defs.~\ref{def:phys teleport} and \ref{def:chan canonical form}), where $\uchan$ captures generic time evolution on $\Hilbert^{\,}_{\rm ph}$ \eqref{eq:physical Hilbert space} for duration $T$ and obeys the bound \eqref{eq:LR bound} with Lieb-Robinson velocity $\LRvel$; $\MeasChannel$  realizes exactly one projective measurement of some (effectively) involutory observable $\mobserv$ \eqref{eq:involutory part}; and $\QECChannel$ is an arbitrary recovery channel conditioned on the measurement's outcome. If the protocol $\chan$ teleports a single logical qubit a distance $\Dist$ in time $T$, then $\Dist \leq \LRvel T$. %, where $\LRvel$ and $T$ are determined solely by $\uchan$.
\end{prop}

The proof of Prop.~\ref{prop:M=1 useless} appears in App.~\ref{app:M=1 useless proof}. Prop.~\ref{prop:M=1 useless} establishes that a single projective measurement cannot lead to a speedup compared to a protocol $\chan'$ with no measurements at all. Essentially, Prop.~\ref{prop:teleportation conditions} is only  fulfilled with $\Dist > \LRvel T$ if the dilated channels attach distinct operators (via Lemma~\ref{lem:measure = attach}) to, e.g., the logical operators $\LX$ and $\LZ$; otherwise, $\LY$ can only travel a distance $\Dist \leq \LRvel T$.

We now consider protocols with two distinct measurements. We first present and prove Lemma~\ref{lem:compatible observables}, which establishes that a standard teleportation protocol $\chan$ with $\Nmeas=2$ \emph{can} teleport a single qubit a distance $\Dist > \LRvel T$, as well as the condition \eqref{eq:compatible observ} that must be satisfied. 

\begin{lem}[Compatible measurements]
\label{lem:compatible observables}
    Consider a standard teleportation protocol $\chan$ in canonical form $\chan = \QECChannel \, \MeasChannel \, \uchan$ (see Defs.~\ref{def:chan canonical form} and \ref{def:standard teleport}), where $\uchan$ realizes unitary time evolution on $\Hilbert^{\,}_{\rm ph}$ for duration $T$ and obeys the bound \eqref{eq:LR bound}, $\MeasChannel$ involves two projective measurements, and $\QECChannel$ consists of Clifford recovery operations (see Def.~\ref{def:Clifford recovery}). If $\chan$ teleports a logical qubit a distance $\Dist> \LRvel T$ \eqref{eq:L def}, then (i) $\chan$ must fulfill the criteria of Lemma~\ref{lem:measure = attach}, (ii) the effective recovery channels \eqref{eq:effective recovery} for the two logical basis operators \eqref{eq:logical alpha beta expval} must be conditioned on different measurement outcomes, and (iii) the observables must satisfy  
    \begin{equation}
        \label{eq:compatible observ}
         \comm{\mobserv^{\,}_1}{\mobserv^{\,}_2} = 0 \, ,~~
    \end{equation}
   where $\mobserv^{\,}_{1}$ and $\mobserv^{\,}_{2}$ are the \emph{effective} (involutory) operators \eqref{eq:CF observable} measured in the canonical-form protocol $\chan$ \eqref{eq:canonical Clifford form}.
\end{lem}

The proof of Lemma~\ref{lem:compatible observables} appears in App.~\ref{app:compatible obs}. That proof invokes many previous results, and particularly, the condition that $\Dist > \LRvel T$. In addition to proving that $\Nmeas=2$ is sufficient to realize $\Dist>\LRvel T$ \eqref{eq:L def}, the Lemma establishes a compatibility condition \eqref{eq:compatible observ} on the canonical-form observables $\mobserv^{\,}_{1,2}$ \eqref{eq:CF observable}. Denoting by $\mobserv_1'$ and $\mobserv_2'$  the single-qubit observables measured in the \emph{na\"ive} protocol $\chan'$ (which need not be in canonical form), we find
\begin{equation}
\label{eq:noncanon compatible observ}
    \comm{\mobserv_1^{\prime}}{\CliffDag{12} \, \mobserv_2^{\prime}  \, \Cliff{12}} = 0 \, ,~~
\end{equation}
if, in the Schr\"odinger picture, $\chan'$ involves the measurement of $\mobserv_1'$, followed by the physical unitary channel $\cliff^{\,}_{12}$, followed by the measurement of $\mobserv_2'$. Finally, we note that the two observables $\mobserv^{\,}_{1,2}$ (either in the na\"ive protocol $\chan'$ or in the canonical-form protocol $\chan$) must act on different sites, or they either fail to commute or fail to realize two distinct physical observables, so that Prop.~\ref{prop:M=1 useless} applies.

The compatibility condition \eqref{eq:compatible observ} makes sense when one considers the information that can be extracted from the resource state $\ket{\Psi_t}$ \eqref{eq:resource state}. Suppose that one measures $\PZ{j}$ on site $j$. Measuring $\PZ{j}$ again is gives the same outcome, while measuring $\PX{j}$ or $\PY{j}$ gives a random outcome;  in either case, no new information can be extracted. In the na\"ive protocol, the single-qubit measurements only extract information from the resource state if they act on distinct sites or if there are intervening multiqubit unitaries. For the canonical-form observables \eqref{eq:CF observable}, one expects that at most $n$ independent outcomes can be extracted from $n$ qubits via measurements, provided that the $n$ observables mutually commute.

However, it remains to understand \emph{how} a pair of compatible measurements can be useful to teleportation, as well as the extent of that utility in terms of exceeding the measurement-free Lieb-Robinson bound \eqref{eq:LR bound}. We now present Lemma~\ref{lem:M=2 dist}, which establishes the maximum enhancement to the teleportation distance $\Dist =d(i,f)$ \eqref{eq:L def} that a pair of compatible measurements \eqref{eq:compatible observ} can realize, using the fact that $\uchan$ obeys the bound \eqref{eq:LR bound}.

\begin{lem}[Speedup from two measurements]
\label{lem:M=2 dist}
    Consider a standard teleportation protocol $\chan$ (see Def.~\ref{def:standard teleport})  in canonical form $\chan = \QECChannel \, \MeasChannel \, \uchan$ \eqref{eq:canonical Clifford form}, where $\uchan$ is a local unitary (e.g., an FDQC) that obeys the bound \eqref{eq:LR bound} with Lieb-Robinson velocity $\LRvel$ and total duration (or depth) $T$, $\MeasChannel$ involves the measurement of two compatible, involutory observables $\mobserv^{\,}_1 \neq \mobserv^{\,}_2$, and $\QECChannel$ involves outcome-dependent operations that are useful in the sense of Lemma~\ref{lem:compatible observables}. Then the task distance $\Dist = d(i,f)$ \eqref{eq:L def} obeys the bound
    \begin{equation}
        \label{eq:M=2 speed limit}
        \Dist \leq \LRvel (3  T- 2)  - 2 \LRvel \, \max_\nu \tau^{\,}_\nu   \, ,~~
    \end{equation}
    where $\tau^{\,}_{\nu}$ is the total duration of time evolution following the measurement of $\mobserv_{\nu}'$ in the na\"ive protocol $\chan'$  (in the Schr\"odinger picture), and $T \geq 2$.  
\end{lem}
The proof of Lemma~\ref{lem:M=2 dist} appears in App.~\ref{app:M=2 speed limit}, and relies on the fact that the light cones of both measured observables must intersect with that of the final-state logical operator $\Pauli{\nu}{f}$ (at some site $\ell \approx i + 2 \Dist/3$). The intuition is that measurements ``reflect'' (or ``link up'')  light cones \cite{SpeedLimit}, leading to a threefold enhancement to $\LRvel$ in the $\Nmeas=2$ bound \eqref{eq:M=2 speed limit} compared to $\Dist \leq \LRvel T$ \eqref{eq:LR bound} for $\Nmeas < 2$.

\subsection{Additional measurement regions}
\label{subsec:k=1 daisy chain}

\begin{figure}[t]
\centering
\includegraphics[width=0.45\textwidth]{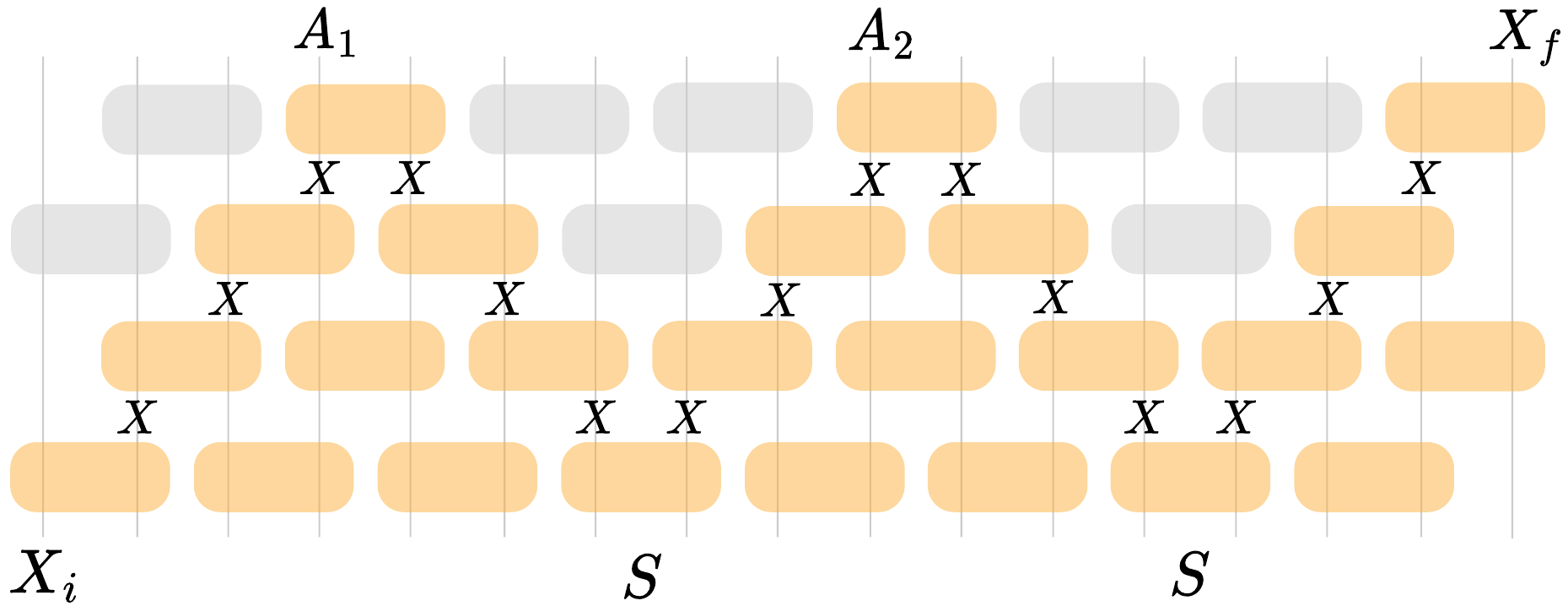}
\caption{The light cones of measured observables in canonical form are illustrated in a FDQC. The light cones of the logical operator and the measured observables $\mobserv^{\,}_j$ attached thereto via Lemma~\ref{lem:measure = attach} must intersect for teleportation to succeed.} 
\label{fig:daisy-changed regions}
\end{figure}

We next turn to standard teleportation protocols that achieve a task distance $\Dist > 3 \LRvel T$ \eqref{eq:L def}, which exceeds the bound \eqref{eq:M=2 speed limit} of Lemma~\ref{lem:M=2 dist} for $\Nmeas=2$. Theorem~\ref{thm:FYHL} \cite{SpeedLimit} suggests that this is possible if one uses additional measurements and accompanying feedback to link up the light cones of the additional measured observables attached to the logical operators, as depicted  in Fig. \ref{fig:daisy-changed regions}.

As a starting point, consider a protocol $\chan = \chan^{\,}_{\Nregions} \cdots \chan^{\,}_1$ in which $\Nregions$ teleportation protocols $\chan^{\,}_s$ are applied in series. Each protocol $\chan^{\,}_s$ satisfies $\Dist^{\,}_s=d(i_s,f_s) \leq \LRvel (3 \, T^{\,}_s -2)$ \eqref{eq:M=2 speed limit}, where $f_s=i_{s+1}$ is the initial site of the subsequent protocol, so that the total distance \eqref{eq:L def} is
\begin{align}
    \Dist &= d(i^{\,}_1,f^{\,}_{\Nregions}) = \sum\limits_{s=1}^{\Nregions} \Dist^{\vpp}_s \LRvel \, \sum\limits_{s=1}^{\Nregions} (3 \, T^{\vpp}_s - 2 )  \notag \\
    &\leq \LRvel (3 \, T - 2 \, \Nregions) < 3 \, \LRvel \, T \, ,~~ \label{eq:consecutive protocol dist}
\end{align}
where $T = \sum_s T^{\,}_s$ is the total depth. Unsurprisingly, applying protocols in series obeys the same bound \eqref{eq:M=2 speed limit}.

However, an improvement can be realized by applying the protocols in \emph{parallel}. Most of the channels in the protocol $\chan^{\,}_s$ do not overlap with those of $\chan^{\,}_{s-1}$. For convenience, suppose each protocol $\chan^{\,}_s$ has the same depth $T^{\,}_s = T'$. As explained following Lemma~\ref{lem:M=2 dist}, each protocol $\chan^{\,}_s$ uses $\Nmeas^{\,}_s=2$ measurements to link a unitary ``staircase'' region of size $\ell^{\,}_{s,{\rm uni}} \leq \LRvel T'$ and another measurement-assisted ``light-cone'' region of size $\ell^{\,}_{s,{\rm meas}} \leq 2 \LRvel (T'-1)$. 

Crucially, this perspective elucidates how to achieve parallelization. The unitaries that act in the light-cone regions can all be applied in parallel in time $T'$. The unitaries in each staircase region must be applied after all light-cone unitaries, for an additional unitary depth of $T'$, so that the final site $f_s$ of protocol $s$ is the initial site $i_{s+1}$ of protocol $s+1$.  After time $2T'$, all measurements and outcome-dependent operations are applied. Noting that the total unitary depth is $T=2 T'$, the task distance $\Dist$ \eqref{eq:L def} for the parallelized protocol obeys
\begin{align}
    \Dist^{\vpp}_{\rm par} %&\leq \Nregions\, \LRvel \left[ \frac{T}{2} + 2 \left( \frac{T}{2} -1 \right) \right] 
    &\leq \frac{3}{4} \LRvel \, \Nmeas \left( T - \frac{4}{3} \right) \, , ~~ \label{eq:parallelized k=1 speedup}
\end{align}
which exceeds the $\Nmeas=2$  bound \eqref{eq:M=2 speed limit} when $\Nmeas  > 4$.

However, it is possible to exceed the bound \eqref{eq:parallelized k=1 speedup} by constructing fully \emph{parallel} protocols. In particular, by omitting the staircase segments of each protocol $\chan^{\,}_s$ except $\chan^{\,}_1$, the total depth can be cut in half compared to the parallel\emph{ized} case. As a result, the upper bound \eqref{eq:parallelized k=1 speedup} is modified by sending $T \to 2 T$ and subtracting $(\Nregions-1) \LRvel T$. The resulting bound for a \emph{generic} standard teleportation protocol with $\Nmeas > 1$ is captured by Lemma~\ref{lem:M>2 speedup}.

\begin{lem}[Maximum speedup for $\Nmeas$ measurements] \label{lem:M>2 speedup}
    The distance $\Dist = d(i,f)$ \eqref{eq:L def} that a standard teleportation protocol $\chan = \QECChannel\, \MeasChannel \, \uchan$ (where $\uchan$ has depth $T$ and speed limit $\LRvel$) can transfer a single logical state $\ket{\psi}$ using $\Nmeas$ measurements and accompanying feedback obeys
    \begin{align}\label{eq:k=1 speed limit}
    \Dist \leq \LRvel T + 2\LRvel \left\lfloor \frac{\Nmeas}{2} \right\rfloor \left( T - 1 \right) 
    \, ,~~
    \end{align}
    where the $\LRvel T$ term captures the measurement-free distance and the second term captures the maximum enhancement due $\Nmeas$ measurements. We must have $T \geq 2$, as in Lemma~\ref{lem:M=2 dist}, and the separation $\ell$ between pairs of measurements obeys $\ell \leq 2 \LRvel (T-1)$.
\end{lem}

The proof of Lemma~\ref{lem:M>2 speedup} appears in App.~\ref{app:proof M>2 speedup}. The bound \eqref{eq:k=1 speed limit} is saturated by parallel protocols in which $\Nregions = \Nmeas/2$ pairs of measurements are spaced by a distance $\ell = 2 \LRvel (T-1)$ to daisy chain $\Nregions = \floor{\Nmeas/2}$ light-cone regions of size $\ell$ and a single staircase region of size $\ell^{\,}_0 \leq \LRvel T$. Note that Eq.~\ref{eq:k=1 speed limit} saturates the $\numQ=1$ bound \eqref{eq:FYHL k=1} \cite{SpeedLimit}, up to $\Order{1}$ offsets, while the teleportation bound \eqref{eq:FYHL bound} of Theorem~\ref{thm:FYHL} with $\numQ=1$ is loose by a factor of two compared to the $\numQ=1$ bound \eqref{eq:k=1 speed limit}, as conjectured in Ref.~\citenum{SpeedLimit}.

At this point, we note that the ``measurement regions'' discussed in Ref.~\citenum{SpeedLimit} are only unambiguously defined for optimal protocols. In that case, each pair of (adjacent) measurements used to daisy chain regions in the proof of Lemma~\ref{lem:M>2 speedup} (see App.~\ref{app:proof M>2 speedup}) defines a distinct measurement region (the support of the pair of operators). There  are $\Nmeas = \floor{\Nmeas/2}$ such regions in total, where the $s$th region has size $\LRvel + 1 \leq \ell^{\,}_s \leq 2 \LRvel (T-1)$. However, we note that (\emph{i}) the bounds herein are all phrased in terms of the number of measurement \emph{outcomes} and (\emph{ii}) an unambiguous delineation of measurement regions can be identified even in suboptimal protocols by considering string order, as we establish in Sec.~\ref{sec:SPT proof}.

\subsection{Multiqubit teleportation}
\label{subsec:k>1 main}

We now generalize the foregoing results to standard teleportation protocols $\chan$ (see Def.~\ref{def:standard teleport}) that transfer $\numQ \geq 1$ logical qubits. We first note that the various results of Sec.~\ref{subsec:meas+feedback} apply directly to protocols with $\numQ>1$ as stated. Next, the proofs of Prop.~\ref{prop:M=1 useless} and Lemma~\ref{lem:compatible observables} include proofs for $\numQ>1$ logical qubits. Essentially, $\Nmeas \geq 2\numQ$ are required, or else at least one logical qubit is teleported using a single outcome, which is impossible; moreover, \emph{all} measured observables must commute. Using these results, we now state the most general bound for standard teleportation protocols, captured by Theorem~\ref{thm:standard bound}.

\begin{thm}[Standard teleportation bound]
\label{thm:standard bound}
    Consider a standard teleportation protocol $\chan= \QECChannel \, \MeasChannel \, \uchan$ written in canonical form \eqref{eq:chan canonical form}, where $\uchan$ realizes unitary time evolution for duration $T$ with Lieb-Robinson velocity $\LRvel$ and $\MeasChannel$ consists of $\Nmeas>1$ projective measurements. Then $\chan$ can teleport $\numQ$ logical qubits a distance \eqref{eq:L def}
    \begin{align}
        \label{eq:k>1 bound}
        \Dist \leq \LRvel T + 2 \, \LRvel \left\lfloor \frac{\Nmeas}{2\numQ} \right\rfloor \left( T - 1 \right) \, ,~~
    \end{align}
    where $\floor{\cdot}$ is the floor function, and  $T$ obeys
    \begin{align}
        \label{eq:k>1 min depth}
        T \geq 1 + \numQ/\LRvel \, ,~~
    \end{align}
    where all measured observables $\mobserv^{\,}_j$ \eqref{eq:CF observable} commute with one another and with all final-state logical operators  \eqref{eq:teleport final mb state}. 
\end{thm}
The proof of Theorem~\ref{thm:standard bound} appears in App.~\ref{app:standard bound proof}. The bounds \eqref{eq:k>1 bound} and \eqref{eq:k>1 min depth} follow from identifying the optimal protocol described in the Theorem, which is further constrained in the proof in App.~\ref{app:standard bound proof}. Essentially, each of the $\numQ$ logical qubits can be treated individually, and all previous results apply, giving the bound \eqref{eq:k>1 bound}. The depth constraint \eqref{eq:k>1 min depth} follows from recognizing that there must be ample spacing between equivalent measurements to accommodate all $2 \numQ$ measurements required per region. 

The perspective is that a single region of size $\ell^{\,}_0 \leq \LRvel T$ can be daisy chained with $\Nregions$ regions of size $2 \numQ \leq \ell \leq 2 \LRvel (T-1)$ using measurements of $2\numQ$ sites. In the optimal case, the measurements are made on $2\numQ$ consecutive sites, where the sites $2n-1$ and $2n$ correspond to the $n$th logical qubit. In this scenario, the optimal spacings of initial and final sites are given by $i_n = i_1 + 2(n-1)$ and $f_n = f_1 + 2(n-1)$, respectively.

Finally, we remark that the teleportation bound \eqref{eq:k>1 bound} is tighter than corresponding bound \eqref{eq:FYHL bound} of Theorem~\ref{thm:FYHL} and Ref.~\citenum{SpeedLimit}. Essentially,  Eq.~\ref{eq:FYHL bound} is loose by a factor of two in the context of standard teleportation protocols, as initially conjectured in Ref.~\citenum{SpeedLimit}. While, in principle, it is possible that nonstandard teleportation protocols could exceed the teleportation bound \eqref{eq:k>1 bound}---and instead obey Eq.~\ref{eq:FYHL bound}---we note that (\emph{i}) both bounds have the same asymptotic scaling in terms of $\Dist$, $\Nmeas$, and $T$, which makes them equivalent in the context of such locality bounds, and (\emph{ii}) as far as we are aware, the types of channels that we have eschewed in our analysis of optimal protocols in App.~\ref{app:standard bound proof}---such as conditional measurements and non-Pauli recovery operations---cannot lead to violations of the standard bound \eqref{eq:k>1 bound}. Hence, we expect that Eq.~\ref{eq:k>1 bound} optimally bounds all physical teleportation protocols.

\subsection{Suboptimal measurements}
\label{subsec:suboptimal meas}

Before moving on to the main results in Sec.~\ref{sec:SPT proof}, we briefly discuss suboptimal measurements. There are a number of ways in which a standard teleportation protocol $\chan$ (see Def.~\ref{def:standard teleport}) can fail to be optimal in the sense of saturating the teleportation bound \eqref{eq:k>1 bound}. Of these, only protocols that involve more measurements than are necessary for a given teleportation distance $\Dist$ \eqref{eq:L def} have the potential to complicate the main results in  Sec.~\ref{sec:SPT proof}.

However, the requirement in Def.~\ref{def:standard teleport} that no unitaries be applied to qubits that have already been measured (in the Schr\"odinger picture) avoids this complication. The analysis of Sec.~\ref{sec:SPT proof} requires that all of the measurements that are useful to teleportation commute, and that the rest can be considered part of the preparation of the resource state (e.g., in the sense that there exists an equivalent protocol in which these measurements are replaced by unitaries). As we establish below, in Corollary~\ref{cor:meas comm} of Lemma~\ref{lem:measure = attach}, the aforementioned criterion of  Def.~\ref{def:standard teleport} ensures that all ``useful'' measured observables commute  with each other and with all logical operators, and that all other measurements can be omitted without consequence.

\begin{cor}[Necessary measurements commute]
    \label{cor:meas comm}
    The measurements that are necessary for teleportation correspond to observables that (i) commute with one another and (ii) have no support on the final logical sites $f \in F$. The remaining observables are (i) applied to sites that were already usefully measured, (ii) anticommute with that useful measurement and thus have a random outcome, (iii) are not attached to any logical operator via Lemma~\ref{lem:measure = attach}, and (iv) may simply be omitted from $\chan$.
\end{cor}

The proof of Corollary~\ref{cor:meas comm} is straightforward and appears in App.~\ref{app:proof meas comm}. As we note in Sec.~\ref{subsec:standard teleport} below Def.~\ref{def:standard teleport}, we expect that the result above holds generically, even if the requirement that no unitaries are applied to previously measured sites (in the Schr\"odinger picture) is relaxed.

\section{The resource state is an SPT}
\label{sec:SPT proof}

We now derive the main result: Standard teleportation protocols involve preparing an SPT resource state $\ket{\Psi_t} = \uchan \, \ket{\Psi_0}$ \eqref{eq:resource state} protected by a $\Ints^{\,}_2 \times \Ints^{\,}_2$ symmetry, measuring the corresponding string order parameters, and applying outcome-dependent Pauli gates. We briefly review SPTs in Sec.~\ref{subsec:SPT intro} before proving that the resource state $\ket{\Psi_t}$ for $\numQ=1$ teleportation has $\Ints^{\,}_2 \times \Ints^{\,}_2$  string order. We then extend this result to $\numQ > 1$ in Sec.~\ref{sec:k>1 spt}, where the protecting symmetry consists of (at least) $\numQ$ copies of $\Ints^{\,}_2 \times \Ints^{\,}_2$, each corresponding to pairs of sublattices. 

We stress that SPT phases are distinct from phases with ``instrinsic'' topological order like the toric code~\cite{toric_1998}. In some sense, the SPT phase can be viewed as a ``trivial'' topological order with no anyon content (i.e., a only a single, trivial anyonic excitation). Another important distinction is that SPT states are \emph{not} long-range entangled; as a result, they can be prepared using FDQCs~\cite{Chen2011}, making them promising resource states $\ket{\Psi_t}$~\eqref{eq:resource state}. Additionally, while it is \emph{possible} to teleport using a topologically ordered resource state---e.g., the ground state of the toric code~\cite{Herringer2023classificationof}---it is inefficient, because preparing the long-range-entangled resource state requires the same resources as teleportation itself~\cite{SpeedLimit}.

\subsection{Overview of SPTs}
\label{subsec:SPT intro}

We begin with a brief summary of 1D SPT phases and their salient features. For the SPTs of interest, the ``protecting symmetry'' corresponds to a group $\mathcal{G}$, whose generators $g \in \mathcal{G}$ have an \emph{on-site} representation
\begin{equation}
    \label{eq:on site}
    U(g) = \prod\limits_{j=1}^{\Nspins} \, u_j (g) \, ,~~
\end{equation}
where the $u_j(g)$ is a unitary representation of $g$ on site $j$. 

A 1D SPT phase is captured by a state $\ket{\Phi}$ that is symmetric under $U(g)$ \eqref{eq:on site}. Because the chain may be open or closed, the definition of a symmetric state must be independent of boundary conditions. We denote by $\interval \subseteq [1,\Nspins]$ a subinterval of the chain, and define $U_{\interval} (g)$ as the restriction of $U(g)$ \eqref{eq:on site} to the interval $\interval$. Then $\ket{\Psi}$ is symmetric under $U(g)$ \eqref{eq:on site} if, for every interval $\interval$ with $\abs{\interval} \gg \xi$ (where $\xi$ is the correlation length of $\ket{\Psi}$), there exist ``endpoint operators'' $V_L(g)$ and $V_R(g)$ localized to the left and right boundaries of $\interval$, respectively, such that
\begin{equation}
    \label{eq:sym state}
    V_L(g) \, U_{\interval} (g) \, V_R (g) \, \ket{\Psi} = \ket{\Psi} \, ,~~
\end{equation}
meaning that applying the on-site symmetry operators $u_j(g)$ in $U(g)$ \eqref{eq:on site} to some subregion $\interval$ only affects the symmetric state $\ket{\Psi}$ near the boundaries of $\interval$. Eq.~\ref{eq:sym state} can be shown rigorously using matrix product states \cite{MPSStringOrder}. 

The operators  $V_L(g) \, U_{\interval} (g) \, V_R (g)$ in Eq.~\ref{eq:sym state} are known as \emph{string order parameters} \cite{MPS1dDetect}. More generally, a string order parameter is any operator consisting of an on-site representation $u_j(g)$ \eqref{eq:on site} of $g \in \mathcal{G}$ in an interval $\interval$ decorated by charged operators $V$ at the endpoints \cite{MPS1dDetect}. Here, we can always choose endpoints such that the string order is \emph{perfect}, meaning it has expectation value one. While the operators $u_j(g)$ in $U(g)$ \eqref{eq:on site} realize a linear representation of $g \in \mathcal{G}$, the endpoint operators $V_{L/R}(g)$ only realize a \emph{projective} representation of $g \in \mathcal{G}$, i.e.,
\begin{subequations}
    \label{eq:proj rep}
    \begin{align}
        V_L (g) \, V_L (h) &= \omega (g,h) \, V_L (g h) \label{eq:left proj rep} \\
        V_R (g) \, V_R (h) &= \omega^{-1} (g,h) \, V_L (g h) \label{eq:right proj rep} \, ,~~
    \end{align}
\end{subequations}
for all $g,h \in \mathcal{G}$, where $\omega(g,h)\in \Unitary{1}$ is a complex phase. This phenomenon is known as \emph{symmetry fractionalization}, and is the basis for classifying SPT phases in 1D \cite{Chen2011, MPS_classification, MPS1dEntSpect}. 

For the finite Abelian groups $\mathcal{G}$ relevant herein, it is useful to define an additional complex phase $\Omega(g,h)$ via
\begin{equation}
    \label{eq:Abelian phase}
    V_L(g)\, V_L(h) = \Omega(g,h) \, V_L(h) \, V_L(g) \, , ~~
\end{equation}
where one can check that $\Omega (g,h) = \omega(g,h) \, \omega^{-1} (h,g)$. We now define \emph{nontrivial} string order using the phases $\Omega$.

\begin{defn}
\label{def:string order}
A state $\ket{\Psi}$ has \emph{nontrivial string order} with respect to an (Abelian) symmetry group $\mathcal{G}$ if, for \emph{every} subinterval $\interval \subseteq [1,\Nspins]$, there exist endpoint operators $V_L(g)$ and $V_R(g)$ localized to the boundary of $\interval$ satisfying $\Omega(g,h) \neq 1$ \eqref{eq:Abelian phase} for \emph{some} choice of $g,h \in \mathcal{G}$, along with
\begin{equation} \label{eq:string order def}
    \matel{\Psi}{\, V_L(g) \, U_{\interval}(g) \, V_R(g) \,}{\Psi} = 1 \, , ~~
\end{equation}
where $U_{\interval}(g)$ is the restriction of $U(g)$ \eqref{eq:on site} to $\interval$.
\end{defn}

The fact that $\Omega(g,h)\neq 1$ for some $g,h\in\mathcal{G}$ is necessary and sufficient for the projective representation to belong to a nontrivial class is shown in \cite{Berkovich1998}. One can show that nontrivial string order \eqref{eq:string order def} implies that (\emph{i}) any gapped Hamiltonian with symmetry group $\mathcal{G}$ and ground state $\ket{\Psi}$ has zero-energy \emph{edge modes} \cite{MPS1dEntSpect} and (\emph{ii}) the state $\ket{\psi}$ cannot be mapped to a product state via a FDQC consisting of local gates \emph{that respect the symmetry} \cite{Huang_2015}. Accordingly, nontrivial string order \eqref{eq:string order def} of $\ket{\Psi}$ is often regarded as synonymous with nontrivial SPT order of $\ket{\Psi}$ (or the parent Hamiltonian). However, since SPT order delineates a phase of matter, it is only well defined in the thermodynamic limit \cite{Zeng2015}. Hence, establishing SPT order more carefully requires considering a family of states $\ket{\Psi}$ defined on chains of increasing length, each of which possesses nontrivial string order in the sense of Def.~\ref{def:string order}.

\subsection{Single-qubit case: \texorpdfstring{$\Ints^{\,}_2 \times \Ints^{\,}_2$}{Z2 cross Z2}}
\label{subsec:k=1 spt}

We first prove that standard teleportation protocols $\chan$ that teleport a single logical qubit ($\numQ=1$), when written in canonical form \eqref{eq:chan canonical form}, correspond to preparing an SPT resource state $\ket{\Psi^{\,}_t}$ \eqref{eq:resource state}, measuring its string order parameters, and applying Pauli recovery gates to the final site $f$. In particular, the nontrivial string order of $\ket{\Psi^{\,}_t}$ \eqref{eq:resource state} is with respect to the discrete Abelian symmetry $\mathcal{G}=\Ints^{\,}_2 \times \Ints^{\,}_2$. The endpoint operators satisfy the Pauli algebra generated by $\PX{}$ and $\PZ{}$, which realize a projective representation of $\mathcal{G}$. As a first step toward proving that $\ket{\Psi^{\,}_t}$ \eqref{eq:resource state} has nontrivial string order in the sense of Def.~\ref{def:string order}, we first present Lemma~\ref{lem:k=1 big string order}, establishing the existence of order parameters \eqref{eq:string order def} when $\interval$ is the full chain.

\begin{lem}[End-to-end string order parameter]
\label{lem:k=1 big string order}
    Suppose that the standard teleportation protocol $\chan$ teleports $\numQ=1$ logical qubit from  $i$ to $f$ and has canonical form $\chan = \QECChannel \, \MeasChannel \, \uchan$ \eqref{eq:chan canonical form}. Then the resource state $\ket{\Psi^{\,}_t}$ \eqref{eq:resource state} has string order characterized by, e.g., the order parameters
    \begin{subequations}
    \label{eq:k=1 big string order}
    \begin{align}
        \matel{\Psi^{\vpp}_t}{ \, \uchan \, \PX{i} \, \uchan^{\dagger}\, \bar{\mobserv}^{\vpp}_x \, \PX{f} \, }{\Psi^{\vpp}_t} &= 1 \label{eq:k=1 big string order X} \\
        \matel{\Psi^{\vpp}_t}{\, \uchan\,  \PZ{i} \, \uchan^{\dagger} \, \bar{\mobserv}^{\vpp}_z \, \PZ{f} \, }{\Psi^{\vpp}_t} &= 1 \, , ~~ \label{eq:k=1 big string order Z} 
    \end{align}
    \end{subequations}
    where the bulk operators are given by
    \begin{equation}
        \label{eq:k=1 big string bulk}
        \bar{\mobserv}^{\vpp}_{\nu} \equiv \bar{\mobserv}_1^{\lambda^{\,}_{1,\nu}} \cdots \bar{\mobserv}_{\Nmeas}^{\lambda^{\,}_{\Nmeas,\nu}} \, ,~~ 
    \end{equation}
    i.e., the product of the involutory parts \eqref{eq:involutory part} of all \emph{necessary} measured observables $\bar{\mobserv}^{\,}_j$ (see Cor.~\ref{cor:meas comm}) attached to $\Pauli{\nu}{f}$ by Lemma~\ref{lem:measure = attach}. Hence, the string order parameters \eqref{eq:k=1 big string order} commute in the bulk region $(i + \LRvel T, f)$, i.e.,
    \begin{equation}
    \label{eq:k=1 big string bulk comm}
        \comm{\bar{\mobserv}^{\vpp}_{x}}{\bar{\mobserv}^{\vpp}_{z}} = 0  \, ,~~
    \end{equation}
    while the \emph{endpoint} operators $\uchan \, \Pauli{\nu}{i} \, \uchan^{\dagger}$ and $\Pauli{\nu}{f}$ anticommute for different $\nu = x,z$, as required by Def.~\ref{def:string order}. 
\end{lem}

The proof of Lemma~\ref{lem:k=1 big string order} appears in App.~\ref{app:proof k=1 big string order}, and follows straightforwardly from previous results. The end-to-end string order parameters are depicted graphically in Fig.~\ref{fig:string order}(a). The operators above $\uchan$ represent the two string order parameters \eqref{eq:k=1 big string order} acting on the resource state $\ket{\Psi^{\,}_t}$ \eqref{eq:resource state}, with the anticommuting endpoint operators color coded. In the Heisenberg picture, the order parameters evolve under $\uchan$ into the stabilizers $\StabEl^{\,}_{x,z} \in \UStabOf{\ket{\Psi^{\,}_0}}$ at the bottom of Fig.~\ref{fig:string order}(a), which have expectation value one \eqref{eq:sym state} acting on the initial state $\ket{\Psi^{\,}_0}$ \eqref{eq:teleport initial mb state}.

\begin{figure}[t]
\centering
\includegraphics[width=0.45\textwidth]{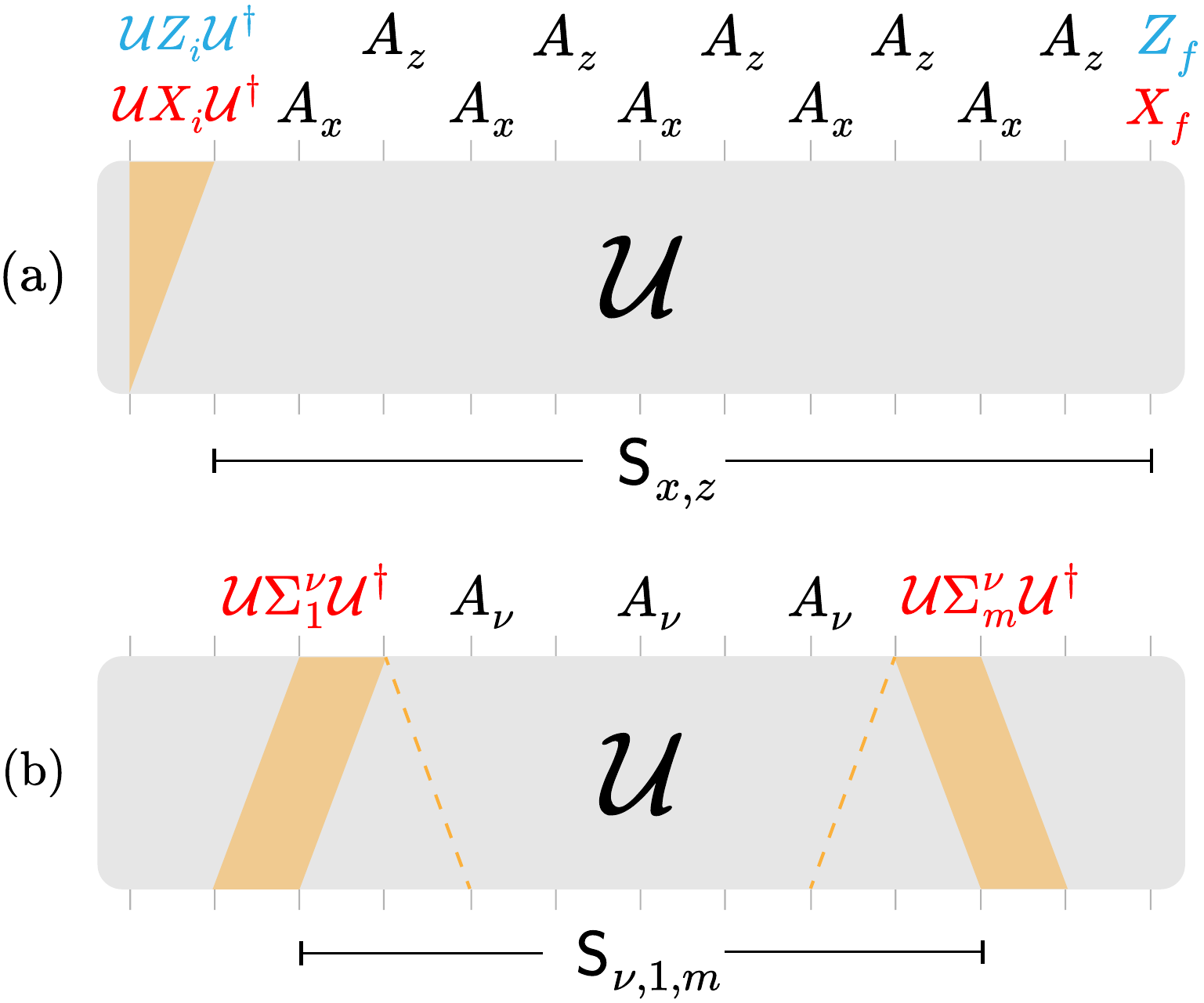}
\caption{The construction of $\Ints^{\,}_2 \times \Ints^{\,}_2$ string order parameters (a) on the endpoints and (b) in the bulk, where the colored operators represent the anticommuting endpoints.}
\label{fig:string order}
\end{figure}

However, in general, a state is only said to have string order \eqref{eq:string order def} if a similar relation holds for \emph{many} choices of endpoints. As it turns out, the string order parameters \eqref{eq:k=1 big string order} in Lemma~\ref{lem:k=1 big string order}, with right endpoint $f$ and left endpoint in $[i,i+\LRvel T]$, are not the only valid choice for the resource state $\ket{\Psi^{\,}_t}$ \eqref{eq:resource state}. In Lemma~\ref{lem:k=1 small string order}, we prove that a set of valid string order parameters \eqref{eq:string order def} corresponding to various subintervals $\interval \subset [i,f]$ can be identified.

\begin{lem}[$\Ints^{\,}_2 \times \Ints^{\,}_2$ string order]
    \label{lem:k=1 small string order}
    In addition to the end-to-end string order parameter \eqref{eq:k=1 big string order} of Lemma~\ref{lem:k=1 big string order}, the resource state $\ket{\Psi^{\,}_t}$ \eqref{eq:resource state} exhibits nontrivial string order (see Def.~\ref{def:string order}) captured by the order parameters
    \begin{equation}
        \label{eq:k=1 small string order}
        \mathcal{S}^{\vpp}_{\nu,a,b} = \left( \uchan \, \Sigma^{\nu}_a \, \uchan^\dagger \right) \, \bar{\mobserv}^{\vpp}_{\nu,a+1} \cdots \bar{\mobserv}^{\vpp}_{\nu,b} \, \left( \uchan \, \Sigma^{\nu}_{b} \, \uchan^{\dagger} \right) \, ,~~
    \end{equation}
    where the left and right endpoints are given by
    \begin{equation}
    \label{eq:k=1 small string new endpoint ops}
        \uchan \, \Pauli{\nu}{i} \, \uchan^\dagger \to \uchan \,  \Sigma^{\nu}_{a}  \, \uchan^\dagger ~~,~~\Pauli{\nu}{f} \to \uchan \Sigma^{\nu}_b \, \uchan^{\dagger} \, ,~~
    \end{equation}
    compared to Eq.~\ref{eq:k=1 big string order}, where the new endpoint operator is
    \begin{align}
    \label{eq:k=1 small string Sigma def}
        \Sigma^{\nu}_s &= \Pauli{\nu}{i} \, \uchan^{\dagger} \, \bar{\mobserv}^{\vpp}_{\nu,1} \cdots \bar{\mobserv}^{\vpp}_{\nu,s} \, \uchan = \Sigma^{\nu}_{s-1} \, \uchan^{\dagger} \, \bar{\mobserv}^{\vpp}_{\nu,s} \, \uchan \, ,~~   
    \end{align}
    where $\Sigma^{\nu}_0 = \Pauli{\nu}{i}$ and $\acomm{\Sigma^x_s}{\Sigma^z_s}=0$, as required. The $\Nregions$ operators $\bar{\mobserv}^{\,}_{s}$ correspond to products of the measured observables $\bar{\mobserv}^{\,}_j$ \eqref{eq:involutory part} in disjoint regions labeled $s$, such that $\uchan^{\dagger} \, \bar{\mobserv}^{\,}_{s+1} \, \uchan$ has no overlap with $\Sigma^{\nu}_{s}$ and $\Sigma^{\nu}_{s}$ acts to the right of $\Sigma^{\nu}_{s-1}$ and to the left of $\Sigma^{\nu}_{s+1}$. 
    
    Since, for all $a<b \in [1,\Nregions]$, $\ket{\Psi^{\,}_t}$ \eqref{eq:resource state} satisfies
    \begin{equation}
        \label{eq:k=1 string order}
        \matel{\Psi^{\vpp}_t}{\mathcal{S}^{\vpp}_{\nu,a,b}}{\Psi^{\vpp}_t}  = 1\, ,~~
    \end{equation}
    the resource state has nontrivial string order (see Def.~\ref{def:string order}). 
\end{lem}

The proof of Lemma~\ref{lem:k=1 small string order} appears in App.~\ref{app:k=1 smol string proof}. The Lemma shows that the resource state $\ket{\Psi^{\,}_t}$ \eqref{eq:resource state} has nontrivial string order captured by a set of order parameters \eqref{eq:k=1 small string order} corresponding to intervals $\interval \subseteq [i,f]$. Valid endpoints are labelled $s=[0,\Nregions+1]$, where $s=0$ corresponds to $j \approx i + \LRvel T$, $s=\Nregions+1$  to $j=f$, and all others  to the  $\Nregions$ \emph{nonoverlapping} ``measurement regions.'' 

The string order parameters \eqref{eq:k=1 small string order} are illustrated in Fig.~\ref{fig:string order}(b). The operator $\Sigma^{\nu}_1$ is guaranteed to have no overlap in support with $\uchan^{\dagger} \bar{\mobserv}^{\,}_{\nu,2}$---the observable immediately to the right of the left endpoint operator in Fig.~\ref{fig:string order}(b)---by Lemma~\ref{lem:k=1 small string order}. This can be ensured by multiplying the ``na\"ive'' choice of $\Sigma^{\nu}_1$ by an initial-state stabilizer $\StabEl^{\,}_{\nu,1}$ to remove any support in triangle defined by the parallelogram and dashed line. The same holds for the right endpoint operator. We also point out that $\Sigma^{\nu}_1$ has no overlap with $\Pauli{\nu}{i}$, and that this extends to any choice of interval $\interval=[a,b]$, where the endpoint operators of the corresponding string order parameter $\mathcal{S}^{\,}_{\nu,a,b}$ \eqref{eq:k=1 small string order} have no overlap with those of $\mathcal{S}^{\,}_{\nu,a\pm 1,b}$ nor $\mathcal{S}^{\,}_{\nu,a,b\pm 1}$ by Lemma~\ref{lem:k=1 small string order}.

The string order parameters \eqref{eq:k=1 small string order} indicate that  $\ket{\Psi^{\,}_t}$ \eqref{eq:resource state} is symmetric \eqref{eq:sym state} with respect to a group $\mathcal{G}=\Ints^{\,}_2 \times \Ints^{\,}_2$, which is projectively represented by the Pauli matrices $\{\ident, \PX{}, \PY{}, \PZ{} \}$, so that $\Omega(x,z) = -1$ \eqref{eq:Abelian phase}. In the bulk, $\mathcal{G}$ has a linear (unitary) representation
\begin{equation}
    \label{eq:k=1 pseduo on site}
    U_{\interval} (\nu) = \prod\limits_{s \in \interval} \, \bar{\mobserv}^{\vpp}_{\nu,s} \, ,~~
\end{equation}
where $\nu \in \{x,z\}$ and $\comm{U_{\interval}(x)}{U_{\interval}(z)}=0$ by Cor.~\ref{cor:meas comm}, which is \emph{on site} for a sufficiently large unit cell. 

However, for nontrivial string order \eqref{eq:k=1 string order} to imply an SPT, one must show that the string order corresponds to a \emph{phase of matter}. This requires that one can extend Eqs.~\ref{eq:k=1 small string order} and \ref{eq:k=1 string order} to the thermodynamic limit (where $\abs{\interval} \to \infty$), which we establish in the following Corollary.

\begin{cor}[Embedding the resource state]
\label{cor:k=1 embed}
    The resource state $\ket{\Psi^{\,}_t}$ \eqref{eq:resource state} can be embedded in a thermodynamically large (i.e., infinite) chain, and the string order \eqref{eq:k=1 string order} extended to arbitrarily large intervals by implementing the truncation procedure of Lemma~\ref{lem:k=1 small string order} in reverse.
\end{cor}

Corollary \ref{cor:k=1 embed} follows from Lemma~\ref{lem:k=1 small string order} and its proof in App.~\ref{app:k=1 smol string proof}. The Corollary establishes that one can use the resource state $\ket{\Psi^{\,}_t}$ \eqref{eq:resource state} to define a family of states of unbounded length, each of which has perfect string order, meaning that the string order survives into the thermodynamic limit. This also establishes that $\ket{\Psi^{\,}_t}$ \eqref{eq:resource state} \emph{cannot} be prepared in finite time $t \ll d(i,f)$ under local unitary time evolution (with $\LRvel = \Order{1}$ finite) that respects the symmetry $\mathcal{G} = \Ints^{\,}_2 \times \Ints^{\,}_2$ \cite{Huang_2015}. Thus, the resource state \eqref{eq:resource state} for \emph{any} $\numQ=1$ standard teleportation protocol $\chan$ belongs to a nontrivial SPT phase indicated by $\Ints^{\,}_2 \times \Ints^{\,}_2$ string order, where the measured operators $\mobserv_j$ used to determine $\QECChannel_{x/z}$ \eqref{eq:effective recovery} are those that make up the bulk of the $x/z$ string order parameters \eqref{eq:k=1 small string order}.

\subsection{Multiqubit case: \texorpdfstring{$\Ints_2^{2\numQ}$}{k Z2s (Zs 2?)}}
\label{sec:k>1 spt}

We now show that the results of Sec.~\ref{subsec:k=1 spt} for a single logical qubit extend straightforwardly to the standard teleportation of $\numQ >1$ logical qubits. This result is established by Theorem~\ref{thm:spt order}, which we state and prove below.

\begin{thm}[Standard teleportation implies SPT order]
    \label{thm:spt order}
    Consider a protocol with canonical form $\chan= \QECChannel \, \MeasChannel \, \uchan$ \eqref{eq:chan canonical form} that achieves standard teleportation of $\numQ \geq 1$ qubits (see Def.~\ref{def:standard teleport}). Omitting any unnecessary measurements from $\MeasChannel$ (see Cor.~\ref{cor:meas comm}), we identify $\Nregions$ \emph{measurement regions}---where $\bar{\mobserv}^{\,}_{n,\nu,s}$ is a product of all measured observables $\bar{\mobserv}^{\,}_j$ \eqref{eq:involutory part} in region $s$ attached to $\Pauli{\nu}{f_n}$ \eqref{eq:meas attach}---such that $\bar{\mobserv}^{\,}_{n,\nu,s+1}$ has no overlap with the \emph{endpoint} operator
    \begin{align}
        \Sigma^{\nu}_{n,s} &= \Pauli{\nu}{i_n} \, \uchan^{\dagger} \, \bar{\mobserv}^{\vpp}_{n,\nu,1} \cdots \bar{\mobserv}^{\vpp}_{n,\nu,s} \, \uchan \notag \\
        &= \Sigma^{\nu}_{n,s-1} \, \uchan^{\dagger} \, \bar{\mobserv}^{\vpp}_{n,\nu,s} \, \uchan \, ,~~ 
        \label{eq:k>1 endpoint def}
    \end{align}
    where $\Sigma^{\nu}_{n,0}=\Pauli{\nu}{i_n}$ and $\Sigma^{\nu}_{n,\Nregions+1}=\uchan^{\dagger} \, \Pauli{\nu}{f_n} \, \uchan$. Then, for all $n \in [1,\numQ]$,  $\nu=x,y,z$, and $a<b \in [1,\Nregions]$, the operator
    \begin{align}
    \label{eq:k>1 string order parameter}
         \mathcal{S}^{n,\nu}_{a,b}  \equiv \uchan \, \Sigma^{\nu}_{n,a} \, \uchan^{\dagger} \, \bar{\mobserv}^{\vpp}_{n,\nu,a+1} \cdots  \bar{\mobserv}^{\vpp}_{n,\nu,b} \,\uchan \, \Sigma^{\nu}_{n,b} \, \uchan^{\dagger} \, ,
    \end{align}
    with anticommuting endpoint operators $\Sigma^{\nu}_{n,s}$ \eqref{eq:k>1 endpoint def} satisfies
    \begin{align}
    \label{eq:k>1 string order}
        \matel*{\Psi^{\vpp}_t}{\,\mathcal{S}^{n,\nu}_{a,b} \,}{\Psi^{\vpp}_t} = 1  \, ,~~
    \end{align}
    for the resource state $\ket{\Psi^{\,}_t}$ \eqref{eq:resource state}, indicating \emph{nontrivial string order} (see Def.~\ref{def:string order}) with respect to the symmetry
    \begin{equation}
        \label{eq:protecting symmetry}
        \mathcal{G}_{\numQ} = \left( \Ints^{\vpp}_2 \times \Ints^{\vpp}_2 \right)^{\numQ} = \Ints^{2\numQ}_2 \, , ~~
    \end{equation}
    so that $\uchan$ prepares a resource state $\ket{\Psi^{\,}_t}$ \eqref{eq:resource state} with nontrivial SPT order, $\MeasChannel$ measures its end-to-end \eqref{eq:k=1 big string order} string order parameters $\mathcal{S}^{n,\nu}_{0,\Nregions+1}$ \eqref{eq:k>1 string order parameter}, and $\QECChannel$ corrects the errors $\conjchan$, which belong to a projective representation of $\mathcal{G}_{\numQ}$ \eqref{eq:protecting symmetry}.
\end{thm}

\begin{proof}
    The case for $\numQ=1$ is proven in Lemmas~\ref{lem:k=1 big string order} and \ref{lem:k=1 small string order}. By definition, the logical operators for different logical qubits commute; by Cor.~\ref{cor:meas comm}, all \emph{necessary} measured observables commute for any $\numQ$, while all other measurements have been omitted without affecting $\chan$. Thus, Lemmas~\ref{lem:k=1 big string order} and \ref{lem:k=1 small string order} hold for each logical qubit \emph{individually} (i.e., independently): The preceding arguments guarantee that the  string order parameters \eqref{eq:k>1 string order parameter} for distinct logical qubits \emph{locally} commute. Moreover, Cor.~\ref{cor:k=1 embed} extends to $\numQ>1$ without caveat. As a result, $\ket{\Psi^{\,}_t}$ \eqref{eq:resource state} belongs to a nontrivial SPT phase with $\Ints^{2\numQ}_2$ string order. 
\end{proof}

Theorem~\ref{thm:spt order} is the main result of this paper. In the optimal case, the $\Nmeas=2\numQ \Nregions$ measurements are applied to the $2\numQ$ sublattices of the chain at the endpoints of each valid subinterval $\interval$ described in Sec.~\ref{subsec:SPT intro}. The
same resource state $\ket{\Psi^{\,}_t}$ \eqref{eq:resource state} can be used to teleport logical states $\ket{\psi_n}$ between the endpoints of \emph{any} such interval $\interval$; by Cor.~\ref{cor:k=1 embed}, $\ket{\Psi^{\,}_t}$ can also be embedded in a thermodynamically large chain with nontrivial string order on infinitely many such endpoints. Hence, the resource state $\ket{\Psi^{\,}_t}$ \eqref{eq:resource state} belongs to a nontrivial SPT phase protected by the symmetry $\mathcal{G}_{\numQ}$ \eqref{eq:protecting symmetry}, where logical states are teleported between the edge modes of open intervals on the chain by measuring the string order parameters \eqref{eq:k>1 string order parameter} on the intervals $[i_n,f_n)$.

This also holds even for suboptimal protocols that realize standard teleportation. In particular, while using more measurements are made than are needed to achieve teleportation distance $\Dist$ \eqref{eq:L def} complicates the delineation of valid intervals (i.e., endpoints), they are guaranteed to exist by Theorem~\ref{thm:spt order} (and Lemmas~\ref{lem:k=1 big string order} and \ref{lem:k=1 small string order}). Hence, even  suboptimal standard teleportation protocols $\chan$ are equivalent to preparing an SPT resource state $\ket{\Psi^{\,}_t}$ \eqref{eq:resource state} protected by $\mathcal{G}=\Ints_2^{2\numQ}$ and  measuring its $2 \, \numQ$ string order parameters \eqref{eq:k>1 string order parameter} on the $\numQ$ intervals $[i_n,f_n)$.

\section{Example teleportation protocols}
\label{sec:examples}

We now describe several examples that illustrate our main result. In Sec.~\ref{subsec:cluster-y}, we consider a $\numQ=1$ teleportation protocol whose resource state is the same cluster state that appears in Sec.~\ref{subsec:phys teleport}, and we explain how $\PY{}$-basis measurements are nonetheless compatible with Theorem~\ref{thm:spt order}. In Sec.~\ref{subsec:hypergraph} we explain how \emph{hypergraph} states \cite{hypergraph, hypergraphMiyake, hypergraphMiyake2D} are a viable resource for $\numQ=1$ teleportation, despite not generally being recognized as SPTs and being nonstabilizer states \cite{gottesman1997stabilizer}. In Sec.~\ref{subsec:k>1 cluster}, we generalize the cluster states of Secs.~\ref{subsec:phys teleport} and \ref{subsec:cluster-y} to $\numQ\geq 1$, and describe the protecting symmetries and string order parameters. Finally, in Sec.~\ref{subsec:k>1 estp}, we present a family of optimal resource states for standard teleportation of $\numQ \geq 1$ qubits and discuss the reelvant symmetries and order parameters.

\begin{figure*}[t!]
\centering
\includegraphics[width=0.8\textwidth]{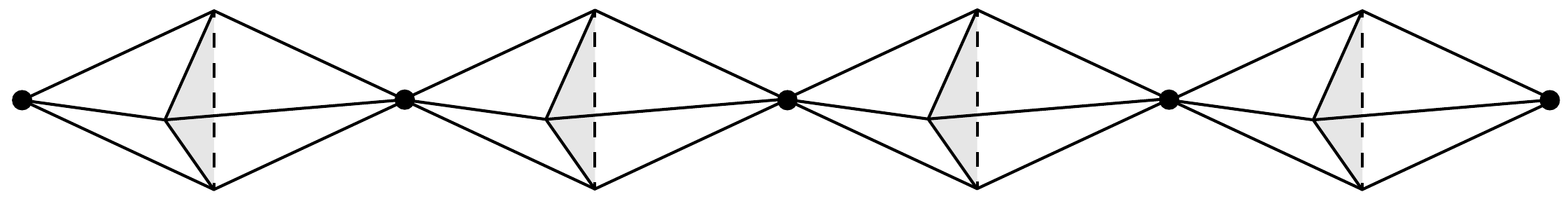}
\caption{A hypergraph lattice with qubits assigned to the $\Nspins=17$ vertices. The hypergraph state $\ket{\Psi^{\,}_t}$ recovers from $\ket{\Psi^{\,}_0} = \ket{+++ \cdots}$ \eqref{eq:hypergraph initial state}, where $\PXp{j}{} \ket{+} = \ket{+}$, by applying CCZ gates to every unshaded triangular face via $\uchan$ \eqref{eq:hypergraph uchan}.}
\label{fig:hypergraph chain}
\end{figure*}

\subsection{Alternative measurements on the cluster state}
\label{subsec:cluster-y}

We first consider an extension of the cluster state example of Sec.~\ref{subsec:phys teleport} \cite{Briegel_2001}. However, instead of measuring in the $\PX{}$ basis (as in Fig.~\ref{fig:cluster state circuit}), one can instead measure in the $\PY{}$ basis \cite{dumbqc}. Theorem~\ref{thm:spt order} then implies that the cluster state has nontrivial string order with respect to a unitary symmetry \eqref{eq:on site} that is diagonal in the on-site $\PY{}$ basis. We now show that this is, indeed, the case, and identify an alternative measurement pattern for $\PY{}$-basis measurements compatible with cluster-state teleportation.

The protocol with $\PY{}$-basis measurements is similar to the $\PX{}$-basis protocol in Sec.~\ref{subsec:phys teleport}. In fact, both protocols have the same physical unitary channel $\uchan$ \eqref{eq:cluster CZ circ}. To understand the measurement and recovery channels, it is helpful to consider the Heisenberg evolution of the logical operators. In analogy to Eq.~\ref{eq:cluster premeasurement logical}, we require that the logical operators acting on the cluster state $\ket{\Psi^{\,}_t}$ \eqref{eq:resource state} are
\begin{subequations}
\label{eq:cluster v2 premeas logicals}
\begin{align}
    \MeasChannel^\dagger \, \QECChannel^\dagger \, \PX{f} \, \QECChannel \, \MeasChannel &= Y Y \ident \, Y Y \ident \, \cdots Y Y \ident \,\PX{f} \label{eq:cluster v2 premeas logical X} \\
    \MeasChannel^\dagger \, \QECChannel^\dagger\,  \PZ{f} \, \QECChannel\,  \MeasChannel &= \ident \, Y Y \ident \, Y Y \cdots \ident \, Y Y \, \PZ{f} \, ,~\label{eq:cluster v2 premeas logical Z}
\end{align}
\end{subequations}
for chains of length $\Nspins= 3 \Nregions+1$ (though it is possible to modify the protocol for other $\Nspins$), where we have projected onto the default initial Stinespring state $\ket{\bvec{0}}$. Essentially, one measures all sites except $f$ in the $\PY{}$ basis, where ``measurement regions'' involve \emph{three} neighboring sites.

Following Sec.~\ref{subsec:phys teleport}, we again identify the stabilizers $\StabEl^{\,}_x$ and $\StabEl^{\,}_z$ \eqref{eq:cluster string ops} using Eq.~\ref{eq:teleportation conditions} of Prop.~\ref{prop:teleportation conditions} to ensure successful teleportation. Then Eq.~\ref{eq:cluster v2 premeas logicals} implies that 
\begin{subequations}
\label{eq:cluster v2 string ops}
\begin{align}
    \uchan \, \StabEl^{\vpp}_x \, \uchan^\dagger &= \PZ{1} \, \PX{2} \, (\ident \, Y Y \ident \, Y Y \cdots \ident \,) \, \PX{\Nspins}  \label{eq:cluster v2 string op X} \\
    \uchan \, \StabEl^{\vpp}_z \, \uchan^\dagger &= \PZ{1} \, \PY{2} (Y \ident \, Y Y \ident \, Y \cdots Y) \, \PZ{\Nspins} \label{eq:cluster v2 string op Z} \, ,~
\end{align}
\end{subequations}
which are indeed stabilizer elements for the cluster state.

Moreover, the operators $\StabEl^{\,}_x$ and $\StabEl^{\,}_z$ \eqref{eq:cluster v2 string ops} realize end-to-end string order parameters \eqref{eq:k=1 big string order}. These two operators commute in the bulk but have anticommuting endpoints. Since they have expectation value one in the cluster state, they imply that the cluster state also has nontrivial string order with respect to the symmetry $\mathcal{G} = \Ints^{\,}_2 \times \Ints^{\,}_2$ whose (bulk) generators are $Y \ident \, Y Y \ident \, Y \cdots$ and $\ident \, Y Y \ident \, Y Y \cdots$. 

Hence, we identify measurement regions as having size three, since this is the minimal translation-invariant unit cell of the symmetry generators. Finally, we note that the number of local measurements is minimized if one measures $\PY{1}\PY{2}$ and $\PY{2}\PY{3}$ in each region; however, one can instead measure only $\PY{1}$ and $\PY{3}$, at the cost of applying an additional FDQC after $\uchan$ \eqref{eq:cluster CZ circ}. These single-qubit measurements are compatible with Def.~\ref{def:phys teleport}, and more similar to the $\PX{}$-basis protocol discussed in Sec.~\ref{subsec:phys teleport}, but the measurement regions still have size three.

\subsection{Hypergraph states}
\label{subsec:hypergraph}

We present a $\numQ=1$ standard teleportation protocol involving a hypergraph \cite{hypergraph, hypergraphMiyake, hypergraphMiyake2D} resource state $\ket{\Psi^{\,}_t}$ \eqref{eq:resource state}. Because $\UStabOf*{\ket{\Psi^{\,}_t}}  \nsubseteq \PauliGroup*{\Nspins}$ for the hypergraph state $\ket{\Psi^{\,}_t}$, this example highlights how our results extend beyond ``stabilizer states'' \cite{gottesman1997stabilizer}. The hypergraph state $\ket{\Psi^{\,}_t}$ is defined on the vertices $v \in V$ of a graph $G=(V,E)$ comprising connected tetrahedra, as depicted in Fig.~\ref{fig:hypergraph chain}. The initial state is the same as for the cluster state \eqref{eq:cluster initial state}
\begin{align}
\label{eq:hypergraph initial state}
    \ket{\Psi^{\vps}_0} = \ket{\psi}^{\vpp}_i \otimes \ket{+++\cdots} \otimes \ket{\bvec{0}}^{\vpp}_{\rm ss} \, ,
\end{align}
with the logical state $\ket{\psi}$ initialized on the leftmost site $i$ of Fig.~\ref{fig:hypergraph chain}. The entangling unitary $\uchan$ is given by
\begin{align}
\label{eq:hypergraph uchan}
    \uchan = \prod_{\Delta \in \mathcal{F}_u} \mathrm{CCZ}^{\vps}_{\Delta} \, ,
\end{align}
where $\Delta$ runs over all \emph{unshaded} faces (the set $\mathcal{F}_u$) of the graph in Fig.~\ref{fig:hypergraph chain}, and the controlled-controlled-$\PZ{}$ gate $\mathrm{CCZ}^{\,}_{\Delta}$ applies CZ to two of the qubits on the face $\Delta$ if the third is in the $\PZ{}$-basis state $\ket{1}$. In the computational ($\PZ{}$) basis, $\mathrm{CCZ}^{\,}_{\Delta} = \mathrm{diag}(1,1,1,1,1,1,1,-1)$ sends $\ket{111}^{\,}_{\Delta} \to -\ket{111}^{\,}_{\Delta}$, and acts trivially otherwise.

We note that a length-three version of this hypergraph state appears in Ref.~\citenum{hypergraphMiyake}, but its connection to SPT order was not discussed. The hypergraph state in Fig.~\ref{fig:hypergraph chain} can be viewed as a cluster state \cite{Briegel_2001} whose odd sites correspond to the vertices connecting the tetrahedra (noted by dots in Fig.~\ref{fig:hypergraph chain}), and whose even sites are encoded in the three qubits around the shaded faces in Fig.~\ref{fig:hypergraph chain}.

The CCZ gate corresponds to the third order of the Clifford hierarchy \cite{CliffordHierarch}, where the first order is the Pauli group $\PauliGroup*{\Nspins}$, the second order is the Clifford group $\CliffGroup*{\Nspins}$, which maps  $\PauliGroup*{\Nspins}$ to itself, while gates like CCZ map $\CliffGroup*{\Nspins}$ to itself. In particular, for three vertices $i,j,k$, we have
\begin{align}
\label{eq:CCZ action on X}
    \mathrm{CCZ}^{\vps}_{ijk}\, \PX{i} \,\mathrm{CCZ}^{\vps}_{ijk} = \PX{i} \,\mathrm{CZ}^{\vps}_{jk} \, ,~
\end{align}
which is an element of $\CliffGroup*{\Nspins}$, and the same holds for$\PX{j}$ and $\PX{k}$ by the permutation symmetry of CCZ.

The hypergraph teleportation protocol $\chan$ is straightforward. After applying the entangling unitary $\uchan$ \eqref{eq:hypergraph uchan} to the initial state $\ket{\Psi^{\,}_0}$ \eqref{eq:hypergraph initial state}, one next measures all qubits (except the rightmost site $f$) in the $\PX{}$ basis in the channel $\MeasChannel$. Thus, the measurement regions in Fig.~\ref{fig:hypergraph chain} correspond to the tetrahedra formed by the shaded triangular faces and the vertex to the left. In the channel $\QECChannel$, the parity of the outcomes of all $\PX{j}$ measurements for qubits $j \in \mathcal{F}_s$ in the \emph{shaded} triangles in Fig.~\ref{fig:hypergraph chain} determine a final $\conjchan^{\,}_z = \PX{f}$ rotation, while the parity of the $\PX{j}$ outcomes for all other qubits $j \in V \setminus \mathcal{F}_s$ determine a final $\conjchan^{\,}_x=\PZ{f}$ rotation. 

In the Heisenberg picture, the two logical operators $\LX$ and $\LZ$ evolve under the dilated channels $\QECChannel$ and $\MeasChannel$ via
\begin{subequations}
\label{eq:hypergraph logical resource}
\begin{align}
    \PXp{\rm L}{\prime} &= \trace_{\rm ss} \left[ \, \MeasChannel^\dagger \, \QECChannel^\dagger\,  \PX{f} \, \QECChannel\,  \MeasChannel \, \SSProj{\bvec{0}}{\rm ss} \, \right] = \prod_{j \in V \setminus \mathcal{F}_s} \PX{j} \label{eq:hypergraph logical resource X} \\
    \PZp{\rm L}{\prime} &= \trace_{\rm ss} \left[ \, \MeasChannel^\dagger \, \QECChannel^\dagger \, \PZ{f} \, \QECChannel\,  \MeasChannel \, \SSProj{\bvec{0}}{\rm ss} \,  \right] = \PZ{f} \prod_{\Delta \in \mathcal{F}_s} \PX{\Delta}\, ,\label{eq:hypergraph logical resource Z} 
\end{align}
\end{subequations}
acting on the hypergraph resource state $\ket{\Psi^{\,}_t}$, where $\mathcal{F}_s$ denotes the \emph{shaded} faces in Fig.~\ref{fig:hypergraph chain}, the set $V \setminus \mathcal{F}_s$ corresponds to the filled dots in Fig.~\ref{fig:hypergraph chain}, and $\PX{\Delta} = \prod_{j \in \Delta} \, \PX{j}$. 

To show that the protocol $\chan$ satisfies Eq.~\ref{eq:teleportation conditions} of Prop.~\ref{prop:teleportation conditions}, we evolve through the unitary $\uchan$ \eqref{eq:hypergraph uchan} to find
\begin{subequations}
\label{eq:hypergraph logical initial}
\begin{align}
    \PX{f} (T) &= \trace_{\rm ss} \left[ \, \chan^{\dagger} \, \PX{f} \, \chan \, \SSProj{\bvec{0}}{\rm ss} \, \right] = \prod_{j \in V \setminus \mathcal{F}_s} \PX{j} \label{eq:hypergraph logical initial X} \\
    \PZ{f} (T) &= \trace_{\rm ss} \left[ \, \chan^{\dagger} \, \PZ{f} \, \chan \, \SSProj{\bvec{0}}{\rm ss} \, \right] =  \PZ{i}  \prod_{\Delta \in \mathcal{F}_s} \PX{\Delta}  \, ,\label{eq:hypergraph logical initial Z}
\end{align}
\end{subequations}
where $\uchan$ acts trivially on all $\PX{}$ operators. Essentially, $\uchan$ \eqref{eq:hypergraph uchan} acts on $\PXp{\rm L}{\prime}$ \eqref{eq:hypergraph logical resource X} by attaching CZ operators to $\PX{j}$ along all edges of each shaded face whose vertices neighbor qubit $j$ \eqref{eq:CCZ action on X}. However, each shaded face borders two such $\PX{j}$ operators, and so each CZ operator is attached twice to $\PXp{\rm L}{\prime}$ \eqref{eq:hypergraph logical resource X}, but $\mathrm{CZ}^2=\ident$. The action of $\uchan$ \eqref{eq:hypergraph uchan} on $\PZp{\rm L}{\prime}$ \eqref{eq:hypergraph logical resource Z} again follows from Eq.~\ref{eq:CCZ action on X}, and again, we find that CZ gates are always attached twice, so that the only change from Eq.~\ref{eq:hypergraph logical resource Z} to Eq.~\ref{eq:hypergraph logical initial Z} is that $\PZ{f} \to \PZ{i}$. Note that Eqs.~\ref{eq:hypergraph logical initial} satisfy Prop.~\ref{prop:teleportation conditions} for the initial state $\ket{\Psi^{\,}_0}$ \eqref{eq:hypergraph initial state}. 

Finally, we identify the string order parameters establishing that the hypergraph resource state $\ket{\Psi^{\,}_t}$ has nontrivial $\Ints^{\,}_2 \times \Ints^{\,}_2$ SPT order. The string order parameters \eqref{eq:string order def} correspond to the stabilizers $\StabEl^{\,}_{\nu,t}$ given by
\begin{align}
    \StabEl^{\vpp}_{\nu,t} &= \uchan \, \StabEl^{\vpp}_{\nu} \, \uchan^{\dagger} = \uchan \, \Pauli{\nu}{i} \, \Pauli{\nu}{f} (T) \, \uchan^{\dagger} \notag \\
    &= \uchan \, \Pauli{\nu}{i} \, \uchan^{\dagger} \,  \trace_{\rm ss} \left[ \, \MeasChannel^\dagger \, \QECChannel^\dagger\,  \Pauli{\nu}{f} \, \QECChannel\,  \MeasChannel \, \SSProj{\bvec{0}}{\rm ss} \, \right] \, ,~
\end{align}
which is the product of the initial-state logical operator $\Pauli{\nu}{i}$ evolved ``forward'' under $\uchan$ \eqref{eq:hypergraph uchan} and the resource-state logical operators \eqref{eq:hypergraph logical resource}. We note that $\PZ{i}$ commutes with $\uchan$\eqref{eq:hypergraph uchan}, and is thus unchanged, while $\PX{i}$ evolves according to Eq.~\ref{eq:CCZ action on X}. The resulting order parameters are
\begin{subequations}
\label{eq:hypergraph SOP}
\begin{align}
    \StabEl^{\vpp}_{x,t} &= \mathrm{CZ}_{\Delta_i}\,  \Big( \prod\limits_{j \in \mathcal{V}} \, \PX{j} \, \Big) \, \PX{f} \label{eq:hypergraph SOP X} \\
    \StabEl^{\vpp}_{z,t} &= \PZ{i}  \, \PX{\Delta_i} \, \Big( \prod\limits_{\substack{\Delta \in \mathcal{F}_s \\ \delta \neq \Delta_i}}  \PX{\Delta} \,\Big) \,  \PZ{f} \label{eq:hypergraph SOP Z} \, ,~
\end{align}
\end{subequations}
where $\mathrm{CZ}_{\Delta_i}$ is the product of CZ gates along the three edges of the shaded face $\Delta_i$ to the right of the leftmost qubit $i$ and $\mathcal{V} =V \setminus \left( \mathcal{F}_s \cup \{i,f\} \right)$ is the set of all vertices in the bulk $(i,f)$ that are not part of a shaded face (i.e., the three middle dotted vertices in Fig.~\ref{fig:hypergraph chain}). Importantly, the endpoint operators for $\StabEl^{\,}_{x,t}$ and  $\StabEl^{\,}_{z,t}$ both anticommute (on qubit $f$ and the shaded face $\Delta_i$), while the operators commute everywhere else (in the bulk). 

The generators of the $\Ints^{\,}_2 \times \Ints^{\,}_2$ are given by
\begin{equation}
\label{eq:hypergraph z2^2}
    U(x) = \prod\limits_{j \notin \mathcal{F}_s} \PX{j} ~,~~U(z) = \prod_{\Delta \in \mathcal{F}_s} \PX{\Delta} \, ,~~
\end{equation}
on an infinite chain, which are represented projectively by the endpoint operators in Eq.~\ref{eq:hypergraph SOP}. The nontrivial string order of the hypergraph state is captured by
\begin{equation}
    \label{eq:hypergraph nontrivial string order}
    \matel{\Psi^{\,}_t}{\, \StabEl^{\vpp}_{\nu,t}}{\Psi^{\,}_t} = 1 \, ,~~
\end{equation}
for $\nu=x,z$. Note that any of the dotted vertices $j>i$ in Fig.~\ref{fig:hypergraph chain} constitute valid right endpoints, and any tetrahedron corresponding to a dotted vertex $j<f$ and the shaded face $\Delta_j$ to its right constitutes a valid left endpoint. The chain in Fig.~\ref{fig:hypergraph chain} can clearly be extended to be thermodynamically large, so that the hypergraph state has nontrivial SPT order protected by the symmetry group $\mathcal{G} = \Ints^{\,}_2 \times \Ints^{\,}_2$ generated by Eq.~\ref{eq:hypergraph z2^2}.

\subsection{Cluster states for \texorpdfstring{$ \numQ \geq 1$}{one or more} logical qubits}
\label{subsec:k>1 cluster}

\begin{figure}[t]
    \centering
    \includegraphics[width=0.95\linewidth]{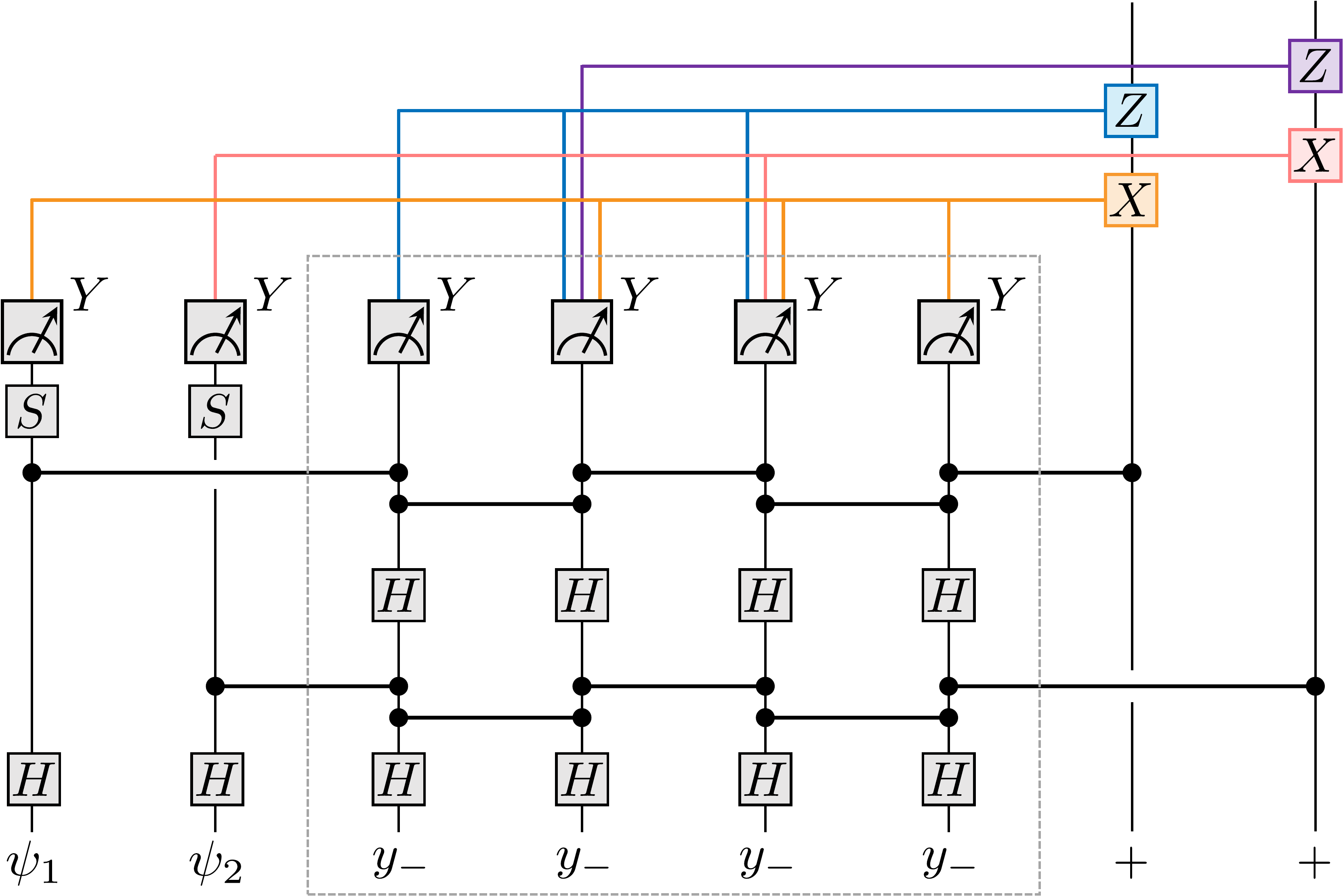}
    \caption{Standard teleportation of $\numQ=2$ logical qubits using a modification of the generalized cluster state $\ket{C^{\,}_{2}}$ \eqref{eq:generalized cluster}. The conditioning of the recovery operations on the measurements is indicated by color-coded lines. The region within the dotted box is precisely the ``bulk'' resource state $\ket{C^{\,}_{2}}$ \eqref{eq:generalized cluster}, where the operations outside the dotted box capture the required modifications near the boundaries.}
    \label{fig:gen_cluster}
\end{figure}

We first consider an extension of the cluster-state teleportation protocol discussed in Secs.~\ref{subsec:phys teleport} and \ref{subsec:cluster-y} \cite{Briegel_2001} to $\numQ \geq 1$ \cite{dumbqc}. We present a family of SPT states $\ket{C^{\,}_{\numQ}}$ \eqref{eq:generalized cluster} with protecting symmetry $\mathcal{G}=\Ints_2^{2 \numQ}$ (as required by Theorem~\ref{thm:spt order}) that is closely related to other known families of resource states for $\numQ$-qubit teleportation \cite{Raussendorf2022, dumbqc}. 

For convenience of presentation, we first define the extended ``$\numQ$-fold'' cluster state in the \emph{bulk} of the chain,
\begin{equation} \label{eq:generalized cluster}
    \ket{C^{\vpp}_{\numQ}} \equiv \bigg( \prod_j \mathrm{CZ}^{\,}_{j,j+1} \, \prod_{j'} H^{\,}_{j'} \bigg)^{\numQ} \, \ket{y_-y_-\cdots y_-} \, ,~
\end{equation}
where $H^{\,}_j = (\PX{j} + \PZ{j})/\sqrt{2}$ is the Hadamard gate, and the $+1$ eigenstates $\ket{+}$ of $\PX{}$ in Eq.~\ref{eq:cluster initial state} have been replaced by the $-1$ eigenstates $\ket{y_-}$ of $\PY{}$ (where $\ket{y_{\pm}}= (\ket{0} \pm \ii \, \ket{1})/\sqrt{2}$ in the $\PZ{}$ basis). We note that the $\numQ=1$ state $\ket{C^{\,}_{1}}$ \eqref{eq:generalized cluster} is unitarily equivalent to the standard cluster state \cite{Briegel_2001} via a $\PZ{}$-basis rotation of all qubits by $\pi/4$. The stabilizer group $\UStabOf*{\ket{C^{\,}_{\numQ}}}$ is generated by the operators
\begin{equation}
    \StabEl^{\vpp}_j = \PZ{j} \, \PY{j+1} \cdots \PY{j+2 \numQ-1} \, \PZ{j+2\numQ} \, ,~~
\end{equation}
which bears some resemblance to Eq.~\ref{eq:cluster v2 string ops}. Furthermore, $\ket{C^{\,}_{\numQ}}$ \eqref{eq:generalized cluster} has a $\Ints_2^{2k}$ symmetry generated by 
\begin{equation}
    \label{eq:generalized cluster sym gen}
    g^{\vpp}_p = \prod\limits_j \, \PY{2 \, \numQ \, j + p} \, ,~~
\end{equation}
for $p \in [ 1,2 \, \numQ]$. Hence, these $\numQ$-fold cluster states can be used to teleport $\numQ$ logical qubits via bulk measurements in local $Y$-basis, as depicted in Fig.~\ref{fig:gen_cluster} for $\numQ=2$ and a single ``measurement region'' of size $2\numQ$ in the bulk. The color-coded lines indicate how the measurement outcomes determine the recovery operations on $F$. 

Importantly, we note that the resource state $\ket{\Psi^{\,}_t}$ for teleportation requires modifying $\ket{C^{\,}_{\numQ}}$ \eqref{eq:generalized cluster} in the vicinity of the initial and final logical sets $I$ and $F$. In the bulk of the chain (i.e., a distance $r > \LRvel \, T$ from both $I$ and $F$), the initial state $\ket{\Psi^{\,}_0}$ is simply $\ket{y_-}$ on each site, and the unitary $\uchan$ is \emph{identical} to the one implicit in Eq.~\ref{eq:generalized cluster}. The required modifications are shown explicitly in Fig.~\ref{fig:gen_cluster} for $\numQ=2$ logical qubits. They include modifying (\emph{i}) the initial state $\ket{\Psi^{\,}_0}$ in $F$ according to  $\ket{y_-} \to \ket{+} $; (\emph{ii}) the unitary $\uchan$ in the region $I=\{1,2\}$ to include phase gates $S = \mathrm{diag}(1,\ii)$ in the computational basis; and (\emph{iii}) the measurement channel $\MeasChannel$ in the region $I$ to include $\numQ=2$ additional measurements in $I$ (only the four measurements in the boxed ``bulk'' region in Fig.~\ref{fig:gen_cluster} repeat for $\Nregions>1$). 

The important takeaway is that the resource state $\ket{\Psi^{\,}_t}$ \eqref{eq:resource state} is equivalent to $\ket{C^{\,}_{\numQ}}$ \eqref{eq:generalized cluster} in the bulk of the chain. By the same token, this state is compatible with Theorem~\ref{thm:spt order} for any $\numQ$, as the modifications to $\ket{C^{\,}_{\numQ}}$ \eqref{eq:generalized cluster} required for teleportation are localized to the two edges, and merely modify the endpoint operators, which anticommute as required. Although it is possible to work out the modifications required at the two edges (i.e., outside of the boxed region in Fig.~\ref{fig:gen_cluster}) for any $\numQ>2$, this is not particularly illuminating or instructive. Instead, we consider an alternative family of states with a systematic procedure for encoding the edges in Sec.~\ref{subsec:k>1 estp}.

\subsection{Valence-bond states for \texorpdfstring{$ \numQ \geq 1$}{one or more} logical qubits}
\label{subsec:k>1 estp}

We now describe another family of resource states suitable for the standard teleportation of $\numQ \geq 1$ logical qubits. In contrast to the generalized cluster states $\ket{C^{\,}_{\numQ}}$ \eqref{eq:generalized cluster} of Sec.~\ref{subsec:k>1 cluster}, the unitary $\uchan$ that prepares the actual resource state $\ket{\Psi^{\,}_t}$ is straightforwardly related to the unitary that prepares the corresponding bulk SPT state $\ket{V^{\,}_{\numQ}}$ \eqref{eq:alpha chains}.

\begin{figure}
    \centering
    \includegraphics[width=0.95\linewidth]{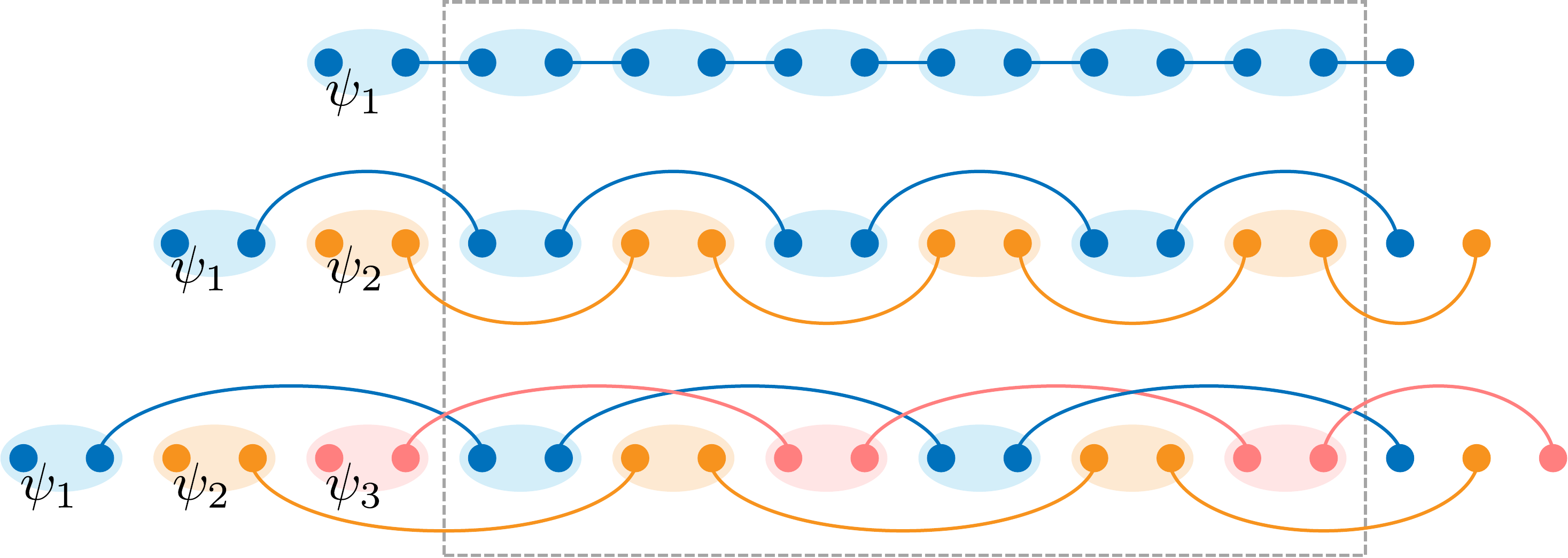}
    \caption{The valence-bond resource states for $\numQ=1,2,3$ (top to bottom), where two dots connected by a line denote the Bell state $\ket{\mathrm{Bell}}$. The \emph{bulk} states in the dotted box realize $\ket{V^{\,}_{\numQ}}$ \eqref{eq:alpha chains}, which are closely related to the resource state $\ket{\Psi^{\,}_t}$ \eqref{eq:resource state} for teleportation. Performing Bell-basis $XX$ and $ZZ$ measurement on each oval teleports the qubits to the rightmost end points of Bell states, up to Pauli errors. }
    \label{fig:alpha chains}
\end{figure}

We refer to this family of states as \emph{valence-bond states}, due to their resemblance to, e.g., the Affleck-Kennedy-Lieb-Tasaki (AKLT) state \cite{AKLT}. These states are also closely related to the $\alpha$ chains introduced in Ref.~\citenum{Ruben1dSPT}. Given a 1D chain with an even number of qubits $\Nspins$, we define the $\numQ$-fold valence bond state $\ket{V^{\,}_{\numQ}}$ \eqref{eq:alpha chains} as
\begin{equation} \label{eq:alpha chains}
    \ket{V^{\vpp}_{\numQ}} \equiv \bigotimes_j \, \ket{\mathrm{Bell}}^{\vpp}_{2j,2 j+2 \numQ -1} \, , ~~
\end{equation}
where $\ket{\mathrm{Bell}}^{\,}_{a,b} = ( \ket{00}^{\,}_{a,b} + \ket{11}^{\,}_{a,b})/\sqrt{2}$. These states are depicted in Fig.~\ref{fig:alpha chains} for $\numQ=1,2,3$, where the dotted box encloses the bulk of the chain. The states $\ket{V^{\,}_{\numQ}}$ \eqref{eq:alpha chains} are symmetric under the group $\mathcal{G}=\Ints_2^{2\numQ}$ generated by
\begin{subequations}
    \label{eq:alpha chain sym gen}
\begin{align}
     g^x_p &= \prod\limits_j \, \PX{2 \, \numQ \, j + 2 \, p} \,\PX{2 \, \numQ \, j + 2 \, p+1} \label{eq:alpha chain sym gen x} \\
     g^z_p &= \prod\limits_j \, \PZ{2 \, \numQ \, j + 2 \, p} \,\PZ{2 \, \numQ \, j + 2 \, p+1} \label{eq:alpha chain sym gen z}
\end{align}
\end{subequations}
for $p \in [1,\numQ]$. The corresponding string order parameters \eqref{eq:k>1 string order parameter} act as one of the two generators \eqref{eq:alpha chain sym gen} in the the bulk.

We now describe the valence-bond teleportation protocol. Given $\numQ \geq 1$ logical qubits and a target distance $\Dist^{\,}_*$ \eqref{eq:L def}, we choose the minimum number of measurement regions $\Nregions \geq 1$ such that  $\Dist = \left(2 \, \Nregions + 1 \right) \numQ + 1 \geq \Dist^{\,}_*$. These protocols involve at most nearest-neighbor gates (with $\LRvel=1$), and are optimal in that they use the minimum required depth $T = \numQ +1$ \eqref{eq:k>1 min depth}, so that $\Dist = T + 2 \, \Nregions \, (T-1)$ saturates the teleportation bound \eqref{eq:k>1 bound}.

The protocol is implemented on a 1D chain with $\Nspins= (2 \, \Nregions+3) \numQ$ qubits. These qubits can be partitioned into a ``bulk'' region with $2 \, \numQ \, \Nregions$ qubits (the region enclosed in the dotted box in Fig.~\ref{fig:alpha chains}) along with an additional $2\, \numQ$ qubits to the left of the bulk region and another $\numQ$ qubits to the right. The system is initialized in the state
\begin{align}
\label{eq:valence bond initial}
    \ket{\Psi^{\vpp}_0} &= \prod\limits_{n=1}^{\numQ} \ket{\psi^{\vpp}_n}^{\vpp}_{2n-1} \otimes \ket{0}^{\vpp}_{2n}  \otimes \prod\limits_{j = 2\numQ +1}^{\Nspins} \ket{0}^{\vpp}_j \, ,~~
\end{align}
where $\ket{\psi^{\,}_n}$ is the $n$th logical state on site $i_n=2n-1$. All bulk qubits $j \in [2\numQ+1,2(\Nregions+1)\numQ]$ are prepared in the initial state $\ket{0}$, though the sites $j=i_n+1$ and $j=\Nspins-n$ for $n\in [1,\numQ]$ may be prepared in \emph{any} state.

The \emph{resource} state $\ket{\Psi^{\,}_t}$ \eqref{eq:resource state} for teleportation is prepared in $T = \numQ+1$ layers of unitary evolution labelled $\uchan^{\,}_{p}$ for $p \in [0,\numQ]$. While the \emph{valence-bond} state $\ket{V^{\,}_{\numQ}}$ \eqref{eq:alpha chains} is prepared by the first $\numQ$ layers alone, the final unitary layer $\uchan^{\,}_{\numQ}$ merely converts the symmetry generators in Eq.~\ref{eq:alpha chain sym gen} into on-site operators, so both states are SPTs).  

The first unitary layer is given by 
\begin{align}
\label{eq:valence bond U0}
    \uchan^{\vpd}_0 &= \Big( \prod\limits_{n=1}^{\numQ} \, \mathrm{SWAP}^{\vpd}_{2n-1,2n} \, \Big) \, \Bigg[ \prod\limits_{j=\numQ+1}^{(\Nregions+1)\numQ} \, \mathcal{B}^{\vpd}_{2j-1,2j} \, \Bigg] \, \ident \, ,
\end{align}
corresponding to (\emph{i}) the $2 \numQ$ sites on the left, (\emph{ii}) the $2 \Nregions \numQ $ sites in the bulk, and (\emph{iii}) the $\numQ$ sites on the right, respectively, where $\mathrm{SWAP} \, \ket{a,b}=\ket{b,a}$  and, for convenience, we define the Bell encoding channel $\mathcal{B}$ as
\begin{align}
\label{eq:Bell channel}
    \mathcal{B}^{\vpd}_{a,b} \equiv \mathrm{CNOT} \left( a \to b \right) \, H^{\vpd}_{a} \, , ~~
\end{align}
where $\mathrm{CNOT} \left( a \to b \right)$ flips the $\PZ{}$ state of qubit $b$ if qubit $a$ is in the $\PZ{}$ state $\ket{1}$, and acts trivially otherwise, so that $\mathcal{B}^{\,}_{a,b} \, \ket{00}^{\,}_{a,b} = \ket{\mathrm{Bell}}^{\,}_{a,b}$. Note that the SWAP gates in $\uchan^{\,}_0$ \eqref{eq:valence bond U0} can be omitted if we swap the states of sites $2n-1$ and $2n$ in $\ket{\Psi^{\,}_0}$ \eqref{eq:valence bond initial}, at the cost of decreasing $\Dist$ \eqref{eq:L def} by one. When preparing the valence-bond state $\ket{V^{\,}_{\numQ}}$ \eqref{eq:alpha chains} directly, $\uchan^{\,}_0$  only contains the bulk term in Eq.~\ref{eq:valence bond U0}.

The next layers of $\uchan$ are given by
\begin{equation}
    \label{eq:valence bond U p}
    \uchan^{\vpd}_{p} = 
    \begin{cases}
        \prod\limits_{j=1}^{(\Nregions+1)\numQ^{\vpd}} \, \mathrm{SWAP}^{\vpd}_{2j+2p-2,2j+2p-1} \, & \, p~{\rm odd} \\[2pt]
        \prod\limits_{j=1}^{(\Nregions+1)\numQ} \, \mathrm{SWAP}^{\vpd}_{2j+2p-3,2j+2p-2} \, & \, p~{\rm even}
    \end{cases} \, ,~~
\end{equation}
which can be thought of more simply as alternating layers of SWAP gates on even versus odd bonds, where the first layer of SWAP gates acts on even bonds---i.e., the SWAP gates in $\uchan^{\,}_1$ \eqref{eq:valence bond U p} are staggered by one site with respect to the Bell encoding channels $\mathcal{B}$ in $\uchan^{\,}_0$ \eqref{eq:valence bond U0}.

The product of the untiaries $\uchan^{\,}_{\numQ-1} \cdots \uchan^{\,}_0$ map the state $\ket{\bvec{0}}$ to the valence-bond state $\ket{V^{\,}_{\numQ}}$ \eqref{eq:alpha chains}. The analogous state $\ket{\Psi^{\,}_{T-1}}$ for teleportation is captured by Fig.~\ref{fig:alpha chains} for $\numQ=1,2,3$; however, writing a general expression is inconvenient. The final layer of unitary evolution applies the Bell \emph{decoding} channel to the qubits in ovals in Fig.~\ref{fig:alpha chains}, and SWAP gates to all qubits on the right end, i.e.,
\begin{align}
    \uchan^{\vpd}_{\numQ} =  \prod\limits_{j=1}^{\numQ  \Nregions}  \mathcal{B}^{\dagger}_{\numQ+2j-1,\numQ+2j}   \, \prod\limits_{n=0}^{\numQ-1} \mathrm{SWAP}^{\vpd}_{\Nspins-2n-1,\Nspins-2n} ,
\end{align}
where the SWAP gates on the final $2\numQ$ sites can be omitted at the cost of decreasing $\Dist$ \eqref{eq:L def} by one. Following the application of all unitaries, one measures $\PZ{j}$ on all sites $[\numQ+1,(2\Nregions+1)\numQ]$, where every $2\numQ$ consecutive outcomes are determine recovery channels according to $(m_{x,1},m_{z,1},m_{x,2},m_{z,2},\dots,m_{x,\numQ},m_{z,\numQ})$, where $\Pauli{z/x}{f_n}$ is applied to $f_n$ if  $\sum_{s=1}^{\Nregions} m_{x/z,s} \cong 1$ mod $2$.

\section{Outlook}
\label{sec:outlook}

We have shown that all ``standard'' teleportation protocols that teleport $\numQ$ logical states a distance $\Dist > \LRvel T$ \eqref{eq:L def} are equivalent to preparing a resource state with nontrivial SPT order protected by the symmetry $\mathcal{G}=\Ints_2^{2\numQ}$ \eqref{eq:protecting symmetry} using local unitary evolution for time $T \ll \Dist$, measuring the $2 \numQ$ string order parameters \eqref{eq:k>1 string order parameter} in the bulk, and applying Pauli recovery gates \eqref{eq:unitary feedback teleport} based on the outcomes. These protocols are physical in that the logical states $\ket{\psi^{\,}_n}$ are realized on particular physical qubits. We require that no unitaries are applied to a site after it is measured, though we expect that this requirement can be relaxed.

The crucial feature of standard teleportation protocols, from the perspective of constraining the resource state, is that all outcome-dependent operations are Pauli gates and all classical processing is linear. The recovery group---generated by all physical error-correction operators $\conjchan$ in $\QECChannel$~\eqref{eq:unitary feedback general}, modulo complex phases---is simply the $\numQ$-qubit Pauli group $\PauliGroup*{\numQ}/\Unitary{1}$ for standard teleportation, which forms a projective representation of the symmetry group $\mathcal{G}=\Ints_2^{2\numQ}$~\eqref{eq:protecting symmetry}. This simple structure of the recovery group ensures that the resource state $\ket{\Psi_t}$~\eqref{eq:resource state} belongs to a \emph{single} class of SPT phases protected by the symmetry $\mathcal{G}=\Ints_2^{2\numQ}$.
We are confident that our analysis can be extended, \emph{mutatis mutandis}, to more general recovery operations, where the choice of recovery group dictates both the protecting symmetry $\mathcal{G}$ and phase of matter of the resource state, while the nature of the classical postprocessing of measurement outcomes dictates how that symmetry is \emph{represented} on the physical qubits.
%{In general, we expect that the recovery group will form a projective representation of (some subgroup of) the protecting symmetry $\mathcal{G}$, while }

In this sense, our results (namely, Theorem~\ref{thm:spt order}) represent a first step toward classifying all quantum teleportation protocols based solely on their recovery (or decoding) channels \eqref{eq:unitary feedback general}.
There are numerous recovery operations to consider beyond the ``Paulis and linear processing'' of standard teleportation: The corresponding group can be continuous and non-Abelian (i.e., $\conjchan^{\,}_g \, \conjchan^{\,}_h \neq \omega (g,h) \, \conjchan^{\,}_h \, \conjchan^{\,}_g$ for some complex phase $\omega$), one can condition measurements (or the basis of a measurement) on prior outcomes, and one may condition operations on the outcomes of measurements in other ways. For example, Ref.~\citenum{SIMPS} considers Pauli recovery operations, but requires nonlinear classical processing, in which case the resource state can be chosen to have SPT order, but the symmetry generators $U(g)$ are no longer ``on site''~\eqref{eq:on site}. Other teleportation schemes use more complicated decoders to determine the appropriate recovery operations in the presence of, e.g., noisy measurement outcomes~\cite{lee2022, chen2023realizing, google_teleport}. Finally, we note that matrix product states~\cite{Gross2007, Gross2010} \emph{without} a clear notion of string order~\eqref{eq:string order def}~\cite{LocalizableEntanglement} may be used as resource states for teleportation using a nondeterministic ``repeat-until-success'' recovery scheme. In general, we expect our approach using the Stinespring picture~\cite{SpeedLimit, AaronDiegoFuture} to be crucial to analyzing these more exotic protocols and their potential connection to SPT orders.

Lastly, we comment on the robustness of the teleportation schemes that fall under our formalism to noise and errors. It is known that teleportation protocols using SPT states are resources are robust to symmetry-respecting perturbations to the resource state \cite{DominicMBQC_SPT,Marvian2017}. Therefore, by unveiling the presence of SPT order in all standard teleportation protocols, we have shown that all such protocols are robust to symmetry-respecting perturbations. Conversely, in the protocols we consider, precise knowledge of the measurement outcomes is crucial to the final recovery operations---there is no tolerance for even a \emph{single}  measurement error or generic error in the resource state. To achieve robustness to generic perturbations as well as noise and measurement errors, we need to consider the broader problem of fault-tolerant teleportation.  There exist fault-tolerant teleportation protocols using 3D cluster states as resource states~\cite{Raussendorf2005a}. These have furthermore been connected to the presence of SPT order protected by higher-form symmetries~\cite{Roberts_2017}. From the perspective of our results, the ability to teleport using the 3D cluster state in the \emph{absence} of errors follows from the presence of the membrane order parameters defined in Ref.~\citenum{Roberts_2017}, which play the role of string-order parameters. However, the fault tolerance of the protocol, and its relation to higher-form symmetries, lies beyond our framework, and is a compelling direction for future work.

% For fault tolerance, the simplest method would be to ``concatenate'' a teleportation protocol with a quantum error-correcting code (QECC)---i.e., each physical qubit in the teleportation protocols we describe is replaced by a \emph{logical} qubit of a QECC  (e.g., a surface code \cite{Kitaev_2003, toric_1998}). In this case, both the initial state $\ket{\Psi_0}$~\eqref{eq:teleport initial mb state}  and resource state $\ket{\Psi_t}$~\eqref{eq:resource state} are long-range entangled due to this encoding (the entanglement is quantified by the code distance of the underlying QECC). One could then imagine a two-stage measurement-based protocol involving (\emph{i}) prepare the codespace of the relevant QECC with measurements and (\emph{ii}) teleport the logical qubits by entangling and measuring them. \yifan{Can include Ref A's suggested citations here.} The latter step may be nonlocal in space, and potentially no more efficient than performing a logical swap operation. Nonetheless, an example of such a fault-tolerant resource state is the 3d cluster state \cite{Raussendorf2005a}, whose localizable entanglement (and teleportation) is protected by a higher-form symmetry \cite{Roberts_2017}. We speculate that our $\mathbb{Z}^{2k}_2$ classification can be extended to fault-tolerant protocols by means of higher-form symmetries. \aaron{continue editing}}

\section*{Acknowledgments}
We thank Ollie Hart, Andy Lucas, Charles Stahl, and Chao Yin for useful discussions and collaboration on related work. This work was supported in part by the U.S. Air Force Office of Scientific Research under Grant No. FA9550-21-1-0195 (YH), the Office of Naval Research via Grant No. N00014-23-1-2533 (YH), the Simons Collaboration on Ultra-Quantum Matter via Award No. 651440 (DTS), and the US Department of Energy via Grant No. DE-SC0024324 (AJF). Publication of this article was funded by the University of Colorado Boulder Libraries Open Access Fund.

\bibliographystyle{apsrev4-2}
\bibliography{thebib}

\appendix

\renewcommand{\thesubsection}{\thesection.\arabic{subsection}}
\renewcommand{\thesubsubsection}{\thesubsection.\arabic{subsubsection}}
\renewcommand{\theequation}{\thesection.\arabic{equation}}

\section{Proof of Proposition~\ref{prop:teleportation conditions}}
\label{app:Proof state=operator teleport}

\begin{proof}
For the teleportation conditions \eqref{eq:teleportation conditions} to be necessary and sufficient for quantum state transfer (see Def.~\ref{def:state transfer}), they must be  equivalent to Eq.~\ref{eq:W state transfer} for the initial state $\ket{\Psi^{\,}_0}$ \eqref{eq:teleport initial mb state} and final state $\ket{\Psi^{\,}_T}$ \eqref{eq:teleport final mb state}. We prove both directions of implication in turn, using the fact that all operators acting on a qubit can be decomposed onto the Pauli operators $\ident,\PX{},\PY{},\PZ{}$, which are spanned by $\PX{}$ and $\PZ{}$. The coefficients $\alpha$ and $\beta$ of a logical state $\ket{\psi}$ can be extracted from the expectation values of the Pauli operators \eqref{eq:logical alpha beta expval}.

We first establish that Eqs.~\ref{eq:W state transfer},  \ref{eq:teleport initial mb state}, and \ref{eq:teleport final mb state} imply Eq.~\ref{eq:teleportation conditions}---i.e., Eq.~\ref{eq:teleportation conditions} is necessary for state transfer. First, Eq.~\ref{eq:teleport final mb state} implies that $\PX{f_n}$ is the $\PX{}$-type logical operator for the $n$th logical qubit acting on the final state $\ket{\Psi^{\,}_T}$ \eqref{eq:teleport final mb state}, as it interchanges $\alpha^{\,}_n$ and $\beta^{\,}_n$. In the Heisenberg picture, this logical operator evolves under $\chan$ according to
\begin{equation}
\label{eq:logical teleport logical X def}
   \LXp{(n)}(T) \, = \,\PX{f_n}  (T) \, = \, \chandag \, \PX{f_n} \, \chan \, ,~~
\end{equation}
and now, using the fact that $\PXp{i_n}{2} = \ident$ is involutory, we introduce an arbitrary operator $\StabEl^{\,}_{n,x}$ via
\begin{align}
\label{eq:logical teleport logical X stab implicit}
    \LXp{(n)}(T) \, &= \, \ident \, \chandag \PX{f_n} \chan \, =\,  \PXp{i_n}{2} \, \chandag \PX{f_n} \chan  \notag \\
    &= \, \PX{i_n} \, \left( \PX{i_n} \, \chandag \PX{f_n} \chan  \right) \, \equiv \, \PX{i_n} \, \StabEl^{\,}_{n,x}  \, ,~~
\end{align}
which implicitly defines the (unitary) operator $\StabEl^{\,}_{n,x}$ as
\begin{equation}
    \label{eq:logical teleport proof S def}
    \StabEl^{\,}_{n,x} \, \equiv \, \PX{i_n} \, \chandag \PX{f_n} \chan \, .
\end{equation}
Leaving the coefficients $\alpha^{\,}_j,\beta^{\,}_j$ for $j\neq n$ implicit for convenience, we apply $\StabEl^{\,}_{n,x}$ \eqref{eq:logical teleport proof S def} to the initial state $\ket{\Psi^{\,}_0 (\alpha^{\,}_n,\beta^{\,}_n)}$ \eqref{eq:teleport initial mb state} and use the first part of Eq.~\ref{eq:W state transfer} to find
\begin{align}
    \StabEl^{\,}_{n,x} \, \ket{\Psi^{\vpp}_0 (\alpha^{\,}_n,\beta^{\,}_n)} \, &= \,  \PX{i_n} \, \chandag \PX{f_n} \chan \, \ket{\Psi^{\vpp}_0 (\alpha^{\,}_n,\beta^{\,}_n)}  \notag \\ 
    &=\,  \PX{i_n} \, \chandag \PX{f_n} \ket{\Psi^{\vpp}_T (\alpha^{\,}_n,\beta^{\,}_n)} \notag \\
    &= \, \PX{i_n} \, \chandag \ket{\Psi^{\vpp}_T (\beta^{\,}_n,\alpha^{\,}_n)}  \notag \\
    &= \, \PX{i_n}  \ket{\Psi^{\vpp}_0 (\beta^{\,}_n,\alpha^{\,}_n)}  \notag \\
    &= \, \ket{\Psi^{\vpp}_0 (\alpha^{\,}_n,\beta^{\,}_n)}\, ,~~ \label{eq:logical teleport logical X stab action}
\end{align}
meaning that $\StabEl^{\,}_{n,x}$ \eqref{eq:logical teleport proof S def} is an element of the initial-state stabilizer group $\UStabOf*{\ket{\Psi^{\,}_0}}$ \eqref{eq:UStab}, and is guaranteed to be unitary. Next, applying the  Heisenberg-evolved operator $\PX{f_n}(T)$ \eqref{eq:logical teleport logical X def} to the initial state $\ket{\Psi^{\,}_0}$ \eqref{eq:teleport initial mb state} gives
\begin{align}
    \LXp{(n)}(T) \,\ket{\Psi^{\vpp}_0 (\alpha^{\,}_n,\beta^{\,}_n)} \, &= \, \PX{i_n} \, \StabEl^{\,}_{n,x} \, \ket{\Psi^{\vpp}_0 (\alpha^{\,}_n,\beta^{\,}_n)}  \notag \\
    &= \, \PX{i_n} \,  \ket{\Psi^{\vpp}_0 (\alpha^{\,}_n,\beta^{\,}_n)} \, , ~~
    \label{eq:logical teleport evo X apply init}
\end{align}
which holds for \emph{any} choice of $\alpha^{\,}_n,\beta^{\,}_n$. This implies that $\LXp{(n)}(T) \, = \, \chandag \, \PX{f_n} \, \chan \,$ realizes an $\PX{}$-type logical operator for the $n$th logical qubit acting on the \emph{initial} state $\ket{\Psi^{\,}_0}$ \eqref{eq:teleport initial mb state}---i.e., it acts as $\PX{i_n}$ times some stabilizer $\StabEl^{\,}_{n,x}$ \eqref{eq:logical teleport proof S def}, where the latter acts trivially on $\ket{\Psi^{\,}_0}$ \eqref{eq:teleport initial mb state}. 

An analogous procedure applies to the $\PZ{}$-type logical operator for the $n$th logical qubit. That operator evolves in the Heisenberg picture according to
\begin{equation}
\label{eq:logical teleport logical Z def}
   \LZp{(n)}(T) \, = \,\PZ{f_n}  (T) \, = \, \chandag \, \PZ{f_n} \, \chan \, ,~~
\end{equation}
in analogy to Eq.~\ref{eq:logical teleport logical X def}, and implied by Eq.~\ref{eq:teleport final mb state}.  We then define another stabilizer $\StabEl^{\,}_{n,z} \equiv \PZ{i_n} \, \chandag \PZ{f_n} \chan$,  implying that $\StabEl^{\,}_{n,z}$ is also an element of $\UStabOf*{\ket{\Psi^{\,}_0}}$ \eqref{eq:UStab}. Applying the logical $\PZ{}$ operator \eqref{eq:logical teleport logical Z def} to  $\ket{\Psi^{\,}_0}$ \eqref{eq:teleport initial mb state} reveals that
\begin{equation}
\label{eq:logical teleport evo Z apply init}
    \LZp{(n)}(T) = \chandag \, \PZ{f_n} \, \chan  = \PZ{i_n} \,\StabEl^{\,}_{n,z} \, ,~~ 
\end{equation}
is the $\PZ{}$-type logical operator for the $n$th logical qubit acting on the initial state $\ket{\Psi^{\,}_0}$ \eqref{eq:teleport initial mb state}, in analogy to Eq.~\ref{eq:logical teleport evo X apply init}. 

Thus, the fact that $\chan$ teleports $\numQ$ logical states in the Schr\"odinger picture via  EQ.~\ref{eq:W state transfer} implies that (\emph{i}) $\chan$ teleports $\numQ$ pairs of logical \emph{operators} in the Heisenberg picture; (\emph{ii}) $\LXp{(n)}(T) \, = \, \chandag \, \PX{f_n} \, \chan \, = \, \PX{i_n} \, \StabEl^{\vpp}_{n,x}$ and $\LZp{(n)}(T) \, = \, \chandag \, \PZ{f_n} \, \chan\, = \, \PZ{i_n} \, \StabEl^{\vpp}_{n,z}$ are valid $\PX{}$- and $\PZ{}$-type logical operators for the $n$th logical qubit acting on the \emph{initial} state  $\ket{\Psi^{\,}_0}$ \eqref{eq:teleport initial mb state}; and (\emph{iii}) the operators $\StabEl^{\,}_{n,x}$ and $\StabEl^{\,}_{n,z}$ are elements of the initial-state stabilizer group $\UStabOf*{\ket{\Psi^{\,}_0}}$ \eqref{eq:UStab}. Hence, Eq.~\ref{eq:W state transfer} implies Eq.~\ref{eq:teleportation conditions}.

Now, we prove the reverse implication: Any protocol $\chan$ satisfying Eq.~\ref{eq:teleportation conditions} with the initial state $\ket{\Psi^{\,}_0}$ \eqref{eq:teleport initial mb state} implies state transfer \eqref{eq:W state transfer} with final state $\ket{\Psi^{\,}_T}$  \eqref{eq:teleport final mb state}---i.e., Eq.~\ref{eq:teleportation conditions} is sufficient for Eq.~\ref{eq:W state transfer}. For convenience, we define 
\begin{align}
    \DensMat^{\vpd}_n \, &= \, \BKop{\psi^{\,}_n}{\psi^{\,}_n}  \notag \\
    &= \, \frac{1}{2} \ident +  c^{(n)}_1 \, \PX{i_n} + + c^{(n)}_2 \, \PY{i_n} + c^{(n)}_3 \, \PZ{i_n}  \notag \\
    &= \, \frac{1}{2} \left( \ident + \sum_{\nu=1}^3 c^{(n)}_\nu \Pauli{\nu}{i_n} \right) \, ,~\label{eq:logical rho_n}
\end{align}
where the coefficients $c^{(n)}_{\nu}$ are given by
\begin{align}
    c^{(n)}_{\nu} \, &\equiv \, \matel{\psi^{\,}_n}{\Pauli{\nu}{i_n}}{\psi^{\,}_n} \label{eq:logical c_n} \, , ~~
\end{align}
for the $n$th logical qubit, as in Eq.~\ref{eq:logical alpha beta expval}. We can then write the initial density matrix associated with $\ket{\Psi^{\,}_0}$ \eqref{eq:teleport initial mb state} as
\begin{align}
    \DensMat^{\,}_0  &=  \BKop{\Psi^{\vpp}_0 \left( \bvec{\alpha},\bvec{\beta} \right)}{\Psi^{\vpp}_0 \left( \bvec{\alpha},\bvec{\beta} \right)} \, \notag \\
    &=  \left[ \prod\limits_{n=1}^{\numQ} \, \DensMat^{\vpd}_n \right] \otimes \BKop{\Phi^{\vpp}_{i}}{\Phi^{\vpp}_{i}}^{\vpd}_{\rm rest} \, ,
    \label{eq:teleport initial mb rho}
\end{align}
which realizes $\DensMat^{\,}_n=\BKop{\psi^{\,}_n}{\psi^{\,}_n}$ on each initial logical qubit $i_n$. Using Eqs.~\ref{eq:teleport initial mb rho} and \ref{eq:logical rho_n}, we write
\begin{align}\label{eq:logical rho_n stab final}
    \chan \, \ket{\Psi^{\vps}_i} &= \chan \, \BKop{\psi^{\,}_n}{\psi^{\,}_n}^{\vps}_{i_n} \ket{\Psi^{\vps}_i}  \notag \\
    &= \chan\, \frac{1}{2} \left( \ident + \sum_{\nu=1}^3 c^{(n)}_\nu \Pauli{\nu}{i_n} \right) \ket{\Psi^{\vps}_i}  \notag \\
    &= \chan\, \frac{1}{2} \left( \ident + \sum_{\nu=1}^3 c^{(n)}_\nu \Pauli{\nu}{i_n} \StabEl^{\,}_{n,\nu} \right) \ket{\Psi^{\vps}_i}  \notag \\
    &= \chan\, \frac{1}{2} \left( \ident + \sum_{\nu=1}^3 c^{(n)}_\nu \chandag \Pauli{\nu}{f_n} \,\chan \right) \ket{\Psi^{\vps}_i}  \notag \\
    &= \frac{1}{2} \left( \ident + \sum_{\nu=1}^3 c^{(n)}_\nu \Pauli{\nu}{f_n} \right) \chan\, \ket{\Psi^{\vps}_i}  \notag \\
    &= \BKop{\psi^{\,}_n}{\psi^{\,}_n}^{\vps}_{f_n} \, \chan\, \ket{\Psi^{\vps}_i} \, ,~~
\end{align}
where the second line follows from Eq.~\ref{eq:logical rho_n}; the third line uses $\StabEl^{\,}_{n,\nu} \in \UStabOf*{\ket{\Psi^{\,}_0}}$ to insert stabilizers inside the parenthetical expression; the fourth line uses Eqs.~\ref{eq:teleportation conditions} and \ref{eq:Y teleportation condition}; and the last line uses Eq.~\ref{eq:logical rho_n} once again. Thus, the final state $\ket{\Psi^{\,}_T} = \chan\,\ket{\Psi^{\,}_0}$ \eqref{eq:teleport final mb state} is stabilized by the $n$th logical qubit's density matrix on site $f^{\,}_n$. For brevity, we represent the final density matrix as $\DensMat^{\,}_T = \chan\,\BKop{\Psi^{\,}_0}{\Psi^{\,}_0}\chandag$. Since $\BKop{\psi^{\,}_n}{\psi^{\,}_n}^{\,}_{f_n}$ only acts nontrivially on site $f_n$, we have
\begin{align}
    1 &= \tr{\, \DensMat^{\vps}_T \,}  = \tr{\, \BKop{\psi^{\,}_n}{\psi^{\,}_n}^{\vps}_{f_n} \, \DensMat^{\vps}_T \,}  \notag \\
    &= \tr{\, \BKop{\psi^{\,}_n}{\psi^{\,}_n} \, \DensMat^{\vps}_{f_n} \,}  = F \left( \rho^{\vps}_{f_n} , \ket{\psi^{\vps}_n} \right) \, ,~
\end{align}
where $F$ is the \emph{fidelity} \cite{Jozsa_1994}, and in second line we used Eq.~\ref{eq:logical rho_n stab final}; in the third line we have defined $\DensMat^{\,}_{f_n} \equiv \ptr{f^c_n}{\,\DensMat^{\,}_T\,}$ as the reduced final density matrix on site $f_n$; in the last line we have used the fact that the fidelity obeys $F(\rho,\ket{\psi}) = \matel{\psi}{\rho}{\psi} = \tr{\, \BKop{\psi}{\psi} \, \rho \,}$ when one of the states is pure \cite{Jozsa_1994}. Note that, in the second line the trace is over the entire system, whereas in the third line, it is only over $f_n$. Since $F (\rho^{\vps}_{f_n} , \ket{\psi^{\vps}_n} ) = 1$, we must have that
\begin{align}
    \rho^{\vps}_{f_n} = \ptr{f^c_n}{\,\DensMat\,} = \BKop{\psi^{\,}_n}{\psi^{\,}_n} \, ,~
\end{align}
i.e., the final reduced density matrix on site $f_n$ is the logical $n$th (pure) state. Since the above arguments hold for all $n=1,\dots,\numQ$ independently, we conclude that $\ket{\Psi^{\,}_T} = \chan\,\ket{\Psi^{\,}_0}$ satisfies \eqref{eq:W state transfer} for the final state $\ket{\Psi^{\,}_T}$ \eqref{eq:teleport final mb state}. Thus, Eq.~\ref{eq:teleportation conditions} is both necessary and sufficient for Eq.~\ref{eq:W state transfer}.
\end{proof}

\section{Proof of Proposition~\ref{prop:single-qubit observable}}
\label{app:Proof of involutory meas}

\begin{proof}
A generic and nontrivial single-qubit observable $\mobserv^{\,}_j$ is guaranteed to have exactly two unique eigenvalues $\meig{0}>\meig{1}$, with spectral decomposition \eqref{eq:spectral decomp}
\begin{align}
    \label{eq:Aj spectral decomp}
    \mObserv{j}{\vpd} &= \meig{0} \, \Proj{0}{j} + \meig{1} \, \Proj{1}{j} \, , ~~
\end{align}
which we decompose onto the Pauli basis via
\begin{align}
    \mObserv{j}{\vpp} &=  \frac{1}{2} \, \sum\limits_{\nu=0}^3 \, \trace\limits_j \, \left[ \, \mObserv{j}{\vpp} \, \Pauli{\nu}{j} \, \right] \, \Pauli{\nu}{j} \notag \\
    &= \frac{1}{2} \left(  \meig{0}  + \meig{1} \right) \, \ident + \vec{\alpha}^{\vpp}_j \cdot \vec{\sigma}^{\vpp}_j \, , ~~\label{eq:single qubit Pauli decomp nice}
\end{align}
where $\vec{\sigma}^{\,}_j = ( \PX{j}, \PY{j},\PZ{j})$ is a vector of Pauli matrices and $\vec{\alpha}^{\,}_j$ is a vector of coefficients as determined by the Pauli decomposition. Importantly, using the properties of $2 \times 2$ matrices and the fact that $\det \, \left[ \mobserv^{\,}_j \right] = \meig{0} \, \meig{1}$ we find
\begin{equation}
    \label{eq:coeff vector magnitude}
    \abs{\vec{\alpha}} = \frac{1}{2} \abs{ \meig{0}-\meig{1}} \, ,~~
\end{equation}
for the magnitude of $\vec{\alpha}$. This provides for the definition of a unit vector $\hat{\alpha}$ that is parallel to the original vector $\vec{\alpha}$.

We rewrite the decomposed observable $\mobserv^{\,}_j$ \eqref{eq:single qubit Pauli decomp nice} as 
\begin{align}
    \mObserv{j}{\vpd} &= \frac{1}{2}\left( \meig{0}+\meig{1} \right) \, \ident + \frac{1}{2} \left( \meig{0}-\meig{1} \right) \, \sum\limits_{n=1}^3 \, \hat{\alpha}^{\vpp}_{j,n} \, \Pauli{n}{j} \notag \\
    &= \frac{1}{2} \trace\limits_j \left[ \mObserv{j}{\vpp} \right] \, \ident + \frac{1}{2} \, \left( \meig{0}-\meig{1} \right) \, \bar{\mobserv}^{\vpp}_j \, ,~~
    \label{eq:single-qubit observable nice}
\end{align}
where the vector $\hat{\alpha}^{\,}_j = ( \hat{\alpha}^{\,}_{j,1}, \hat{\alpha}^{\,}_{j,2}, \hat{\alpha}^{\,}_{j,3})$ is real valued and has unit length. This expression agrees with Eq.~\ref{eq:involutory part}.

We next recover simple expressions for the projectors $\Proj{n}{j}$ \eqref{eq:Aj spectral decomp} onto the states (or spaces of states) with eigenvalues $\meig{n}$. We first note that these two projectors must be \emph{orthogonal} for the eigenvalue equation $\mobserv^{\,}_j \, \Proj{n}{j} = \meig{n} \, \Proj{n}{j}$ to hold. Second, we note that these two projectors must be \emph{complete} (i.e., sum to the identity), or else evaluating $\det \, \left[ \mobserv^{\,}_j \right] $ in the eigenbasis of $\mobserv^{\,}_j$ would return zero, independent of the values of $\meig{0,1}$. As a result,
\begin{align}
    \mObserv{j}{\vpd} &= \meig{n} \, \Proj{n}{j} + \meig{1-n} \left( \ident - \Proj{n}{j} \right) \notag \\
    &= \meig{1-n} \, \ident + \left( \meig{n} - \meig{1-n} \right) \, \Proj{n}{j} \label{eq:Aj single-qubit decomp n} \, , ~~
\end{align}
for either $n=0,1$. Rearranging this expression gives
\begin{align}
    \Proj{n}{j} &= \left( \meig{n}-\meig{1-n} \right)^{-1} \, \left( \mObserv{j}{\vpp} - \meig{1-n} \, \ident \right) \label{eq:Aj proj n naive} \\
    &= \frac{1}{\meig{n}-\meig{1-n}} \left( \frac{1}{2} \trace\limits_j \left[ \mObserv{j}{\vpp} \right] \ident +\abs{\vec{\alpha}} \,  \bar{\mobserv}^{\vpp}_j - \meig{1-n} \ident \right) \notag \\
    &= \frac{1}{\meig{n}-\meig{1-n}} \left( \frac{1}{2} ( \meig{n}-\meig{1} ) \ident + \frac{1}{2} (\meig{0}-\meig{1}) \, \bar{\mobserv}^{\vpp}_j  \right) \notag \\
    &= \frac{1}{2} \left( \ident + (-1)^n \, \bar{\mobserv}^{\vpp}_j \right) \, ,~~\label{eq:Aj proj n invol}
\end{align}
where we used Eq.~\ref{eq:single qubit Pauli decomp nice} to write $\mobserv^{\,}_j$ in terms of its involutory part $\bar{\mobserv}^{\,}_j$ \eqref{eq:involutory part}, along with the relations $\meig{0}-\meig{1} = (-1)^n (\meig{n}-\meig{1-n}) = 2 \abs{\vec{\alpha}}$ \eqref{eq:coeff vector magnitude} and $\trace_j \left[\mobserv^{\,}_j \right] = \meig{0}+\meig{1}$.

Because the involutory operator $\mobserv^{\,}_j$ squares to the identity, its eigenvalues are $\pm 1$. Accordingly, the projectors onto eigenstates of $\mobserv^{\,}_j$ (which are also eigenstates of the original observable $\mobserv^{\,}_j$) are given by Eq.~\ref{eq:Aj proj n invol}. Intuitively, these observables must share eigenstates because they commute: They differ only by a scalar factor and an overall, universally commuting identity component. Finally, because only the projectors onto eigenstates of $\mobserv^{\,}_j$ appear in the unitary measurement channel \eqref{eq:unitary meas general}, measurement of $\mobserv^{\,}_j$ is equivalent to measurement of $\bar{\mobserv}^{\,}_j$ \eqref{eq:involutory part}.
\end{proof}

\section{Proof of Proposition~\ref{prop:feedback required}}
\label{app:Proof feedback required}

\begin{proof}
    At Heisenberg time $\tau=0$, the initial operator acts as $\Gamma \otimes \SSid{\rm ss}$ (i.e., trivially on all Stinespring registers). The outcome-averaged Heisenberg update to $\Gamma$ is given by
    \begin{align}
        \Gamma' &= \trace\limits_{{\rm ss},j} \, \left[ \, \Umeasdag{\mobserv^{\,}_j} \, \Gamma \otimes \SSid{j} \, \Umeas{\mobserv^{\,}_j} \, \SSProj{0}{j} \, \right] \, , ~~\label{eq:no feedback Heis meas update}
    \end{align}
    since the $j$th Stinespring qubit is prepared in the $\ket{0}$ state by default. Using Eqs.~\ref{eq:unitary meas general} and \ref{eq:Aj proj n invol}, the above becomes
    \begin{align}
        \Gamma' &=  \sum\limits_{m,n=0,1} \frac{1}{4} \left( \ident + (-1)^m \, \bar{\mobserv}^{\vpp}_j \right) \, \Gamma \, \left( \ident + (-1)^m \, \bar{\mobserv}^{\vpp}_j \right) \notag \\
        &~~~\quad~~\times ~ \tr{\, \SSXp{j}{m} \, \SSid{j} \, \SSXp{j}{n} \, \SSProj{n}{j} \, } \notag \\
        &= \frac{1}{4} \sum\limits_{m=0,1}  \, \left( \Gamma + (-1)^m \,\acomm{\bar{\mobserv}^{\,}_j}{\Gamma} +\bar{\mobserv}^{\,}_j \, \Gamma \, \bar{\mobserv}^{\,}_j \right)  \notag \\
        &= \frac{1}{4} \sum\limits_{m=0,1} \, \left( \bar{\mobserv}^{\,}_j  + (-1)^m \, \ident \right) \, \acomm{\bar{\mobserv}^{\,}_j}{\Gamma} \, ,~~\label{eq:no feedback Gamma' 1}
    \end{align}
and if $\Gamma$ acts trivially (i.e., as the identity $\ident$) on qubit $j$, then we find $\Gamma' = \Gamma$, meaning that the measurement did nothing, on average. However, if $\Gamma$ acts on qubit $j$ as $\Pauli{\nu}{j}$, then Eq.~\ref{eq:no feedback Gamma' 1} can be rewritten
\begin{align}
    \Gamma' &= \frac{1}{2} \sum\limits_{m=0,1} \,(-1)^m \, \Proj{m}{j} \, \sum\limits_{\mu=1}^{3} \, \hat{\alpha}^{\vpp}_{j,\mu} \, \acomm{\Pauli{\mu}{j}}{\Pauli{\nu}{j}} \notag \\
    &= \sum\limits_{m=0,1} \, (-1)^m \, \Proj{m}{j} \, \hat{\alpha}^{\vpp}_{j,\nu} = \trace\limits_j \, \left[ \, \bar{\mobserv}^{\vpp}_j \, \Pauli{\nu}{j} \, \right] \, \bar{\mobserv}^{\vpp}_j \,, ~~\label{eq:no feedback Gamma' 2}
\end{align}
which is neither involutory nor unitary, and thus incompatible with the teleportation conditions \eqref{eq:teleportation conditions} imposed by, e.g., Def.~\ref{def:standard teleport} and Prop.~\ref{prop:teleportation conditions}. To see this, note that $\Gamma'$ squares to $\tr{\bar{\mobserv}^{\,}_j \, \Pauli{\nu}{j} \, } \, \ident$, which generically has operator norm \eqref{eq:spectral norm} less than one. However, if %the involutory operator 
$\bar{\mobserv}^{\,}_j =\Pauli{\mu}{j}$, then $\Gamma'$ \eqref{eq:no feedback Gamma' 2} is either identically zero or equal to $\Gamma$ (if $\mu=\nu$).  

In other words, if we evolve any logical operator $\Gamma$ (or Pauli-string observable of interest) under a dilated unitary channel $\umeas$ corresponding to projective measurement of a single-qubit observable $\mobserv^{\,}_j$ \emph{without} outcome-dependent recovery operations, then the Heisenberg-evolved operator $\Gamma'$ \eqref{eq:no feedback Heis meas update} realizes one of the following:
\begin{enumerate}
    \item If the operator $\Gamma$ and $\mobserv^{\,}_j$ share no support, then $\Gamma'=\Gamma$, meaning that $\umeas$ acts trivially.
    \item If the involutory part of $\mobserv^{\,}_j$ and the Pauli content of $\Gamma$ on qubit $j$ are both $\Pauli{\mu}{j}$, then $\Gamma'=\Gamma$, so that $\umeas$ acts trivially.
    \item If the involutory part of $\mobserv^{\,}_j$ and the Pauli content of $\Gamma$ on qubit $j$ correspond to \emph{distinct} Pauli operators, then $\Gamma'=0$, meaning that teleportation fails because the logical operator is destroyed.
    \item If the involutory part of $\mobserv^{\,}_j$ realizes a superposition of Pauli operator \eqref{eq:involutory part}, then $\Gamma'$ is no longer involutory nor unitary, and thus cannot reproduce the logical operator algebra \eqref{eq:teleportation conditions} acting on the initial state, which is required by Prop.~\ref{prop:teleportation conditions}.
\end{enumerate}
In all cases above, either the measurement channel $\umeas$ acts trivially on a Pauli string $\Gamma$ of interest, or it causes teleportation to fail (because $\Gamma$ fails to be unitary). Hence, measurement without outcome-dependent recovery operations is not useful to quantum teleportation.
\end{proof}

\section{Proof of Proposition~\ref{prop:pure output}}
\label{app:Proof every trajectory}

\begin{proof}
    By construction, the logical reduced density matrix
    \begin{equation}
        \label{eq:logical DM}
        \DensMat^{\vpp}_{\rm L} \equiv \trace_{F^c} \left[ \chan \, \DensMat^{\vpp}_0 \otimes \BKop{\bvec{0}}{\bvec{0}}^{\vpp}_{\rm ss} \, \chan^{\dagger} \right] \, ,~~
    \end{equation}
    is a pure state, and separable with respect to the partition $F | F^c$. In particular, $\abs{F}=\numQ$ is the number of logical qubits, and $\DensMat^{\,}_{\rm L}$ \eqref{eq:logical DM} is guaranteed to be unitarily equivalent to the pure state $\bigotimes_{n=1}^{\numQ} \BKop{\psi^{\,}_n}{\psi^{\,}_n}^{\,}_n$. Because $F^c$ includes all Stinespring qubits, $\DensMat^{\,}_{\rm L}$ \eqref{eq:logical DM} is implicitly averaged over the outcomes of all measurements in $\chan$. 

    Additionally, the logical reduced density matrix $\DensMat^{\,}_{\rm L}$ \eqref{eq:logical DM} after realizing the sequence of measurement outcomes $\bvec{n}=(n^{\,}_1, \dots, n^{\,}_{\Nmeas})$ is given by \cite{AaronMIPT, AaronDiegoFuture, SpeedLimit}
    \begin{align}
    \label{eq:logical DM traj n}
        \DensMat^{(\bvec{n})}_{\rm L} \equiv \frac{\trace_{F^c} \left[ \, \DensMat^{\vpp}_T \, \SSProj{\bvec{n}}{\rm ss} \, \right]}{\trace_{\rm dil} \left[ \, \DensMat^{\vpp}_T \, \SSProj{\bvec{n}}{\rm ss} \, \right]} \, ,~~
    \end{align}
    where the denominator is equal to the the probability $p^{\,}_{\bvec{n}}$ of the ``outcome trajectory'' $\bvec{n}$.

    Thus, the outcome-averaged logical reduced density matrix $\DensMat^{\,}_{\rm L}$ \eqref{eq:logical DM} is automatically a convex sum over trajectories $\bvec{n}$ of the reduced density matrix $\DensMat^{(\bvec{n})}_{\rm L} $ \eqref{eq:logical DM traj n} weighted by the scalar $p^{\,}_{\bvec{n}}$. However, such a density matrix is pure \emph{if and only if} every term in the sum is proportional to $\DensMat^{\,}_{\rm L}$ \eqref{eq:logical DM}. Because  $\DensMat^{\,}_{\rm L}$ \eqref{eq:logical DM} and $\DensMat^{(\bvec{n})}_{\rm L} $ \eqref{eq:logical DM traj n} are proportional and normalized for all $\bvec{n}$, we must have that $\DensMat^{(\bvec{n})}_{\rm L} = \DensMat^{\vpp}_{\rm L}$ for all $\bvec{n}$. Hence, the output state $\DensMat^{\,}_{\rm L}$ \eqref{eq:logical DM} is realized on \emph{all} outcome trajectories.
\end{proof}

\section{Proof of Lemma~\ref{lem:recovery factorization}}
\label{app:recovery factorization proof}

\begin{proof}
    Let $\QECChannel^{\,}_s$ be a recovery channel \eqref{eq:unitary feedback teleport} and $\Gamma$ to be be any logical operator. By construction (and Def.~\ref{def:Clifford recovery}), both $\conjchan^{\,}_s$ \eqref{eq:unitary feedback teleport} and $\Gamma$ are Pauli-string operators \eqref{eq:Pauli string def}. Evolving $\Gamma$ under $\QECChannel^{\,}_s$ in the Heisenberg picture gives
    \begin{align}
        \QECChannel^{\dagger}_s \, \Gamma \otimes \SSid{} \, \QECChannel^{\vpd}_s &= \sum\limits_{m,n=0,1} \, \ConjChanPow{s}{m} \, \Gamma \, \ConjChanPow{s}{n} \otimes \SSProj{m}{r^{\,}_s} \, \SSid{} \, \SSProj{n}{r^{\,}_s}  \notag \\
        &= \sum\limits_{m=0,1} \, \ConjChanPow{s}{m} \, \Gamma \, \ConjChanPow{s}{m} \otimes \SSProj{m}{r^{\,}_s} \notag \\
        &=  \sum\limits_{m=0,1} \, \left( -1 \right)^{m \, \lambda^{\,}_{\Gamma,s}} \, \Gamma \otimes \SSProj{m}{r^{\,}_s} \notag \\
        %
        %&= \Gamma \otimes \frac{1}{2} \sum\limits_{m=0,1} \, \left( \left( -1 \right)^{m \, \lambda^{\,}_{\Gamma,s}} \SSid{} + \left( -1 \right)^{m \, (\lambda^{\,}_{\Gamma,s}-1)} \, \SSZ{r^{\,}_s} \right) \notag \\
        %
        &= \Gamma \otimes \left( \kron{\lambda^{\,}_{\Gamma,s},0} \, \SSid{} + \kron{\lambda^{\,}_{\Gamma,s},1} \, \SSZ{r^{\,}_s} \right) \notag \\
        &= \Gamma \otimes \SSZp{r^{\,}_s}{\lambda^{\,}_{\Gamma,s}} \label{eq:recovery attach Z proven} \, ,~~
    \end{align}
    where $\SSZ{r^{\,}_s}$ is the product of $\SSZ{j}$ for $j \in r^{\,}_s$, where
    \begin{equation}
    \label{eq:recovery lambda def}
        \ConjChan{s} \, \Gamma \, \ConjChan{s} = \left( -1 \right)^{\lambda^{\vpp}_{\Gamma,s}} \, \Gamma \, , ~~
    \end{equation}
    defines $\lambda^{\vpp}_{\Gamma,s} \in \{0,1\}$ for each  $\Gamma$ and  $\QECChannel^{\,}_s$. 
    
    Next, we prove Eq.~\ref{eq:effective recovery} for a given Pauli string $\Gamma$. The result \eqref{eq:recovery attach Z proven} establishes that evolving $\Gamma$ under any $\QECChannel$ multiplies $\Gamma$ by Stinespring operators $\SSZ{j}$ for $j \in r^{\,}_s$. Hence, any recovery channel $\QECChannel^{\,}_{\Gamma}$ \eqref{eq:effective recovery} attaches $\SSZ{r}$ to $\Gamma$ provided that $\acomm{\conjchan^{\,}_{\Gamma}}{\Gamma}=0$. The particular choice
    \begin{equation} 
    \label{eq:effective conjchan}
    \ConjChan{\Gamma} = \prod\limits_{s=1}^{n} \ConjChanPow{s}{\lambda^{\vpp}_{\Gamma,s}} \, ,~~
    \end{equation}
    always fulfils this condition, where $\lambda^{\vpp}_{\Gamma,s} \in \{0,1\}$ is defined in Eq.~\ref{eq:recovery lambda def} and $n$ is the number of Clifford recovery channels $\QECChannel^{\,}_s$. However, any choice of Pauli string $\conjchan^{\,}_{\Gamma}$ that anticommutes with $\Gamma$ leads to a channel $\QECChannel^{\,}_{\Gamma}$ \eqref{eq:effective recovery} that is equivalent to $\QECChannel$ in its action on $\Gamma$.

    Next, consider a set of logical operators $\bvec{\Gamma}$, consisting of, e.g., all $\PX{}$- and $\PZ{}$-type logical operators for the $\numQ$ logical qubits. Again, Eq.~\ref{eq:recovery attach Z proven}  establishes that $\QECChannel$ acts by attaching (possibly different) $\SSZ{j}$ operators to different logical operators $\Gamma \in \bvec{\Gamma}$. However, by the previous result \eqref{eq:effective recovery}, $\QECChannel$ is equivalent in its action on each $\Gamma \in \bvec{\Gamma}$ to some $\QECChannel^{\,}_{\Gamma}$  \eqref{eq:effective recovery} where $\conjchan^{\,}_{\Gamma}$ is \emph{any} Pauli string that anticommutes with $\Gamma$. Hence, we can write $\QECChannel$ as a product over $\Gamma \in \bvec{\Gamma}$ of channels $\QECChannel^{\,}_{\Gamma}$  \eqref{eq:effective recovery} provided that $\conjchan^{\,}_{\Gamma}$ anticommutes only with the logical operator $\Gamma$, and not its counterpart $\Gamma'$, which acts on the same logical qubit. Since $\acomm{\Gamma}{\Gamma'}=\comm{\Gamma'}{\Gamma'}=0$, the choice $\conjchan^{\,}_{\Gamma} = \Gamma'$ is sufficient to ensure that Eq.~\ref{eq:canonical Clifford form} is equivalent to $\QECChannel$ for all logical operators. We note that the decomposition \eqref{eq:canonical Clifford form} is not unique.
\end{proof}

\section{Proof of Lemma~\ref{lem:measure = attach}}
\label{app:Proof meas = attach}

\begin{proof}
    The canonical-form protocol $\chan= \QECChannel \MeasChannel \uchan$ \eqref{eq:chan canonical form} involves dilated channels of the form 
    \begin{align}
        \QECChannel^{\vpd}_{j,k} &= \sum\limits_{n=0,1} \, \ConjChanPow{j,k}{n} \otimes \SSProj{n}{j} \label{eq:1meas recovery channel} \\
        \umeas^{\vpd}_j  \, &= \frac{1}{2} \sum\limits_{m=0,1} \left( \ident + (-1)^m \, \bar{\mobserv}^{\prime}_j \right) \otimes \SSXp{j}{m} \label{eq:1meas meas channel} 
         \, , ~~
    \end{align}
    where $\conjchan^{\,}_{j,k} \in \PauliGroup*{\rm ph}$ is a physical Pauli string, $\bar{\mobserv}_j'$ is the involutory part \eqref{eq:involutory part} of an operator of the form 
    \begin{equation}
        \mobserv_j'= \cliff^{\vpd}_j \, \mobserv^{\,}_j \cliff^\dagger_j \label{eq:1meas modified observable}
    \end{equation}
    where $\mobserv^{\,}_j$ is a single-qubit observable,  and $\cliff^{\,}_j$ is a generic physical unitary. Note that Eq.~\ref{eq:1meas modified observable} captures modifications due to writing $\chan'$ in canonical form \eqref{eq:chan canonical form}.  %$\SSProj{n}{j}$ is a projector on the $j$th Stinespring qubit,  
    %, and we allow for arbitrary physical Clifford unitaries $\cliff \in \CliffGroup*{\rm ph}$ before and after the recovery channel $\QECChannel$. 
    
    The above form \eqref{eq:1meas meas channel} of the measurement channels is guaranteed by Prop.~\ref{prop:single-qubit observable}. Additionally, by Defs.~\ref{def:Clifford recovery} and \ref{def:standard teleport}, all outcome-dependent recovery operations are captured by Eq.~\ref{eq:1meas recovery channel}. All other channels in $\chan$ correspond to physical unitaries, captured by $\uchan$ in Eq.~\ref{eq:canonical Clifford form}.

    Consider the Heisenberg evolution of the Pauli string 
    \begin{equation}
        \label{eq:1meas Gamma initial}
        \Gamma (0) =  \Gamma = \Pauli{\nu}{f_n} \otimes \SSid{\rm ss} \in \PauliGroup{\rm ph} \, , ~~
    \end{equation}
    which acts as a logical operator on the final state $\ket{\Psi^{\,}_T}$ \eqref{eq:teleport final mb state}. Evolving this operator in the Heisenberg picture, projecting onto the default Stinespring initial state $\ket{\bvec{0}}^{\,}_{\rm ss}$, and averaging over outcomes leads to Eq.~\ref{eq:teleportation conditions}, i.e.,
    \begin{equation}
        \label{eq:measure attach final Gamma}
        \Gamma^{\vpp}_{\rm ph} (T) = \trace\limits_{\rm ss} \left[ \, \chan^{\dagger} \, \Gamma (0) \, \chan \, \SSProj{\bvec{0}}{\rm ss} \, \right]  = \Pauli{\nu}{i_n} \StabEl^{\vpd}_{n,\nu} \, , ~~
    \end{equation}
    which is the same logical operator $\Gamma$ acting on the initial state $\ket{\Psi^{\,}_0}$ \eqref{eq:teleport initial mb state}, where $\StabEl^{\,}_{n,\nu}$ is an element of the initial-state stabilizer group $\UStabOf{\ket{\Psi^{\,}_0}} \subset \Aut*{\Hilbert^{\,}_{\rm ph}}$ \eqref{eq:UStab}.

    Importantly, the na\"ive protocol $\chan$ need not be in canonical form \eqref{eq:chan canonical form}. Without loss of generality, we have omit any measurements without corresponding feedback and convert trivial recovery operations into physical unitaries, as described above. However, for convenience, we also pull all non-Clifford physical unitaries into $\uchan$, and if the combination of a measurement channel $\MeasChannel^{\,}_j$ is conditioned on a prior outcome in the na\"ive protocol $\chan$, then we pull prior recovery channels through $\MeasChannel^{\,}_j$ as needed so that it realizes Eq.~\ref{eq:1meas meas channel}. This is not necessary if $\MeasChannel^{\,}_j$  commutes with all prior recovery channels (in the Schr\"odinger picture) or if $\mobserv^{\,}_j$ is a Pauli string.  
    
    As a result, all operations after the first measurement channel (in the Schr\"odinger picture) correspond either to measurements or Clifford channels. Hence, in the Heisenberg picture, any logical operator $\Gamma (\tau^{\,}_{j,k})$ prior to the recovery channel $\QECChannel^{\,}_{s}$ is guaranteed to be a Pauli string, provided that no measurement channels have been encountered. Note that any number of recovery operations conditioned on various measurement outcomes can be factorized into the desired form \eqref{eq:1meas recovery channel}, up to Clifford gates $\cliff \in \CliffGroup*{\rm ph}$, by Def.~\ref{def:Clifford recovery} and Lemma~\ref{lem:recovery factorization}.
    
    Suppose that  $\umeas^{\,}_j$ is the first measurement channel encountered in the Heisenberg picture. By Lemma~\ref{lem:recovery factorization}, 
    \begin{equation}\label{eq:recovery attach Z}
        \QECChannel^{\dagger}_{j,k} \, \Gamma (\tau^{\vpp}_{j,k} ) \, \QECChannel^{\vpd}_{j,k} = \Gamma (\tau^{\vpp}_{j,k} ) \, \SSZp{j}{\lambda^{\vpp}_{\Gamma,j,k}} \, , ~~
    \end{equation}
    where $\lambda^{\,}_{\Gamma,j,k}$ is defined via
    \begin{equation*}
        \tag{\ref{eq:recovery lambda def}}
        \ConjChan{j,k} \, \Gamma (\tau^{\vpp}_{j,k} ) \, \ConjChan{j,k} = \left( -1 \right)^{\lambda^{\vpp}_{\Gamma,j,k}} \, \Gamma (\tau^{\vpp}_{j,k} )  \,,~~
    \end{equation*}
    so that the physical part of $\Gamma(\tau^{\,}_{j,k})$ is unaffected by the recovery channel, while the Stinespring part is multiplied by the operator $\SSZ{j}$ if and only if  the recovery Pauli $\conjchan^{\,}_{j,k}$ anticommutes with $\Gamma (\tau^{\,}_{j,k})$ immediately prior to the recovery channel $\QECChannel^{\,}_{j,k}$ (in the Heisenberg picture). If $\lambda^{\,}_{\Gamma,j,k}=0$ in Eq.~\ref{eq:recovery lambda def} then Eq.~\ref{eq:recovery attach Z} is trivial, as claimed.

    We then observe that recovery channels conditioned on other measurement outcomes also act on $\Gamma$ via Eq.~\ref{eq:recovery attach Z}. Additionally, only the recovery channels $\QECChannel^{\,}_{j,k}$ encountered prior to the measurement channel $\umeas^{\,}_j$ (in the Heisenberg picture) are actually sensitive to the measurement outcome, and hence we convert all other ``trivial'' recovery channels into physical unitaries. We then define
    \begin{equation}
    \label{eq:1meas sum Lambda}
        \Lambda^{\,}_{\Gamma,j} = \sum_k \lambda^{\,}_{j,k} ~{\rm mod}~2 \, ,~~
    \end{equation}
    which captures whether or not the combination of \emph{all} recovery channels $\QECChannel^{\,}_{j,k}$ conditioned on the $j$th outcome act nontrivially on the logical operator $\Gamma$ via Eq.~\ref{eq:recovery attach Z}.

    We can then write the \emph{physical part} of the time-evolved logical operator $\Gamma(t^{\,}_j)$ immediately prior to the measurement channel $\umeas^{\,}_j$ (in the Heisenberg picture) as
    \begin{equation}
    \label{eq:gamma pre meas j phys}
        \Gamma (t^{\vpp}_j) = \CliffDag{\Nmeas} \, \Gamma \, \Cliff{\Nmeas} \, , ~~
    \end{equation}
    where $\cliff^{\,}_{\Nmeas}$ is the product of all physical Clifford gates $\cliff$ and the outcome-independent parts of all recovery-channels in $\chan$ after the measurement channel $\MeasChannel^{\,}_j$ in the Schr\"odinger picture. Hence, it is the same unitary $\cliff^{\,}_j$ that appears in Eq.~\ref{eq:1meas modified observable}. The Stinespring part of $\Gamma$ is
    \begin{equation}
    \label{eq:gamma pre meas j ss}
        \widetilde{\Gamma}(t^{\vpp}_j) = \SSZp{j}{\Lambda^{\vpp}_{\Gamma,j}} %\, \prod\limits_{j' \neq j} \SSZp{j'}{\sum_k \lambda^{\vpp}_{j',k}} 
        \otimes \cdots \, ,~~
    \end{equation}
    where $\cdots$ captures $\SSZ{}$ operators attached by recovery channels conditioned on the outcomes of other measurements.

    Evolving $\Gamma$ \eqref{eq:gamma pre meas j phys} through the measurement channel $\umeas^{\,}_j$ results in a logical operator $\Gamma$ with physical part
    \begin{align}
        \Gamma (t^{\prime}_j) &= \trace\limits_{ {\rm ss},j} \left[ \, \umeas^{\dagger}_j ~ \Gamma (t^{\vpp}_j) \otimes  \SSZp{j}{\Lambda^{\vpp}_{\Gamma,j}} ~ \umeas^{\vpd}_j \, \SSProj{0}{j}\,  \right] \notag \\
        &= \frac{1}{4} \hspace{-0.6mm} \sum\limits_{m,n=0,1} \hspace{-0.6mm}\left( \ident + (-1)^m \, \bar{\mobserv}^{\vpp}_j \right) \Gamma (t^{\vpp}_j) \left( \ident + (-1)^n \, \bar{\mobserv}^{\vpp}_j \right) \notag \\
        &~~~~~~~~\times \matel*{0}{\SSXp{j}{m} \,\SSZp{j}{\Lambda^{\vpp}_{\Gamma,j}}  \, \SSXp{j}{n} }{0} \notag \\
        &= \sum\limits_{n=0,1} \frac{\left(-1\right)^{n \, \Lambda^{\vpp}_{\Gamma,j}}}{4} \left( \bar{\mobserv}^{\vpp}_j + \left(-1\right)^n \ident \right) \acomm{\bar{\mobserv}^{\vpp}_j}{\Gamma (t^{\vpp}_j)} \notag \\
        %
        %\Gamma (t^{\prime}_j) 
        &= \frac{1}{2} ~\bar{\mobserv}^{\, \Lambda^{\vpp}_{\Gamma,j}-1}_j ~ \acomm{\bar{\mobserv}^{\vpp}_j}{\Gamma (t^{\vpp}_j)} \,, ~~ \label{eq:gamma post meas j}
    \end{align}
    where we have traced out the $j$th Stinespring qubit to average over measurement outcomes and projected onto the default initial state $\ket{0}$, since no subsequent channels act nontrivially on this qubit. The remaining Stinespring content of $\Gamma (t^{\prime}_j)$ corresponds to the $\cdots$ terms in Eq.~\ref{eq:gamma pre meas j ss}.

    Importantly, the operator $\Gamma (t_j')$ \eqref{eq:gamma post meas j} must be \emph{involutory}, which leads to the condition
    \begin{align}
        \ident &= \Gamma^2 (t_j') = \frac{1}{2} \ident + \frac{1}{4} \acomm{\Gamma (t^{\vpp}_j)}{\bar{\mobserv}^{\vpp}_j \, \Gamma (t^{\vpp}_j)\, \bar{\mobserv}^{\vpp}_j} \, , ~
    \end{align}
    independent of $\Lambda^{\,}_{\Gamma,j}$, leading to
    \begin{align}
        0 &= \Gamma (t^{\vpp}_j) \, \bar{\mobserv}^{\vpp}_j \, \Gamma (t^{\vpp}_j)\, \bar{\mobserv}^{\vpp}_j  + \bar{\mobserv}^{\vpp}_j \, \Gamma (t^{\vpp}_j)\, \bar{\mobserv}^{\vpp}_j \,   \Gamma (t^{\vpp}_j) - 2 \, \ident\notag \\
        &=  \comm{\Gamma (t^{\vpp}_j)}{\bar{\mobserv}^{\vpp}_j }^2 \, ,~~\label{eq:gamma post meas j comm sq}
    \end{align}
    and denoting this commutator as
    \begin{equation}
        B =  \comm{\Gamma (t^{\vpp}_j)}{\bar{\mobserv}^{\vpp}_j } = - B^\dagger \, ,~~
    \end{equation}
    we see that $B^2 = B^\dagger B = B B^\dagger = 0$, meaning that $B$ is a normal matrix; since $B^\dagger B \geq 0$,  the only solution to $B^\dagger B = 0$ is $B=0$. Alternatively, we note that $B$ is both skew-Hermitian and nilpotent; the former implies that $B$ has purely imaginary eigenvalues, while the latter implies that all complex eigenvalues of $B$ are zero. 
    
    Either way, the result is the same, and
    \begin{equation*}
        \tag{\ref{eq:obs logical commute}}
        \comm{\mobserv^{\vpp}_j}{\Gamma (t^{\vpp}_j)} = 0 \, ,~~
    \end{equation*}
    for \emph{all} logical operators $\Gamma$ evolved to the time $t^{\,}_j$ \eqref{eq:gamma pre meas j phys} immediately prior to the measurement of $\mobserv^{\,}_j$, in the Heisenberg picture. Given Eq.~\ref{eq:obs logical commute}, Eq.~\ref{eq:gamma post meas j} becomes
    \begin{equation}
        \label{eq:1meas Heise meas good}
        \Gamma (t^{\prime}_j) = \bar{\mobserv}^{\, \Lambda^{\vpp}_{\Gamma,j}}_j ~\Gamma (t^{\vpp}_j) \, ,~~
    \end{equation}
    and the same result recovers from Eq.~\ref{eq:gamma post meas j} in the case where $\Gamma(t^{\,}_j)$ has no support on site $j$. We note that the latter scenario also satisfies the commutation condition \eqref{eq:obs logical commute}; hence, both Eqs.~\ref{eq:obs logical commute} and \ref{eq:1meas Heise meas good} hold generally.
    
    Thus, successful teleportation requires that the measured observable $\mobserv^{\,}_j$ commutes with \emph{all} logical operators $\Gamma(t^{\,}_j)$ evolved in the Heisenberg picture to the time $t^{\,}_j$ immediately prior to the corresponding measurement \eqref{eq:obs logical commute}. Then, the combination of measuring $\mobserv^{\,}_j$ and all recovery operations $\QECChannel^{\,}_{j,k}$ \eqref{eq:1meas recovery channel} conditioned on the outcome attaches the involutory part $\bar{\mobserv}^{\,}_j$ \eqref{eq:involutory part} of the measured observable $\mobserv^{\,}_j$ to $\Gamma(t^{\,}_j)$ \eqref{eq:gamma pre meas j phys} via Eq.~\ref{eq:1meas Heise meas good}, provided that an odd number of recovery Pauli operators $\conjchan^{\,}_{j,k}$ anticommute with the corresponding Heisenberg-evolved logical operators $\Gamma(\tau^{\,}_{j,k})$ \eqref{eq:recovery lambda def} immediately prior to the $k$th recovery channel conditioned on the $j$th outcome---i.e., $\Lambda^{\,}_{\Gamma,j}=1$ \eqref{eq:1meas sum Lambda}. This holds for each logical operator $\Gamma$. 
    
    We now restore the Stinespring operators in the $\cdots$ term in Eq.~\ref{eq:gamma pre meas j ss}, and repeat this procedure for the remaining measurements. We note that subsequent recovery operations (in the Heisenberg picture) generally commute with $\bar{\mobserv}^{\,}_j$, as described below Eq.~\ref{eq:measure attach final Gamma}. The other possibility is that  $\bar{\mobserv}^{\,}_j$ is a Pauli string that anticommutes with one or more $\conjchan^{\,}_{i,\ell}$ recovery operators \eqref{eq:1meas recovery channel}; however, this remains compatible with the update in Eq.~\ref{eq:recovery attach Z}, since the operator conjugated by the channel $\QECChannel^{\,}_{i,\ell}$ is still a Pauli string (until all recovery operations have concluded, all physical unitaries are elements of $\CliffGroup*{\rm ph}$). Thus, all subsequent recovery channels $\QECChannel^{\,}_{i,\ell}$ (in the Heisenberg picture) either act trivially on $\bar{\mobserv}^{\,}_{j}$ in Eq.~\ref{eq:measure attach final Gamma}, or  $\bar{\mobserv}^{\,}_{j}$ is a Pauli string that can be absorbed into $\Gamma (t^{\,}_j)$. 
    
    These properties are preserved by all subsequent gates and measurements. When all recovery channels have been encountered, we allow generic (i.e., non-Clifford) physical unitaries to appear in $\chan$, which may be absorbed into the unitary operators $\cliff$ in Eq.~\ref{eq:gamma pre meas j phys}. Since only the conjugation by recovery channels $\QECChannel^{\,}_{i,\ell}$ is sensitive to having a Pauli-string input, this does not pose an obstacle to evolving under any remaining measurement channels. It is straightforward to verify that the same conditions that applied to prior recovery and measurement channels also apply to all remaining dilated channels. 

    Finally, we note that in a canonical-form protocol $\chan = \QECChannel \, \MeasChannel \, \uchan$, by Lemma~\ref{lem:recovery factorization},  $\Gamma (\tau^{\,}_j) = \Gamma$ for all $j \in [1,\Nmeas]$. However, the operators $\Gamma (t^{\,}_j)$ may be equal to $\Gamma$ times up to $\Nmeas-j$ involutory operators, per Eq.~\ref{eq:1meas Heise meas good}.
\end{proof}

\section{Proof of Proposition~\ref{prop:M=1 useless}}
\label{app:M=1 useless proof}

\begin{proof}
    We first consider the $\numQ=1$ case described in the statement of Prop.~\ref{prop:M=1 useless}. Since $\chan$ only contains a single projective measurement of an effectively involutory observable $\bar{\mobserv}$ \eqref{eq:involutory part} by Prop.~\ref{prop:single-qubit observable}, there is only one associated Stinespring qubit. There can also be no outcome-dependent measurement channels. Written in canonical form $\chan = \QECChannel \MeasChannel \uchan$ \eqref{eq:chan canonical form}, we have
    \begin{align}\label{eq:recovery chan generic}
        \QECChannel = \ident \otimes \SSProj{0}{} + \chi \otimes \SSProj{1}{}
    \end{align}
    for some Pauli-string $\chi$ \eqref{eq:Pauli string def}. By Lemma~\ref{lem:measure = attach}, acting on a logical operator $\Gamma$ in the Heisenberg picture, the combination of $\QECChannel$ and $\MeasChannel$ either attaches $\bar{\mobserv}$ or acts trivially.

    Consider the logical operators $\LX$ and $\LZ$, and Eq.~\ref{eq:teleportation conditions} of Prop.~\ref{prop:teleportation conditions}. If the dilated channels attach $\bar{\mobserv}$ to $\LX$ alone, then $\LZ$ only grows under $\uchan$, so that $\Dist \leq \LRvel T$ for that operator, and hence the logical qubit. The same holds if we switch $\LX$ and $\LZ$. If the dilated channels attach $\bar{\mobserv}$ to \emph{both} $\LX$ and $\LZ$, then they attach $\bar{\mobserv}^2=\ident$ to $\LY$, which therefore obeys $\Dist \leq \LRvel T$. Hence, a single measurement cannot be used to teleport all three logical basis operators a distance $\Dist > \LRvel T$. For $\numQ=\Nmeas=1$, this can also be proven for \emph{any} physical teleportation protocol by showing that Hermiticity of $\StabEl^{\,}_{\nu}$ is equivalent to $\comm{\Pauli{\nu}{i}}{\uchan^{\dagger} \, \bar{\mobserv} \, \uchan} = 0$ for all $\nu$, which implies that at most two logical operators can exceed the Lieb-Robinson bound \eqref{eq:LR bound}.

    Next, consider a standard teleportation protocol $\chan$ with $\numQ \geq 1$ with $\Nmeas < 2 \numQ$. For $\chan$ to succeed with fidelity one, the teleportation conditions \eqref{eq:teleportation conditions} of Prop.~\ref{prop:teleportation conditions} must be satisfied for all $\numQ$ \emph{pairs} of logical operators. To exceed the Lieb-Robinson bound $\Dist \leq \LRvel T$ \eqref{eq:LR bound}, distinct stabilizers $\StabEl^{\,}_{n,\nu}$ must realize for all $2\numQ$ logical basis operators---equivalently, a different involutory operator $\bar{\mobserv}^{\,}_j$ must be attached to all $2\numQ$ logical basis operators. Otherwise, one can define another valid set of logical (basis) operators by taking products of logical operators and/or their associated stabilizers $\StabEl^{\,}_{n,\nu}$ to find one or more logical operators that satisfy $\Pauli{\nu}{i_n} = \uchan^{\dagger} \, \Pauli{\nu}{f_n} \, \uchan$ \eqref{eq:teleportation conditions}, which can only be satisfied if $\Dist \leq d(i_n,f_n) \leq \LRvel T$, as with $\numQ=1$  \cite{SpeedLimit}. 
    
    Alternatively, since Lemma~\ref{lem:compatible observables} \emph{independently} proves that two measurements are sufficient to teleport a logical qubit a distance $\Dist > \LRvel T$, a protocol with $\Nmeas < 2 \numQ$ would require that \emph{at least} one logical qubit is teleported a distance $\Dist > \LRvel T$ using a single measurement, which violates the $\numQ=1$ result proven above.
\end{proof}

\section{Proof of Lemma~\ref{lem:compatible observables}}
\label{app:compatible obs}

\begin{proof}
    Again, we first consider the $\numQ=1$ case, and generalize to $\numQ>1$ afterward. For concreteness, consider the Heisenberg evolution of the logical operators $\LX$ and $\LZ$, which act on the final state $\ket{\Psi^{\,}_T}$ \eqref{eq:teleport final mb state} as $\PX{f}$ and $\PZ{f}$, respectively. In this context,  Lemma~\ref{lem:recovery factorization} guarantees that
    \begin{equation}
        \label{eq:2meas W canonical}
        \chan = \QECChannel^{\vpd}_x \, \QECChannel^{\vpd}_z \, \MeasChannel^{\vpd}_2 \, \MeasChannel^{\vpd}_1 \, \uchan \, ,~~
    \end{equation}
    where the order of the \emph{effective} recovery channels $\QECChannel^{\,}_{x,z}$ is inconsequential, and $\MeasChannel$ and $\uchan$ are unmodified. 
    %, and $\uchan$ realizes time evolution of the physical qubits with total duration (or depth $T$) and Lieb-Robinson velocity $\LRvel$.

    Evolving the logical operators through the recovery channels $\QECChannel$ in the Heisenberg picture leads to 
    \begin{align}
    \label{eq:2meas post recovery}
        \Pauli{\nu}{\rm L}(\tau) &= \QECChannel^{\dagger} \, \Pauli{\nu}{f}\, \QECChannel = \Pauli{\nu}{f}\, \SSZp{1}{\lambda^{\vpp}_{\nu,1}}\, \SSZp{2}{\lambda^{\vpp}_{\nu,2}} \, ,~~
    \end{align}
    where $\lambda^{\vpp}_{\nu,i}=1$ if an odd number of recovery operators $\conjchan^{\,}_{i,k}$ based on the $i$th measurement outcome anticommute with $\Pauli{\nu}{f}$, and is zero otherwise (for $\nu=x,y,z$). Evolving the operator $\Pauli{\nu}{\rm L}(\tau)$ through the two measurement of $\mobserv^{\,}_2$ (and projecting onto the default state $\ket{0}$ of the corresponding Stinespring register) leads to
    \begin{align}
        \Pauli{\nu}{\rm L}(t^{\,}_2) &= \trace\limits_{{\rm ss},2} \left[ \, \MeasChannel^{\dagger}_2 \, \Pauli{\nu}{\rm L} (\tau) \, \MeasChannel^{\vpd}_2 \, \SSProj{0}{2} \, \right] \notag \\
        &= \Pauli{\nu}{f}\, \mobserv^{\lambda^{\vpp}_{\nu,2}}_2 \, \SSZp{1}{\lambda^{\vpp}_{\nu,1}} \, ,~~\label{eq:2meas post M2}
    \end{align}
    where Lemma~\ref{lem:measure = attach} imposes the condition
    \begin{equation}
        \label{eq:2meas logical commute with A2}
        \comm{\mobserv^{\vpp}_2}{\Pauli{\nu}{\rm L}(\tau)} = \comm{\mobserv^{\vpp}_2}{\Pauli{\nu}{f}} = 0 \, ,~~ 
    \end{equation}
    and evolving through the remaining measurement gives
    \begin{align}
        \Pauli{\nu}{\rm L}(t^{\,}_1) &= \trace\limits_{{\rm ss},2} \left[ \, \MeasChannel^{\dagger}_1 \, \Pauli{\nu}{\rm L}(t^{\,}_2)  \, \MeasChannel^{\vpd}_1 \, \SSProj{0}{1} \, \right] \notag \\
        &= \Pauli{\nu}{f}\, \mobserv^{\lambda^{\vpp}_{\nu,2}}_2 \, \mobserv^{\lambda^{\vpp}_{\nu,1}}_1\, ,~~\label{eq:2meas post M1}
    \end{align}
    where, now, Lemma~\ref{lem:measure = attach} imposes the condition
    \begin{equation}
        \label{eq:2meas logical commute with A1}
        \comm{\mobserv^{\vpp}_1}{\Pauli{\nu}{\rm L}(t^{\,}_2)} = \comm{\mobserv^{\vpp}_1}{\Pauli{\nu}{f}\, \mobserv^{\lambda^{\vpp}_{\nu,2}}_2 } = 0 \, ,~~ 
    \end{equation}
    and we now consider the conditions of the Lemma.

    First, $\lambda^{\,}_{\nu,i}=1$ must hold for at least one logical operator labelled $\nu$, for both $i=1,2$; otherwise, one of the measurements acts trivially by Prop.~\ref{prop:feedback required}, and the other is then useless to $\chan$ by Prop.~\ref{prop:M=1 useless}. Second, if $\lambda^{\,}_{\nu,i}=1$ for both $i=1,2$ and $\nu=x,z$, e.g., then  $\LX(t^{\,}_1)$ and $\LZ(t^{\,}_1)$ are both multiplied by $\mobserv^{\,}_2 \, \mobserv^{\,}_1$; however, this is equivalent to a \emph{single} measurement of the involutory operator $B = \mobserv^{\,}_2 \, \mobserv^{\,}_1$, which cannot be useful to teleportation by Prop.~\ref{prop:M=1 useless}. Third, if $\lambda^{\,}_{\nu,i}=0$ for $i=1,2$ for \emph{any} $\nu$, then that logical operator must be teleported by $\uchan$ alone, in which case $\Dist \leq \LRvel T$, violating the conditions of the Lemma.  Hence, $\chan$ must attach one of the observables $\mobserv^{\,}_{1,2}$ to $\LX$ and the other to $\LZ$ (or any pair of logical operators).

    Suppose that $\chan$ attaches $\mobserv^{\,}_2$  to $\LZ$ and $\mobserv^{\,}_1$ to $\LX$---i.e., $\lambda^{\,}_{z,2}=\lambda^{\,}_{x,1}=1$ while $\lambda^{\,}_{x,2}=\lambda^{\,}_{z,1}=0$.  The commutation conditions imposed by Lemma~\ref{lem:measure = attach} take the form
    \begin{subequations}
        \label{eq:2meas nice comms}
    \begin{align}
        \comm{\mobserv^{\vpp}_2}{\PX{f}} &= \comm{\mobserv^{\vpp}_1}{\PX{f}} = 0 \label{eq:2meas nice comm x} \\
        \comm{\mobserv^{\vpp}_2}{\PZ{f}} &= \comm{\mobserv^{\vpp}_1}{\PZ{f} \mobserv^{\vpp}_2} = 0 \label{eq:2meas nice comm z} \, ,~~
    \end{align}
    \end{subequations}
    and we focus on the final relation in Eq.~\ref{eq:2meas nice comm z}, which, by linearity of the commutator can be written
    \begin{equation}
        \label{eq:2meas interesting comm}
        \comm{\mobserv^{\vpp}_1}{\PZ{f} \mobserv^{\vpp}_2} = \comm{\mobserv^{\vpp}_1}{\PZ{f}} \,  \mobserv^{\vpp}_2 + \PZ{f} \, \comm{\mobserv^{\vpp}_1}{\mobserv^{\vpp}_2} = 0 \, . ~~
    \end{equation}
    We note that $\mobserv^{\vpp}_1$ \emph{only} fails to commute with an operator on site $f$ if, in the bare protocol $\chan'$ (which need not be in canonical form), the single-qubit observable $\mobserv_1'$ either (\emph{i}) is supported only on site $f$ or (\emph{ii}) is supported on some site $j \neq f$ where both $j$ and $f$ lie in the support of some physical unitary $\cliff$ that acts after $\MeasChannel_1'$ (in the Schr\"odinger picture). In either case, $\mobserv_1'$ lies in the light cone of site $f$ generated by time evolution alone. Hence, attaching such an operator to any logical operator $\Pauli{\nu}{\rm L}$ cannot be more useful to teleportation than time evolution alone, which obeys $\Dist \leq \LRvel T$ \eqref{eq:LR bound}. However, Prop.~\ref{prop:M=1 useless} establishes that the measurement of $\mobserv^{\,}_2$ cannot lead to $\Dist> \LRvel T$ on its own.

    Hence, both $\mobserv^{\,}_1$ and $\mobserv^{\,}_2$ must act trivially on $f$ (in canonical form or otherwise), so that Eq.~\ref{eq:2meas interesting comm} becomes
    \begin{equation*}\tag{\ref{eq:compatible observ}}
        \comm{\mobserv^{\vpp}_1}{\mobserv^{\vpp}_2} = 0 \, , ~~
    \end{equation*}
    meaning that two observables $\mobserv^{\,}_1$ and $\mobserv^{\,}_2$ are mutually compatible if they commute. This also implies that the channels $\MeasChannel^{\,}_1$ and $\MeasChannel^{\,}_2$ commute, so their order in Eq.~\ref{eq:2meas W canonical} is inconsequential. Moreover, all of the arguments above hold if we replace $\LX$ and $\LZ$ with any pair of logical operators. Note that Eq.~\ref{eq:compatible observ} must hold in the bare protocol and in canonical form \eqref{eq:2meas W canonical}.

    Having proven the Lemma for $\numQ=1$ (i.e., as stated), we now establish that $\Nmeas=2\numQ$ measurements is sufficient to teleport $\numQ \geq 1$ logical qubits a distance $\Dist > \LRvel T$. Again, we assume a standard teleportation protocol $\chan$ in canonical form \eqref{eq:chan canonical form}. Lemma~\ref{lem:recovery factorization} guarantees that
    \begin{align}
        \chan =\Big[ \prod\limits_{n=1}^{\numQ} \, \QECChannel^{\vpd}_{n,x} \, \QECChannel^{\vpd}_{n,x} \Big] \, \MeasChannel^{\vpd}_{2\numQ}\,% \MeasChannel^{\vpd}_{2\numQ-1} 
        \cdots %\MeasChannel^{\vpd}_{2} \, 
        \MeasChannel^{\vpd}_{1} \, \uchan \, ,~
    \end{align}
    acting on logical operators, where the order of the \emph{effective} recovery channels $\QECChannel^{\,}_{n,\nu}$ \eqref{eq:effective recovery} is unimportant, and both $\MeasChannel$ and $\uchan$ are unchanged compared to the true protocol. 

    We next repeat the analysis used in the $\numQ=1$ case, noting that the combination of $\QECChannel$ and $\MeasChannel$ attaches observables to logical operators, per Lemma~\ref{lem:measure = attach}. By Prop.~\ref{prop:M=1 useless}, the same observable cannot be attached to both logical basis operators for a given logical qubit $n$. Moreover, the $\numQ=1$ case above establishes that the attached operator $\bar{\mobserv}^{\,}_{s}$ (the measured observable) evolves under $\uchan$ into the initial logical operator $\Pauli{\nu}{i_n}$ (and some stabilizer element). Hence, if the same operator $\bar{\mobserv}^{\,}_{s}$ is attached to logical operators for the logical qubits $n$ \emph{and} $n'$, then at least one of the logical operators $n$ and $n'$ acts nontrivially on \emph{both} $i_n$ and $i_{n'}$,  violating Eq.~\ref{eq:teleportation conditions}. Hence, we assign a unique observable $\mobserv^{\,}_s$ to each logical basis operator.
    
    The first measured observable encountered in the Heisenberg picture is $\mobserv^{\,}_{2\numQ}$; by Lemma~\ref{lem:measure = attach}, we have that
    \begin{align}
    \label{eq:k>1 first obs comm relation}
        \comm{\mobserv^{\vpp}_{2\numQ}}{\Pauli{\nu}{f_n}} = 0 ~~\forall \, n,\nu \, ,~~
    \end{align}
    and since the label $n$ is meaningless, suppose that $\bar{\mobserv}^{\,}_{2\numQ}$ is attached to $\PX{f_{\numQ}}$ (i.e., $\numQ$th logical $\PX{}$ operator). 

    We next encounter the observable $\mobserv^{\,}_{2\numQ-1}$, finding
    \begin{align}
        \comm{\mobserv^{\vpp}_{2\numQ-1}}{\PX{f_{\numQ}} \, \mobserv^{\vpp}_{2\numQ}} = 0 \, , ~~\label{eq:k>1 second obs comm relation}
    \end{align}
    and note that $\mobserv^{\,}_{2\numQ-1}$ commutes with all other logical operators. Following Eq.~\ref{eq:2meas interesting comm}, this implies that
    \begin{equation}
    \label{eq:k>1 compatability first}
        \comm{\mobserv^{\vpp}_{2\numQ-1}}{\mobserv^{\vpp}_{2\numQ}} = 0 \, ,~~
    \end{equation}
    in analogy to Eq.~\ref{eq:compatible observ}. The logical operator to which $\mobserv^{\,}_{2\numQ-1}$ is attached is unimportant, as long as each measured observable is attached to a unique logical operator. 
    
    As a result, every subsequent observable obeys one or more expressions like Eq.~\ref{eq:k>1 second obs comm relation} by Lemma~\ref{lem:measure = attach}. Since the measured observables always commute with the logical operators (which is required for $\Dist > \LRvel T$), the observable $\mobserv^{\,}_{2\numQ-s}$ obeys Eq.~\ref{eq:k>1 second obs comm relation} for the $s$ logical operators to which prior observables have been attached, and commutes with all other logical operators. Each of the $s$ expressions of the form of Eq.~\ref{eq:k>1 second obs comm relation} reduce to Eq.~\ref{eq:k>1 compatability first} for each of the measured observables already encountered. Repeating this for all $2\numQ$ measurement channels, we find that
    \begin{align}
        \label{eq:k>1 compatible observ}
        \comm{\mobserv^{\vpp}_{j}}{\mobserv^{\vpp}_{j'}}  = 0 ~~\forall \, j,j' \, ,~~
    \end{align}
    i.e., all measured observables mutually commute. The proof that a protocol $\chan$ with $2\numQ$ measurements satisfying Eq.~\ref{eq:k>1 compatible observ} can teleport $\numQ$ logical qubits a distance $\Dist > \LRvel T$ is straightforward: One simply invokes the remainder of the proof for $\numQ=1$ for each logical qubit individually. 
\end{proof}

\section{Proof of Lemma~\ref{lem:M=2 dist}}
\label{app:M=2 speed limit}

\begin{proof}
    The proof of Eq.~\ref{eq:M=2 speed limit} involves several steps, each of which utilizes the limits on operator growth imposed by the Lieb-Robinson bound \eqref{eq:LR bound} in Theorem~\ref{thm:LR theorem}, and builds on the teleportation conditions established in  Prop.~\ref{prop:teleportation conditions},
    \begin{align}
        \label{eq:M=2 teleport conditions}
        \Pauli{\nu}{\rm L} (T) &\equiv \trace\limits_{\rm ss} \left[ \, \chan^{\dagger} \, \Pauli{\nu}{f} \, \chan \, \SSProj{\bvec{0}}{\rm ss} \, \right] \notag \\
        &= \uchan^{\dagger} \, \mobserv^{\vpp}_{\nu,j} \, \Pauli{\nu}{f} \, \uchan = \Pauli{\nu}{i} \, \StabEl^{\vpp}_{\nu} \, ,~~
    \end{align}
    where $\Pauli{\nu}{\rm L} (T)$ is the physical part of the $\nu$-type logical operator acting on the \emph{initial} state and the final equality follows from Eq.~\ref{eq:teleportation conditions}. Lemma~\ref{lem:compatible observables} guarantees that exactly one of the two measured observables is attached to $\Pauli{\nu}{\rm L}$; the attached operator is given by
    \begin{equation*}
        \tag{\ref{eq:CF observable}}
        \mobserv^{\vpp}_{\nu,j} \equiv \Cliff{\nu} \, \bar{\mobserv}_{\nu j}' \, \CliffDag{\nu} \, , ~~
    \end{equation*}
    where $\bar{\mobserv}_j'$ is the involutory part \eqref{eq:involutory part} of the single-qubit observable $\mobserv^{\,}_j$ acting on qubit $j$ (which depends on $\nu$), and $\cliff^{\,}_{\nu}$ comprises all physical unitaries $\cliff %\in \Aut*{\Hilbert^{\,}_{\rm ph}}
    $ \emph{after} the measurement of $\mobserv_{\nu j}'$ in the na\"ive protocol $\chan'$, in the Schr\"odinger picture. We can then write $\uchan = \cliff^{\,}_{\nu} \, \uchan^{\,}_{\nu}$, so that $\uchan^{\,}_{\nu}$ comprises all physical unitaries in $\chan'$ \emph{prior} to the measurement of $\mobserv_{\nu j}'$, where $\cliff^{\,}_{\nu}$ has depth $\tau^{\,}_{\nu}$, $\uchan^{\,}_{\nu}$ has depth $T^{\,}_{\nu}$, and $T = \tau^{\,}_{\nu}+T^{\,}_{\nu}$. We then rewrite Eq.~\ref{eq:M=2 teleport conditions} as
    \begin{equation}
        \label{eq:M=2 teleport conditions nice}
        \Pauli{\nu}{i} \, \StabEl^{\vpp}_{\nu} = \left( \uchan^{\dagger}_{\nu} \, \bar{\mobserv}_{\nu j}' \, \uchan^{\vpd}_{\nu} \right) \, \left( \uchan^{\dagger} \, \Pauli{\nu}{f} \, \uchan \right) \, ,~~
    \end{equation}
    and it will prove convenient to define
    \begin{equation}
        \label{eq:M=2 U evolved logical}
        \Sigma^{\nu}_{\ell} \equiv \uchan^{\dagger} \, \Pauli{\nu}{f} \, \uchan \, , ~~
    \end{equation}
    and the Lemma's requirement that $\Dist > \LRvel T$ implies that $d(\ell,i) \geq 1$, so $\Sigma^{\nu}_{\ell}$ cannot be $\Pauli{\nu}{\rm L} (T)$ \eqref{eq:M=2 teleport conditions}. 

    Next, we establish by contradiction that the light cones of $\bar{\mobserv}_{\nu j}'$ and $\Pauli{\nu}{f}$ in Eq.~\ref{eq:M=2 teleport conditions nice} have to meet. If the light cones do not meet, then since  $\Sigma^{\nu}_{\ell}$ cannot realize $\Pauli{\nu}{i}$, it must realize part of the stabilizer element $\StabEl^{\,}_{\nu}$ in Eq.~\ref{eq:M=2 teleport conditions nice}, while $\uchan^{\dagger}_{\nu} \, \bar{\mobserv}_{\nu j}' \, \uchan^{\vpd}_{\nu}$ realizes $\Pauli{\nu}{i}$ times another part of $\StabEl^{\,}_{\nu}$, so that
    \begin{align}
        \StabEl^{\vpp}_{\nu} &= \StabEl^{(1)}_{\nu} \otimes \StabEl^{(2)}_{\nu} = \Pauli{\nu}{i} \, \uchan^{\dagger}_{\nu} \, \bar{\mobserv}_{\nu j}' \, \uchan^{\vpd}_{\nu} \otimes \Sigma^{\nu}_{\ell} \, , ~~\label{eq:M=2 disjoint stab}
    \end{align}
    where $\StabEl^{(1)}_{\nu}, \StabEl^{(2)}_{\nu} \in \UStabOf{\ket{\Psi^{\,}_0}}$ \eqref{eq:UStab} have disjoint support. 
    
    Note that any two elements $\StabEl^{\,}_1,\StabEl^{\,}_2$ of $\UStabOf*{\ket{\Psi}}$ satisfy
    \begin{align}\label{eq:stab acomm norm}
        \norm{\acomm{\StabEl^{\vps}_1}{\StabEl^{\vps}_2}} = 2 \, ,~~
    \end{align}
    which we prove for all $\StabEl^{\,}_1,\StabEl^{\,}_2 \in \UStabOf*{\ket{\Psi}}$. First, note that $\acomm{\StabEl^{\,}_1}{\StabEl^{\,}_2} \, \ket{\Psi} = 2 \, \ket{\Psi}$, which implies that $\norm{\acomm{\StabEl^{\,}_1}{\StabEl^{\,}_2}} \geq 2$ (since there exists an eigenvalue of the anticommutator with eigenvalue at least two). Second, note that unitarity of the stabilizer elements---along with the triangle inequality---ensure that $\norm{\acomm{\StabEl^{\,}_1}{\StabEl^{\,}_2}} \leq 2 \, \norm{\StabEl^{\,}_1 \, \StabEl^{\,}_2} = 2$. 
    
    Crucially, the initial state $\ket{\Psi^{\,}_0}$ \eqref{eq:teleport initial mb state}  is either a product state, or is separable with respect to the support of the two stabilizer elements. However, Eq.~\ref{eq:M=2 disjoint stab} implies that
    \begin{equation}
        \label{eq:M=2 bad stab acomm}
        \acomm{\StabEl^{(2)}_{\mu}}{\StabEl^{(2)}_{\nu}} = \uchan^{\dagger} \, \acomm{\Pauli{\mu}{f}}{\Pauli{\nu}{f}} \, \uchan = 2 \, \kron{\mu,\nu} \, ,~~
    \end{equation}
    which is incompatible with Eq.~\ref{eq:stab acomm norm} for $\mu \neq \nu$, contradicting the claim that $\StabEl^{(2)}_{\mu,\nu}$ are elements of $\UStabOf{\ket{\Psi^{\,}_0}}$ \eqref{eq:UStab}. Hence, the light cones of the operators $\bar{\mobserv}_{\nu j}'$ and $\Pauli{\nu}{f}$ must meet. 

    Moreover, the operators $\Pauli{\nu}{\rm L}(T)$ \eqref{eq:M=2 teleport conditions} must realize the Pauli algebra on site $i$ only, and commute everywhere else. The most efficient strategy to ensure this is to evolve the commuting operators $\mobserv_{\nu j}'$ (for $\nu=x,z$) to nearby sites (using a circuit with depth $\tau^{\,}_{0}$, and apply a unitary ``decoding'' channel $\mathcal{Q}$ that maps, e.g., 
    \begin{equation}
    \label{eq:M=2 decoding}
        \mobserv^{\prime}_{\nu j} \to \Pauli{\nu}{j} \, \Pauli{\nu}{j+\LRvel} \, .~~
    \end{equation}
    As established in Ref.~\citenum{SpeedLimit}, the most efficient strategy is to evolve $\Pauli{\nu}{j} \to \Pauli{\nu}{i}$ and $\Pauli{\nu}{j+\LRvel} \to \Pauli{\nu}{\ell-\LRvel}$, with
    \begin{equation}
    \label{eq:M=2 link up tau}
        \uchan^{\dagger} \, \Pauli{\nu}{f} \, \uchan = \Sigma^{\nu}_{\ell} = \Pauli{\nu}{\ell} \, , ~~
    \end{equation}
    so that $\Pauli{\nu}{\ell}\, \Pauli{\nu}{\ell-\LRvel}$ can be converted into stabilizers of $\ket{\Psi^{\,}_0}$ \eqref{eq:teleport initial mb state} for both $\nu=x,z$ by the ``encoding'' channel $\mathcal{Q}^{\dagger}$. 
    
    Using the constraints above, we bound the teleportation distance $\Dist = d(i,f)$ \eqref{eq:L def}. We first note that 
    \begin{equation}
        \label{eq:M=2 f travel}
        f - \ell \leq \LRvel \left( T - 1 \right)  \, , ~~
    \end{equation}
    since at least one layer of the circuit must convert $\Pauli{\nu}{\ell-\LRvel}$ and $\Pauli{\nu}{\ell}$ into a stabilizer. We then note that
    \begin{subequations}
    \label{eq:M=2 k dists}
        \begin{align}
            j - i &\leq \LRvel \,\left(  T - \max_{\nu} \tau^{\vpp}_{\nu} -  \tau^{\vpp}_{0} - 1 \right)  \label{eq:M=2 ik dist} \\
            \ell - j &\leq \LRvel \left( T - \max_{\nu} \tau^{\vpp}_{\nu}  -  \tau^{\vpp}_{0}  \right)  \label{eq:M=2 kn dist} \, ,~~
        \end{align}
    \end{subequations}
    where $\max_\nu \tau^{\,}_{\nu}$ is the total depth of all physical unitary channels following the first measurement (in the Schr\"odinger picture). By summing the foregoing inequalities, we obtain the following bound on $\Dist$, 
    \begin{equation}
        \label{eq:M=2 dist 1}
        \Dist = d(i,f) \leq 3 \LRvel T - 2 \LRvel \left( \max_\nu \tau^{\,}_\nu + \tau^{\vpp}_{0} + 1 \right)  \, ,
    \end{equation}
    where the depth $\tau^{\,}_{0} \geq 0$ increases if the measurements of $\mobserv_{\nu,j}'$  are either (\emph{i}) far from one another (i.e., $d(j^{\,}_1,j^{\,}_2) \gg 1$) or (\emph{ii}) far from the optimal position $j^* \approx i + T/3$. Taking $\tau^{\,}_{0} \to 0$ in Eq.~\ref{eq:M=2 dist 1} gives the (potentially loose) bound
    \begin{equation*}
        \tag{\ref{eq:M=2 speed limit}}
        \Dist = d(i,f) \leq 3 \LRvel T - 2 \LRvel \left( \max_\nu \tau^{\vpp}_{\nu} + 1 \right) \, .~
    \end{equation*}
\end{proof}

\section{Proof of Lemma~\ref{lem:M>2 speedup}}
\label{app:proof M>2 speedup}

\begin{proof}
The proof of Lemma~\ref{lem:M>2 speedup} closely resembles that of Lemma~\ref{lem:M=2 dist} (see App.~\ref{app:M=2 speed limit}). Again, the light cone emanating from one pair of measurements must intersect the light cone of the final-state logical operators $\Pauli{\nu}{f}(T)$ in the Heisenberg picture, while the light cone emanating from a (generically distinct) pair of measurements must reach the initial logical site $i$. For $\Nmeas=4$, the light cones of these two measurement regions must intersect each other; for $\Nmeas = 2m$ with $m>2$, we instead ``daisy chain'' the various light cones of the measured observables, in series. 

As in the $\Nmeas=2$ case described by Lemma~\ref{lem:M>2 speedup}, we need only consider the \emph{optimal} scenario. This means that pairs of measurements are well separated, and all of the measured observables commute with one another and all logical operators. This applies both to the canonical-form observables $\mobserv^{\,}_j$ \eqref{eq:CF observable} and the single-qubit observables $\mobserv_j'$ of the na\"ive protocol $\chan'$. Evolving the logical through the dilated channels $\QECChannel$ and $\MeasChannel$ leads to
\begin{equation}
\label{eq:gen M link up logical t}
    \Pauli{\nu}{f}(t) = \Pauli{\nu}{f} \, \prod\limits_{s=1}^{\Nmeas} \, \Cliff{s,\nu} \, \bar{\mobserv}^{\vpd}_{s,\nu} \, \CliffDag{s,\nu} \, ,~~
\end{equation}
where $s$ runs over pairs of measurements (one of which is assigned to $\LX$ and the other to $\LZ$), and $\cliff^{\,}_{s,\nu}$ captures all unitary evolution following (in the Schr\"odinger picture) the measurement of the single-qubit operator  $\mobserv_{s,\nu}$ in the \emph{na\"ive} protocol $\chan'$ (and has depth $\tau^{\,}_{s,\nu}$). 

Evolving through $\chan$ again leads to Eq.~\ref{eq:M=2 link up tau} for the $\Pauli{\nu}{f} $ term in Eq.~\ref{eq:gen M link up logical t}, which is moved to some site $\ell$ obeying
\begin{equation*}
    \tag{\ref{eq:M=2 f travel}}
    f - \ell \leq \LRvel \left( T - 1 \right)  \, , ~~
\end{equation*}
as before. Again, it is optimal to measure nearby sites---e.g., $j^{\,}_{s,z} = j^{\,}_{s,x}+\LRvel$, as in Fig.~\ref{fig:cluster state circuit} and Eq.~\ref{eq:M=2 decoding}. We denote by $j^{\,}_s$ the left site of the pair $s$. As in Lemma~\ref{lem:M>2 speedup}, we have
\begin{align*}
    d(j^{\vpp}_1,i) &\leq \LRvel \,\left(  T - \max_{\nu} \tau^{\vpp}_{1,\nu} - 1 \right)  \tag{\ref{eq:M=2 ik dist}} \\
    d(\ell,j^{\vpp}_{\Nregions}) &\leq \LRvel \left( T - \max_{\nu} \tau^{\vpp}_{\Nregions,\nu}  -  1  \right)  \tag{\ref{eq:M=2 kn dist}} \, ,~~
\end{align*}
where $\Nregions=\Nmeas/2$ is the number of pairs of nearby measurements, and we take $\tau^{\vpp}_0= 0$ compared to Eq.~\ref{eq:M=2 k dists}. The new ingredient compared to Lemma~\ref{lem:M>2 speedup} is
\begin{align}
\label{eq:M>2 jj dist}
    d(j^{\vpp}_{s+1} , j^{\vpp}_s) \leq \LRvel \left( T -  \max_{\nu} \max ( \tau^{\vpp}_{s,\nu}, \tau^{\vpp}_{s+1,\nu} ) -1 \right)  \, , 
\end{align}
and adding these all up---and taking $\tau^{\,}_{s,\nu}=0$ to reflect optimal protocols---we find that
\begin{align}
    \Dist &= d(i,f) \leq 3 \LRvel (T - 1) + 2 \LRvel \sum\limits_{s=1}^{\Nregions-1} ( T -  1 ) \, ,~
\end{align}
and noting that adding a single additional measurement cannot alter this bound, we replace $\Nregions=\floor{\Nmeas/2}$ to recover
\begin{align*}\tag{\ref{eq:k=1 speed limit}}
    \Dist \leq \LRvel T + 2\LRvel \left\lfloor \frac{\Nmeas}{2} \right\rfloor \left( T - 1 \right) 
    \, .~
\end{align*}
\end{proof}

\section{Proof of Theorem~\ref{thm:standard bound}}
\label{app:standard bound proof}

\begin{proof}
    To prove Eq.~\ref{eq:k>1 bound}, we derive the optimal spacing of measurements for a given logical qubit $n$, as in the proofs of Lemmas~\ref{lem:M=2 dist} and \ref{lem:M>2 speedup}. This is straightforward, as Lemmas~\ref{lem:recovery factorization} and \ref{lem:measure = attach} jointly guarantee that the combination of measurements and outcome-dependent operations do not ``mix'' between logical operators. Moreover, Lemma~\ref{lem:compatible observables}---whose proof in App.~\ref{app:compatible obs} applies to all $\numQ \geq 1$---guarantees that all of the measured observables $\mobserv^{\,}_j$ \eqref{eq:CF observable}---and their involutory parts $\bar{\mobserv}^{\,}_j$ \eqref{eq:involutory part}---commute with one another and with all final-state logical operators, by Lemma~\ref{lem:measure = attach}.
    
    Evolving all logical operators through the channels $\QECChannel$ and $\MeasChannel$ in the Heisenberg picture, projecting onto the default Stinespring initial state $\ket{\bvec{0}}^{\,}_{\rm ss}$, and tracing out the Stinespring degrees of freedom leads to
    \begin{align}
        \Pauli{\nu}{f_n} (t) &= \trace\limits_{\rm ss} \left[ \, \MeasChannel^{\dagger} \, \QECChannel \,  \Pauli{\nu}{f_n} \, \QECChannel \, \MeasChannel \, \SSProj{\bvec{0}}{\rm ss} \, \right] \notag \\
        &= \bar{\mobserv}^{\vpp}_{1,n,\nu} \cdots \bar{\mobserv}^{\vpp}_{\Nregions,n,\nu} \, \Pauli{\nu}{f_n} \, ,~~
        \label{eq:k>1 speed limit post meas logical}
    \end{align}
    where the first label on $\bar{\mobserv}^{\,}_{s,n,\nu}$ corresponds to the ``measurement region'' $s \in [1,\Nregions]$ \cite{SpeedLimit}, $n$ labels the logical qubit, and $\nu$ labels the logical basis operator (e.g., $\nu \in \{x,z\}$). 
    
    We now work out the optimal spacing between measurement sites $j^{\,}_{s,n,\nu}$ for different $s$. Following the proofs in Apps.~\ref{app:M=2 speed limit} and \ref{app:proof M>2 speedup}, we note that the logical basis operators $\Pauli{\nu}{f_n}(t)$ \eqref{eq:k>1 speed limit post meas logical} still realize the Pauli algebra on site $f_n \gg i_n$. As before, the strategy is for the pair observables $\mobserv^{\,}_{s,n,x}$ and $\mobserv^{\,}_{s,n,z}$ to be measured on nearby sites $j^{\,}_{s,n}$ and $j^{\,}_{s,n} +\LRvel$, and then converted into a pair of operators of the same type ($\nu$) under the next unitary channel encountered in the Heisenberg picture---i.e., 
    \begin{align}
    \label{eq:k>1 speed limit decode meas}
        \bar{\mobserv}^{\vpp}_{s,n,\nu} \to \Pauli{\nu}{j^{\vpp}_{s,n}} \, \Pauli{\nu}{j^{\vpp}_{s,n}+\LRvel} ~~~{\rm and}~~~\Pauli{\nu}{f_n} \to \Pauli{\nu}{f_n-\LRvel} \, ,~~
    \end{align}
    which can always be achieved in duration $\Delta T \geq 1$ \cite{SpeedLimit}. 
 
    As before, the next step is to send the pair of $\Pauli{\nu}{}$ operators created on nearby sites via \eqref{eq:k>1 speed limit decode meas} in opposite directions. Note that we need to ``save'' $\Delta T \geq 1$ of the total unitary evolution to ``undo'' Eq.~\ref{eq:k>1 speed limit decode meas}. The next segment of the Heisenberg evolution of the logical operator \eqref{eq:k>1 speed limit post meas logical}, with duration $T-2$, sends $f_n -\LRvel \to f_n - \LRvel (T-1)$, $j^{\,}_{s,n} \to j^{\,}_{s,n} - \LRvel (T-1)$, and $j^{\,}_{s,n} + \LRvel \to j^{\,}_{s,n} + \LRvel (T-1)$. At this point, the sites of anticommutation must be ``undone'' (since the initial state is a product state, by Def.~\ref{def:standard teleport}), which is possible using the inverse of the channel that achieves Eq.~\ref{eq:k>1 speed limit decode meas}. For this to succeed---and for $\Pauli{\nu}{f_n}(T)$ to realize the Pauli algebra on site $i_n$---we must have
    \begin{subequations}
        \label{eq:k>1 speed limit spacings}
    \begin{align}
        d(f^{\vpp}_n,j^{\vpp}_{\Nregions,n}) &\leq 2 \LRvel T - \LRvel \label{eq:k>1 speed limit f,j spacing} \\
        d(j^{\vpp}_{s+1,n},j^{\vpp}_{s,n}) &\leq 2 \LRvel (T-1) \label{eq:k>1 speed limit j,j spacing} \\
        d(j^{\vpp}_{1,n},i^{\vpp}_n) &\leq \LRvel (T-1) \, ,~~\label{eq:k>1 speed limit j,i spacing}
    \end{align}
    \end{subequations}
    where the middle line holds for $1 \leq s < \Nregions$. These relations follow from the requirement that all light cones come within $\LRvel$ of one another prior to the final channel in the Heisenberg picture (with depth $\Delta T \geq 1$), except that of the leftmost $\Pauli{\nu}{}$ operator generated in Eq.~\ref{eq:k>1 speed limit decode meas}, which can propagate an extra $\LRvel$ sites to $i_n$. Putting this all together, a protocol with $\Nregions$  measurement regions teleports the $n$th logical qubit a distance $\Dist$ \eqref{eq:L def} satisfying
    \begin{align}
        \label{eq:k>1 proto speed limit}
        \Dist \leq \LRvel T + 2 \Nregions \LRvel (T-1) \, ,~~
    \end{align}
    and now, we consider the remaining $\numQ-1$ logical qubits. 

    To proceed, we first recall several previously established facts. First, Prop.~\ref{prop:M=1 useless} proves that (\emph{i}) the \emph{same} measured observables cannot be attached to distinct logical basis operators and remain useful to teleportation and (\emph{ii}) a given region requires \emph{two} measurements per logical qubit. This implies that each of the $\Nregions$ measurement ``regions'' must involve $2\numQ$ single-qubit measurements; moreover, pairs of measurements corresponding to the same logical operator should neighbor one another (i.e., be no further than $\LRvel$ from each other). The optimal strategy is to have the sites $i_n$ (and $f_n$) for different $n$ correspond to alternating sites (in order), and for each measurement region to comprise $2\numQ$ consecutive sites, of which the pair of sites $2n-1,2n$ correspond to the observables attached to the logical basis operators for the $n$th logical qubit. As a result, the total number of measurements required is $\Nmeas = 2 \Nregions \numQ$; by Prop.~\ref{prop:M=1 useless} and Def.~\ref{def:task dist}, only an additional $2\numQ$ measurements can increase $\Dist$, and thus we replace $\Nregions$ with $\floor{\Nmeas/2\numQ}$ in Eq.~\ref{eq:k>1 proto speed limit} to recover Eq~.\ref{eq:k>1 bound}. 

    Finally, we prove Eq.~\ref{eq:k>1 min depth}. As always, the maximum distance between equivalent sites $j^{\,}_{s,n}$ and $j^{\,}_{s+1,n}$ is $\ell \leq 2 \LRvel (T-1)$. However, we note that at least $2\numQ$ sites must be measured in each region; hence, we must have $\ell \geq 2 \numQ$. Combining these inequalities, we find that $2 \LRvel (T-1) \geq 2 \numQ$, which simplifies to Eq.~\ref{eq:k>1 min depth}. In the case $\numQ=1$, we find $T \geq 2$, as noted in Lemmas~\ref{lem:M=2 dist} and \ref{lem:M>2 speedup}.
\end{proof}

\section{Proof of Corollary~\ref{cor:meas comm}}
\label{app:proof meas comm}

\begin{proof}
    The first important result of Lemma~\ref{lem:measure = attach} is
    \begin{equation*}
        \tag{\ref{eq:obs logical commute}}
        \comm{\mobserv^{\vpp}_j}{\Gamma (t^{\vpp}_j)} = 0 \, , ~~
    \end{equation*}
    where $\Gamma (t^{\,}_j)$ is some final-state logical operator $\Gamma$ evolved to the (Heisenberg) time $t^{\,}_j$ immediately prior to the measurement of $\mobserv^{\,}_j$. The other important result is
    \begin{equation*}
    \tag{\ref{eq:meas attach}}
        \Gamma (t^{\vpp}_{j-1}) = \trace_{{\rm ss}~j} \left[ \, \umeas^{\dagger}_j \, \Gamma (t^{\vpp}_j) \, \umeas^{\vpd}_j \, \right]  = \bar{\mobserv}^{\lambda^{\,}_{j,\Gamma}}_j \, \Gamma (t^{\vpp}_j) \, , ~
    \end{equation*}
    where $\lambda^{\,}_{j,\Gamma} \in \{0,1\}$ is determined by the recovery channels conditioned on the state of the $j$th Stinespring qubit, and $\bar{\mobserv}^{\,}_j$ is the involutory part \eqref{eq:involutory part} of the observable $\mobserv^{\,}_j$.

    The final measured observable $\mobserv^{\,}_{\Nmeas}$ therefore satisfies 
    \begin{align}
    \label{eq:meascomm M logical condition}
        \comm{\mobserv^{\vpp}_{\Nmeas}}{\Pauli{\nu}{f}} = 0 \, ,~~
    \end{align}
    for all final logical sites $f \in F$ and all  $\nu \in \{x,y,z\}$, which requires that $\mobserv^{\,}_{\Nmeas}$ has no support in $F$ (since only $\ident_f$ commutes with $\Pauli{\nu}{f}$ for all nontrivial values of $\nu$). 

    The penultimate observable $\mobserv^{\,}_{\Nmeas-1}$ also has no support on any logical site $f \in F$ to which $\mobserv^{\,}_{\Nmeas}$ is not attached; for any logical operator $\Gamma$ to which $\mobserv^{\,}_{\Nmeas}$ \emph{is} attached,
    \begin{align}
        0 &= \comm{\mobserv^{\vpp}_{\Nmeas-1}}{\bar{\mobserv}^{\vpp}_{\Nmeas} \, \Gamma} \notag \\
        &= \comm{\mobserv^{\vpp}_{\Nmeas-1}}{\bar{\mobserv}^{\vpp}_{\Nmeas}}  \, \Gamma + \bar{\mobserv}^{\vpp}_{\Nmeas} \,\comm{\mobserv^{\vpp}_{\Nmeas-1}}{\Gamma} \, ,~\label{eq:meascomm M-1 expanded comm}
    \end{align}
    where at least one of the two commutators in the sum above must vanish, as $\mobserv^{\,}_{\Nmeas-1}$ \eqref{eq:CF observable} remains a single-qubit operator in the canonical-form protocol $\chan$ \eqref{eq:chan canonical form} by the requirement of Def.~\ref{def:standard teleport} that no unitaries be applied to previously measured sites. Hence, both commutators must vanish, meaning that
    \begin{equation}
    \label{eq:compatible observ 2}
         \comm{\mobserv^{\,}_{\Nmeas-1}}{\mobserv^{\,}_{\Nmeas}} = 0 \, ,~~
    \end{equation}
    and also that $\mobserv^{\,}_{\Nmeas-1}$ has no support in $F$.

    Repeating this analysis for all subsequent observables $\mobserv^{\,}_j$, we again find constraints of the same form as Eq.~\ref{eq:meascomm M-1 expanded comm}, with $\bar{\mobserv}^{\,}_{\Nmeas}$ replaced by some product of observables $\bar{\mobserv}^{\vpp}_{j+1}$ through $\bar{\mobserv}^{\vpp}_{\Nmeas}$. For each $j< \Nmeas$, we find that none of the measured observables previously encountered in the Heisenberg evolution of $\Gamma$  have support in $F$. As a result, Eq.~\ref{eq:obs logical commute} can be split into a vanishing sum of two commutators \`a la Eq.~\ref{eq:meascomm M-1 expanded comm}, where at least one of the commutators must be zero. Hence, both commutators must vanish, meaning that $\mobserv^{\,}_j$ has no support in $F$ and 
    \begin{equation}
        \label{eq:meascomm j comm condition}
        \comm{\mobserv^{\vpp}_j}{\bar{\mobserv}^{\lambda_{j+1,\Gamma}}_{j+1} \cdots \bar{\mobserv}^{\lambda_{\Nmeas,\Gamma}}_{\Nmeas}} = 0 \, ,~~
    \end{equation}
    meaning that $\mobserv^{\,}_j$ commutes with \emph{every} product of observables $\mobserv^{\,}_{j+1}$ through $\mobserv^{\,}_{\Nmeas}$ attached to \emph{any} logical operator $\Gamma = \Pauli{\nu}{f}$ via Eq.~\ref{eq:meas attach} of Lemma~\ref{lem:measure = attach}. 

    Importantly, all of these observables are single-qubit operators, and trivially commute unless they act on the same qubit. By Eq.~\ref{eq:meascomm j comm condition}, either \emph{all} observables correspond to different sites $j \notin F$ \emph{or}, following the measurement of $\mobserv^{\,}_j$ (in the Schr\"odinger picture), two observables that anticommute with $\mobserv^{\,}_j$ are subsequently measured \emph{and} both are attached to the same logical operators. However, these observables must commute with each other, and hence must be the same, so the product of their involutory parts \eqref{eq:involutory part} is the identity. Hence, this pair of measurements can be omitted entirely. Any measurements intermediate to this pair act on other sites, and hence, trivially commute.
    
    This also agrees with the basic intuition that the  subsequent anticommuting measurements necessarily result in random outcomes, yet the final logical state must be the same along \emph{every} outcome trajectory, meaning that no nontrivial operations can be conditioned on the outcomes of the measurements that anticommute with $\mobserv^{\,}_j$. More important, however, is the fact that this anticommuting operator is effectively removed from \emph{all} logical operators prior to encountering  $\mobserv^{\,}_j$ (in the Heisenberg picture), and hence, these measurements may be omitted entirely without affecting any commutation condition \eqref{eq:meascomm j comm condition} on the observables or the evolution of any logical operator. 
\end{proof}

\section{Proof of Lemma~\ref{lem:k=1 big string order}}
\label{app:proof k=1 big string order}

\begin{proof}
    The resource state is given by
    \begin{equation*}
        \tag{\ref{eq:resource state}}
        \ket{\Psi^{\vpp}_t} = \uchan \, \ket{\Psi^{\vpp}_0} \, , ~~
    \end{equation*}
    where $\ket{\Psi^{\,}_0} = \ket{\psi} \otimes \ket{\Phi^{\,}_0}$ \eqref{eq:teleport initial mb state} is the initial state, with $\ket{\Phi^{\,}_0}$ a product state by Def.~\ref{def:phys teleport}. Importantly, the teleportation conditions \eqref{eq:teleportation conditions} of Prop.~\ref{prop:teleportation conditions} guarantee that
    \begin{equation*}
        \tag{\ref{eq:teleportation conditions}}
        \Pauli{\nu}{i} \,\StabEl^{\vpp}_{\nu} = \trace_{\rm ss} \left[ \, \chan^{\dagger} \, \Pauli{\nu}{f} \,\chan \, \SSProj{\bvec{0}}{\rm ss} \, \right] \, ,~~
    \end{equation*}
    where $\StabEl^{\,}_{\nu} \, \ket{\Psi^{\,}_0} = \ket{\Psi^{\,}_0}$ is a stabilizer $\StabEl^{\,}_{\nu} \in \UStabOf{\ket{\Psi^{\,}_0}}$.

    The logical operator that acts on the resource state is
    \begin{align}
        \Pauli{\nu}{f}(T-t) &= \trace_{\rm ss} \left[ \, \MeasChannel^{\dagger} \, \QECChannel^{\dagger} \,\Pauli{\nu}{f} \, \QECChannel \, \MeasChannel \, \SSProj{\bvec{0}}{\rm ss} \, \right] \notag \\
        &= \bar{\mobserv}_1^{\lambda^{\,}_{1,\nu}} \cdots \bar{\mobserv}_{\Nmeas}^{\lambda^{\,}_{\Nmeas,\nu}} \, \Pauli{\nu}{f} = \bar{\mobserv}^{\vpp}_{\nu} \, \Pauli{\nu}{f} \notag \\
        &= \left( \uchan \, \Pauli{\nu}{i}  \, \uchan^{\dagger} \right) \, \left( \uchan \, \StabEl^{\vpp}_{\nu} \, \uchan^{\dagger} \right) \, , ~~\label{eq:k=1 big string order logical facts}
    \end{align}
    by Eqs.~\ref{eq:teleportation conditions} and \ref{eq:meas attach}, and we define
    \begin{align}
    \label{eq:k=1 big string resource stab}
        \StabEl^{\vpp}_{\nu,t} &\equiv \uchan \, \StabEl^{\,}_{\nu} \, \uchan^{\dagger}
    \end{align}
    which is an element of the (unitary) stabilizer group $\UStabOf{\ket{\Psi^{\,}_t}}$ for the \emph{resource} state, since
    \begin{align}
        \StabEl^{\vpp}_{\nu,t} \, \ket{\Psi^{\vpp}_t} &= \left( \uchan \, \StabEl^{\vpp}_{\nu} \, \uchan^\dagger \right) \, \uchan \, \ket{\Psi^{\vpp}_0} \notag \\
        &= \uchan \, \StabEl^{\vpp}_{\nu}  \, \ket{\Psi^{\vpp}_0} = \uchan \, \ket{\Psi^{\vpp}_0} = \ket{\Psi^{\vpp}_t} \, ,~~\label{eq:k=1 big string stab is stab}
    \end{align}
    as expected. Using Eqs.~\ref{eq:k=1 big string order logical facts} and \ref{eq:k=1 big string resource stab}, we have
    \begin{align}
    \label{eq:k=1 big string resource stab nice}
         \StabEl^{\vpp}_{\nu,t} &= \uchan \, \Pauli{\nu}{i} \, \uchan^{\dagger} \, \bar{\mobserv}^{\vpp}_{\nu} \, \Pauli{\nu}{f} \, , ~~
    \end{align}
    and thus, by Eq.~\ref{eq:k=1 big string stab is stab}, we have that
    \begin{align*}
        1 &= \matel{\Psi^{\vpp}_t}{\, \StabEl^{\vpp}_{\nu,t} \, }{\Psi^{\vpp}_t}  \\
        &= \matel{\Psi^{\vpp}_t}{\, \uchan \, \Pauli{\nu}{i} \, \uchan^{\dagger} \, \bar{\mobserv}^{\vpp}_{\nu} \, \Pauli{\nu}{f} \, }{\Psi^{\vpp}_t} 
        \tag{\ref{eq:k=1 big string order}} \, ,~~
    \end{align*}
    as stated in the Lemma. The commutation of the ``bulk'' operator content of the string order parameter $\StabEl^{\,}_{\nu,t}$ \eqref{eq:k=1 big string resource stab nice} is guaranteed by Cor.~\ref{cor:meas comm} for all \emph{necessary} measurements, where we simply omit unnecessary measurements. The anticommutation of the endpoint operators follows automatically from the Pauli algebra.
\end{proof}

\section{Proof of Lemma~\ref{lem:k=1 small string order}}
\label{app:k=1 smol string proof}

\begin{proof}
    We first show that there exist subregions of the interval $[i,f]$ with different \emph{left} endpoints such that a string order parameter analogous to Eq.~\ref{eq:k=1 big string order} has unit expectation value in the \emph{same} resource state $\ket{\Psi^{\,}_t}$ \eqref{eq:resource state}. This requires that multiple observables $\bar{\mobserv}^{\,}_j$ are attached to both logical operators, which is guaranteed when $\Dist \gg \LRvel T$. We show the same for different \emph{right} endpoints subsequently.

    The combination of Eqs.~\ref{eq:teleportation conditions} and \ref{eq:meas attach} imply that
    \begin{align}
    \label{eq:k=1 small string order starting condish}
        \Pauli{\nu}{i} \, \StabEl^{\vpp}_{\nu} &= \left( \uchan^{\dagger} \, \bar{\mobserv}^{\vpp}_{\nu} \, \uchan \right) \left( \uchan^\dagger \, \Pauli{\nu}{f} \, \uchan \right) \, , ~~
    \end{align}
    and importantly, the fact that $\Dist > \LRvel T$ (as required for physical teleportation via Def.~\ref{def:phys teleport}) guarantees that $\uchan^\dagger \, \Pauli{\nu}{f} \, \uchan$ only contributes to the stabilizer part $\StabEl^{\,}_{\nu}$ of Eq.~\ref{eq:k=1 small string order starting condish}. Hence, the logical part $\Pauli{\nu}{i}$ must recover from evolving the leftmost measured observables in $\bar{\mobserv}^{\,}_{\nu}$ \eqref{eq:k=1 big string bulk} under $\uchan$. 
    
    This product of operators can be rewritten as
    \begin{equation}
        \label{eq:k=1 small string order clustered A}
        \bar{\mobserv}^{\vpp}_{\nu} = \bar{\mobserv}^{\vpp}_{\nu,1} \cdots \bar{\mobserv}^{\vpp}_{\nu,\Nregions} \, ,~~
    \end{equation}
    where each operator $\bar{\mobserv}^{\,}_{\nu,s}$ corresponds to a particular ``measurement region,'' and may represent a product of the involutory parts \eqref{eq:involutory part} of two or more measured observables attached to $\Pauli{\nu}{f}$ via Eq.~\ref{eq:meas attach} of Lemma~\ref{lem:measure = attach}.

    Using the above, we rewrite Eq.~\ref{eq:k=1 small string order starting condish} as
    \begin{align}
        \Pauli{\nu}{i}  \, \StabEl^{\vpp}_{\nu}   &= \left( \uchan^{\dagger} \, \bar{\mobserv}^{\vpp}_{\nu,1} \, \uchan \right) \left( \uchan^{\dagger} \,  \bar{\mobserv}^{\vpp}_{\nu,2} \cdots \bar{\mobserv}^{\vpp}_{\nu,\Nregions} \, \Pauli{\nu}{f} \, \uchan \right) \, ,~ 
        \label{eq:k=1 small string order suggestive starting condish}
    \end{align}
    and since only the leftmost term on the right-hand side is causally connected to $\Pauli{\nu}{i}$ via $\uchan$, we define
    \begin{equation}
        \label{eq:k=1 small string order left Sigma1 def}
        \Pauli{\nu}{i} \, \Sigma^{\nu}_{1} = \uchan^{\dagger} \, \bar{\mobserv}^{\vpp}_{\nu,1} \, \uchan \, ,~~
    \end{equation}
    where $\Sigma^{\nu}_{1}$ acts nontrivially only to the right of site $i$, when the above is applied to the initial state $\ket{\Psi^{\,}_0}$ \eqref{eq:teleport initial mb state}. Since the right-hand sides of Eq.~\ref{eq:k=1 small string order left Sigma1 def} for $\nu=x,z$ commute with one another, the left-hand sides must as well. This implies that, e.g., $\Sigma^{x}_{1}$  and $\Sigma^{z}_{1}$ reproduce the Pauli algebra to the right of $i$, up stabilizer elements of $\UStabOf{\ket{\Psi^{\,}_0}}$.

    The string order parameter \eqref{eq:k=1 big string order} from Lemma~\ref{lem:k=1 big string order} can be written as the following element of $\UStabOf{\ket{\Psi^{\,}_t}}$,
    \begin{equation*}
        \tag{\ref{eq:k=1 big string resource stab nice}}
         \StabEl^{\vpp}_{\nu,t} = \uchan \, \Pauli{\nu}{i} \, \uchan^{\dagger} \, \bar{\mobserv}^{\vpp}_{\nu} \, \Pauli{\nu}{f} \, , ~~
    \end{equation*}
    which we next massage using Eq.~\ref{eq:k=1 small string order left Sigma1 def} to find 
    \begin{align}
        \StabEl^{\vpp}_{\nu,t} &= \uchan \, \Pauli{\nu}{i} \, \uchan^{\dagger} \, \bar{\mobserv}^{\vpp}_{\nu,1} \, \left( \uchan \, \uchan^{\dagger} \right) \,\bar{\mobserv}^{\vpp}_{\nu,2} \cdots \bar{\mobserv}^{\vpp}_{\nu,\Nregions} \,  \Pauli{\nu}{f} \notag \\
        &= \uchan \, \Pauli{\nu}{i} \, \left( \uchan^{\dagger} \, \bar{\mobserv}^{\vpp}_{\nu,1} \, \uchan \right) \uchan^{\dagger} \,\bar{\mobserv}^{\vpp}_{\nu,2} \cdots \bar{\mobserv}^{\vpp}_{\nu,\Nregions} \,  \Pauli{\nu}{f} \notag \\
        &= \uchan \, \Pauli{\nu}{i} \, \left( \Pauli{\nu}{i} \, \Sigma^{\nu}_1 \right) \uchan^{\dagger} \, \bar{\mobserv}^{\vpp}_{\nu,2} \cdots \bar{\mobserv}^{\vpp}_{\nu,\Nregions} \,  \Pauli{\nu}{f} \notag \\
        &= \left( \uchan \,\Sigma^{\nu}_1 \, \uchan^{\dagger}  \right) \,  \bar{\mobserv}^{\vpp}_{\nu,2} \cdots \bar{\mobserv}^{\vpp}_{\nu,\Nregions} \,  \Pauli{\nu}{f} \,,~~
        \label{eq:k=1 small string order first new left endpoint}
    \end{align}
    which has the same right endpoint operator $\Pauli{\nu}{f}$ as in Eq.~\ref{eq:k=1 big string order}, but the left endpoint operator has changed from
    \begin{equation}
        \label{eq:k=1 small string order first new left endpoint explicit}
        \uchan \, \Pauli{\nu}{i} \, \uchan^{\dagger} \to \uchan \,  \Sigma^{\nu}_1 \, \uchan^{\dagger} \, ,~~
    \end{equation}
    where the support of $\Sigma^{\nu}_1$ is entirely to the right of $\Pauli{\nu}{i}$, and thus constitutes an equivalent string order parameter $\StabEl^{\,}_{\nu,t}$ \eqref{eq:k=1 small string order first new left endpoint} of the same form as Eq.~\ref{eq:k=1 big string order} with a \emph{new} left endpoint, where the left-endpoint operators for $\nu=x,z$ are guaranteed to anticommute by Eq.~\ref{eq:k=1 small string order left Sigma1 def}. 

    Importantly, we observe that the new string-order parameter in Eq.~\ref{eq:k=1 small string order first new left endpoint} has the same important properties as the na\"ive counterpart \eqref{eq:k=1 big string order}. Both objects have (\emph{i}) unit expectation value in the resource state $\ket{\Psi^{\,}_t}$ \eqref{eq:resource state}; (\emph{ii}) a right-endpoint operator given by $\Pauli{\nu}{f}$; (\emph{iii}) a left-endpoint operator of the form $\uchan \, \tau^{\nu} \, \uchan^{\dagger}$, where the objects $\tau^{\nu}$ satisfy the Pauli algebra; and (\emph{iv}) all of the operators between the two endpoint commute for different values of $\nu$. 

    There are two key features that facilitate iterating this procedure to recover new endpoints. The first is that the left-endpoint operators \eqref{eq:k=1 small string order first new left endpoint explicit} are of the same form as the original Pauli operator in Eq.~\ref{eq:k=1 big string order}. The second is that there exists a ``leftmost'' choice of $\Sigma^{\nu}_1$ \eqref{eq:k=1 small string order left Sigma1 def} such that the clustered observables $\bar{\mobserv}^{\,}_{\nu,3}$ through $\bar{\mobserv}^{\,}_{\nu,\Nregions}$ are not connected to $\Sigma^{\nu}_1$ under $\uchan$. Essentially, a set of equivalent operators recovers upon multiplying $\Sigma^{\nu}_1$ \eqref{eq:k=1 small string order left Sigma1 def} by stabilizer elements $\StabEl^{\,}_0$ of $\UStabOf{\ket{\Psi^{\,}_0}}$; the ``leftmost'' left-endpoint operator \eqref{eq:k=1 small string order first new left endpoint explicit} is the operator $\Sigma^{\nu}_1 \, \StabEl^{\,}_0$ for $\StabEl^{\,}_0 \in \UStabOf{\ket{\Psi^{\,}_0}}$ whose support is closest to---but disjoint from---the initial site $i$ (i.e., whose right weight is closest to $i$).

    We find the next left endpoint by defining
    \begin{equation}
    \label{eq:k=1 small string order left Sigma2 def}
        \Sigma^{\nu}_1 \, \Sigma^{\nu}_2 = \uchan^\dagger \, \bar{\mobserv}^{\vpp}_{\nu,2} \, \uchan \, , ~~
    \end{equation}
    in analogy to Eq.~\ref{eq:k=1 small string order left Sigma1 def}. The procedure then repeats as before, where in the general case we define
    \begin{equation}
    \label{eq:k=1 small string order left Sigma any def 1}
        \Pauli{\nu}{i} \, \Sigma^\nu_s = \uchan^\dagger \, \bar{\mobserv}^{\vpp}_{\nu,1} \cdots \bar{\mobserv}^{\vpp}_{\nu,s} \, \uchan \, ,~~
    \end{equation}
    which is equivalent to the familiar form
    \begin{equation}
    \label{eq:k=1 small string order left Sigma any def 2}
        \Sigma^{\nu}_{s-1} \, \Sigma^{\nu}_s = \uchan^\dagger \, \bar{\mobserv}^{\vpp}_{\nu,s} \, \uchan \, , ~
    \end{equation}
    where the above two definitions reproduce Eq.~\ref{eq:k=1 small string Sigma def}. 

    One can iterate this procedure until all bulk operators except $\bar{\mobserv}^{\,}_{\nu,\Nregions}$ have been exhausted. Choosing the leftmost form of each $\Sigma^{\nu}_s$ guarantees that each $\Sigma^{\nu}_{s}$ only has support to the right of $\Sigma^{\nu}_{s-1}$, and thus constitutes a genuinely new endpoint. This choice only affects the total number of measurement ``regions'' $\Nregions$, which delineate endpoints. 
    
    We then find a set of string order parameters
    \begin{align}
        \StabEl^{\vpp}_{\nu,t,s} &= \left( \uchan \, \Sigma^{\nu}_s \, \uchan^\dagger \right) \, \bar{\mobserv}^{\vpp}_{\nu,s+1} \cdots \bar{\mobserv}^{\vpp}_{\nu,\Nregions} \, \left( \Pauli{\nu}{f} \right) \label{eq:k=1 small string order left}
    \end{align}
    for all $s \in [1,\Nregions]$, such that $\matel{\Psi^{\,}_t}{\StabEl^{\,}_{\nu,t,s}}{\Psi^{\,}_t} = 1$ \eqref{eq:k=1 string order} indicates perfect string order of $\ket{\Psi^{\,}_t}$ \eqref{eq:resource state}.
    
    We next identify new \emph{right} endpoints. Returning to Eq.~\ref{eq:k=1 small string order suggestive starting condish}, using the same operators $\Sigma^{\nu}_{s}$ defined in Eqs.~\ref{eq:k=1 small string order left Sigma any def 1} and \ref{eq:k=1 small string order left Sigma any def 2}. Importantly, we take $\Sigma^{\nu}_{s}$ to be the operator $\Sigma^{\nu}_{s} \, \StabEl^{\,}_0$ for $\StabEl^{\,}_0 \in \UStabOf{\ket{\Psi^{\,}_0}}$ with the leftmost right weight (i.e., whose support is farthest to the left). All such operators are equivalent, since these quantities are (effectively) evaluated in the initial state $\ket{\Psi^{\,}_0}$ \eqref{eq:teleport initial mb state}, so $\StabEl^{\,}_0 \cong \ident$.
    
    For any $s \in [1,\Nregions]$, we define
    \begin{align}
        \StabEl^{\prime}_{\nu,t,s} &\equiv \StabEl^{\vpp}_{\nu,t,0} \, \StabEl^{\dagger}_{\nu,t,s} \notag \\
        &= \left( \uchan \, \Pauli{\nu}{i} \, \uchan^{\dagger} \, \bar{\mobserv}^{\vpp}_{\nu,1} \cdots \bar{\mobserv}^{\vpp}_{\nu,\Nregions} \, \Pauli{\nu}{f} \right) \notag \\
        ~~&\times \left( \Pauli{\nu}{f} \, \bar{\mobserv}^{\vpp}_{\nu,\Nregions} \cdots \bar{\mobserv}^{\vpp}_{\nu,s+1} \, \uchan \, \Sigma^{\nu}_{s} \, \uchan^{\dagger} \right) \notag \\
        &= \left( \uchan \, \Pauli{\nu}{i} \, \uchan^{\dagger} \right) \, \bar{\mobserv}^{\vpp}_{\nu,1} \cdots \bar{\mobserv}^{\vpp}_{\nu,s} \, \left( \uchan \, \Sigma^{\nu}_s \, \uchan^{\dagger} \right) \, \label{eq:k=1 small string order right}  \, ,~~
    \end{align}
    which has the same left-endpoint operator $\uchan \, \Pauli{\nu}{i} \, \uchan^{\dagger}$ as $\StabEl^{\,}_{\nu,t,0}$ \eqref{eq:k=1 big string order} but distinct right-endpoint operators $\Sigma^{\nu}_s$, which are the same endpoint operators in Eq.~\ref{eq:k=1 small string order}. 
    
    So far, we have shown that new left and right endpoints can independently be chosen in Eqs.~\ref{eq:k=1 small string order left} and \ref{eq:k=1 small string order right}, respectively. Finally, we recover the generic result stated in the Lemma with arbitrary left endpoints via
    \begin{align}
        \mathcal{S}^{\vpp}_{\nu,a,b} &\equiv \StabEl^{\vpp}_{\nu,t,a} \, \StabEl^{\dagger}_{\nu,t,b} \\
        &= \left( \uchan \, \Sigma^{\nu}_a \, \uchan^{\dagger} \, \bar{\mobserv}^{\vpp}_{\nu,a+1} \cdots \bar{\mobserv}^{\vpp}_{\nu,\Nregions} \, \Pauli{\nu}{f} \right) \\
        ~~&\times \left( \Pauli{\nu}{f} \, \bar{\mobserv}^{\vpp}_{\nu,\Nregions} \cdots \bar{\mobserv}^{\vpp}_{\nu,b+1} \, \uchan \, \Sigma^{\nu}_{b} \, \uchan^{\dagger} \right) \\
        &= \left( \uchan \, \Sigma^{\nu}_a \, \uchan^{\dagger} \right) \, \bar{\mobserv}^{\vpp}_{\nu,a+1} \cdots \bar{\mobserv}^{\vpp}_{\nu,b} \, \left( \uchan \, \Sigma^{\nu}_b \, \uchan^{\dagger} \right) \, \tag{\ref{eq:k=1 small string order}}  \, ,~~
    \end{align}
    and since the support of $\Sigma^{\nu}_{s+1}$ is entirely to the right of $\Sigma^{\nu}_{s}$,  the support of $\Sigma^{\nu}_{s}$ is to the left of $\Sigma^{\nu}_{s+1}$. Hence, the operators $\Sigma^{\nu}_s$ constitute distinct left \emph{and} right endpoint operators for each $s \in [1,\Nregions]$, where $\matel*{\Psi^{\,}_t}{\mathcal{S}^{\,}_{\nu,a,b}}{\Psi^{\,}_t} = 1$ \eqref{eq:string order def} indicates nontrivial string order of $\ket{\Psi^{\,}_t}$ \eqref{eq:resource state}. 
\end{proof}

\end{document}